\documentclass{ectj}

\usepackage{amsmath,amsfonts,amssymb,graphics,epsfig,verbatim,bm,latexsym,url,amsbsy}
\usepackage{siunitx}

\usepackage{amsthm}
\usepackage{xr}
\usepackage{mathabx}

\usepackage{xr}
\usepackage{natbib}
\usepackage{graphicx}
\usepackage{lscape}
\usepackage{subcaption}
\usepackage{color}
\usepackage{wrapfig}
\usepackage{mathtools}
\usepackage{placeins}
\usepackage{tgpagella}
\usepackage[framemethod=TikZ]{mdframed}
\usepackage{xcolor}
\usepackage{tikz}
\usetikzlibrary {positioning}

\usepackage{bm}
\usepackage{paralist}
\usepackage{accents}
\usepackage{enumerate}
\renewcommand{\Pr}{{\mathrm{P}}}

\newcommand{\mcL}{{\mathcal L}}
\newcommand{\Ep}{\mathbb{E}}

\newcommand{\Enk}{{\mathbb{E}_{n,k}}}
\newcommand{\Gnk}{{\mathbb{G}_{n,k}}}

\newcommand{\EN}{{\mathbb{E}_{N}}}

\newcommand{\GN}{\mathbb{G}_N}

\newcommand{\E}{\mathbb{E}}
\usepackage{nccmath}
		{\end{pmatrix}\end{medsize}}%

\newtheorem{theorem}{Theorem}
\newtheorem{assumption}{Assumption}

\newtheorem{corollary}{Corollary}
\newtheorem{lemma}{Lemma}

\theoremstyle{remark}
\newtheorem{example}{Example}
\newtheorem{remark}{Remark}
\newtheorem{definition}{Definition}[section]
\renewcommand{\thesection}{\arabic{section}}
\renewcommand{\theequation}{\arabic{section}.\arabic{equation}}

\renewcommand{\theassumption}{\arabic{section}.\arabic{assumption}}

\renewcommand{\thelemma}{\arabic{section}.\arabic{lemma}}
\renewcommand{\theexample}{\arabic{section}.\arabic{example}}

\DeclareMathOperator{\eig}{eig}

\received{...}
\accepted{...}
\volume{21}

\title[Debiased Machine Learning]{Debiased Machine Learning of Conditional Average Treatment Effects and Other Causal Functions}

\author[Semenova and Chernozhukov ]{Vira Semenova$^{\dagger}$ and
               Victor Chernozhukov$^{\ddagger}$}

\email{$^{\dagger}$semenovavira@gmail.com}

\email{$^{\ddagger}$vchern@mit.edu}

\def\AmSTeX{$\cal A$\kern-.1667em\lower.5ex\hbox{$\cal M$}\kern-.125em
    $\cal S$-\TeX}
\def\BibTeX{{\rm B\kern-.05em{\sc i\kern-.025em b}\kern-.08em
    T\kern-.1667em\lower.7ex\hbox{E}\kern-.125emX}}

\begin{document}

     \begin{abstract} 
   This paper provides estimation and inference methods for the best linear predictor (approximation) of a structural function, such as conditional average structural and treatment effects, and structural derivatives, based on modern machine learning (ML) tools. We represent this structural function as a conditional expectation of an unbiased signal that depends on a nuisance parameter, which we estimate by modern machine learning techniques. We first adjust the signal to make it insensitive (Neyman-orthogonal) with respect to the first-stage regularization bias. We then project the signal onto a set of basis functions, which grows with sample size, to get the best linear predictor of the structural function. We derive a complete set of results for estimation and simultaneous inference on all parameters of the best linear predictor,  conducting inference by Gaussian bootstrap. When the structural function is smooth and the basis is sufficiently rich, our estimation and inference results automatically targets this function. When basis functions are group indicators, the best linear predictor reduces to the group average treatment/structural effect, and our inference automatically targets these parameters. We demonstrate our method by  estimating uniform confidence bands for the average price elasticity of gasoline demand conditional on income.


\end{abstract}

\section{Introduction and Motivation.}
\label{sec:model}

This paper gives a method for estimating and conducting inference on a nonparametric function $g(x)$ that summarizes heterogeneous treatment/causal/structural effects conditional on a small set of covariates $X$.  We assume that this function can be represented as a conditional expectation function
\begin{align}
\label{eq:main}
g(x) &= \Ep [Y(\eta_0) | X =x],
\end{align}
where the random variable $Y(\eta_0)$, which we refer to as a signal, depends on a nuisance function $\eta_0(z)$ of a (potentially very) high dimensional control vector $Z$. Examples of the nonparametric target function include the conditional average treatment effect (CATE), continuous treatment effects (CTEs), as well as many others discussed below. Examples of the nuisance functions $$\eta_0 =\eta_0(z)$$ include the propensity score, conditional density, and the regression function, among others. In summary, 
\begin{align*}
\text{dim} (Z) \text{ is high; } \quad \text{dim} (X) \text{ is low. }
\end{align*}

Although there are many possible choices of signals $Y(\eta)$, we focus on signals that have the orthogonality property  (\cite{Neyman:1959}). Formally, we require the pathwise derivative of the conditional expectation to be zero conditional on $X$:
\begin{align}
\label{eq:orhog}
\partial_{r} \Ep [Y(\eta_0 +r(\eta-\eta_0)) | X=x ] |_{r=0} = 0,  \quad \text{ for all }  x  \text{ and } \eta.
\end{align}
If the signal $Y(\eta)$ is orthogonal, its plug-in estimate $Y(\widehat{\eta})$ is insensitive to bias in the estimation of $\widehat{\eta}$, which results from applying modern adaptive learning methods in high dimensions. Under mild conditions, $Y(\widehat{\eta})$ delivers a high-quality estimator of the target function $g(x)$.

We demonstrate  the importance of the orthogonality property for continuous treatment effects, studied in \cite{Imbens2000}, \cite{GillRobins}, and \cite{Kennedy}. Let $X \in \mathbb{R}$ be a one-dimensional continuous treatment, and  $Y^{x}$ be the potential outcome corresponding to the subject's response after receiving $x$ units of treatment. The observed data vector $V=(X,Z,Y)$ consists of the treatment $X$, the control vector $Z$, and the observed outcome $Y=Y^{X}$. If potential outcomes $\{Y^{x}, x \in \mathbb{R}\}$ are independent of treatment $X$ conditional on controls $Z$, the 
average potential outcome is identified as
\begin{align}
\label{eq:moment}
\E[Y^{x}] = \E \mu_0(x,Z) = \int \mu_0(x,z) d P_{Z}(z),
\end{align}
where $\mu_0(x,z)=\E[Y|X=x,Z=z]$ is the regression function of the observed outcome. Since the control vector $Z$ is high-dimensional, it is necessary to estimate the regression function $\mu_0(x,z)$ with some regularized technique to achieve convergence.  

A naive approach to estimate $\E[Y^{x}]$ is to consider a sample average
\begin{align*}
   \widetilde g(x)  = \int \widehat{\mu}(x, z)  d \widehat P_{Z}(z), 
\end{align*}
where $\widehat{\mu}(x,z)$ is a regularized estimate of $\mu_0(x,z)$, and $\widehat P_{Z}$ is the empirical analog of $P_{Z}$. This approach results in a biased estimate, and the bias of estimation error, $\widehat{\mu}(x,Z)-\mu_0(x,Z)$,  does not vanish faster than $N^{-1/2}$. The plug-in estimator inherits this first-order bias,
since the moment equation \eqref{eq:moment} is not orthogonal to  perturbations of $\mu$:
\begin{align*}
    \partial_{r}\E [ (\mu_0 + r (\mu-\mu_0)) \circ (x,Z) ]|_{r=0}   = \E [\mu(x,Z)-\mu_0(x,Z)] \neq 0.
\end{align*}
This bias  implies that the plug-in estimator $\widetilde g$ won't converge at the optimal rate.

To deliver a high-quality estimate of $\E[Y^{x}]$, we represent $$g(x)=\E[Y^{x}]$$ as a special case of the signal framework \eqref{eq:main}. We choose the signal $Y(\eta)$ to be the doubly robust signal from \cite{Kennedy}:
\begin{align}
\label{eq:signal}
Y (\eta) &:=\dfrac{Y - \mu(X,Z)}{s(X|Z)} w (X) + \int \mu(X,z) d P_{Z}(z), 
\end{align}
where the nuisance parameter $$\eta_0(x,z) = \{ s_0(x|z), \mu_0(x,z),  w_0(x) \}$$ consists of  the regression function $\mu_0(x,z)$,
the conditional density $s_0(x|z)$ of $X$ given $Z$, and the  marginal treatment density $w_0(x)$. This procedure is more costly because the nuisance parameter now includes two more functions $s_0(x|z)$ and $w_0(x)$, in addition to   $\mu_0(x,z)$. However, the signal in \eqref{eq:signal} has the benefit of being conditionally orthogonal with respect to each nuisance function in $\eta_0(x,z)$:
\begin{align*}
\Ep  \left [ \begin{array}{ll}  &    -  \int_{z \in \mathcal{Z}} (\mu(X,z) - \mu_0(X,z)) dP_Z(z) + \int_{z \in \mathcal{Z}} (\mu(x,z) - \mu_0(x,z)) dP_Z(z)
 \\   &  \frac{ \mu_0(X,Z)-Y}{s^2_0(X|Z)}     (s(X|Z)-s_0(X|Z))   
  \\  &  \frac{Y - \mu_0(X,Z)}{s_0(X|Z)} (w(X) - w_0(X))  \\
  \end{array}  \Bigg| X \right ] =0.
  \end{align*}
Because this signal is conditionally orthogonal to the nuisance function, the bias of the estimation error, $\widehat{\eta}(X,Z) - \eta_0(X,Z)$,  does not create first-order bias in the estimated signal $Y(\widehat{\eta})$ and only affects  its higher-order bias.  As a result, the estimate of the target function based on $Y(\widehat{\eta})$ is high-quality under plausible conditions.

In the second stage we consider a linear projection of an orthogonal signal $Y(\eta)$ onto a vector of basis functions $p(X)$,
$$ \beta := \arg \min_{b \in \mathbb{R}^d} \E (Y(\eta) - p (X)'b)^2.$$

The choice of basis functions depends on the desired interpretation of the linear approximation. For example, 
consider partitioning  the support of $X$ into $d$ mutually exclusive groups $\{G_k\}_{k=1}^d$. Setting   $$p_k(x)=\mathbf{1}\{x  \in G_k \}, \quad k \in \{1,2, \dots, d\}$$  implies
that $p(x)'\beta_0$ is a Group Average Treatment Effect (GATE) 
for group $k$ such that $x \in G_k$. Our inference will target this parameter, allowing the
number of groups to increase at some rate.

For another example, let $p(x) \in \mathbb{R}^d$ be a $d$-dimensional dictionary of series/sieve basis functions, e.g., polynomials, splines, or wavelets. Then, $p(x)'\beta_0$ corresponds to the best linear approximation to the target function $g(x)$ in the given dictionary.  Under some smoothness conditions, as the dimension of the dictionary becomes large  $p(x)'\beta_0$ will  approximate $g(x)$, and our inference will target this function.  We derive a complete set of results for estimation and simultaneous inference on all parameters of the best linear predictor, conducting inference by Gaussian bootstrap. When the structural function is smooth and the basis is sufficiently rich, our estimation and inference results automatically target this function. When basis functions are group indicators, the best linear predictor reduces to the group average treatment/structural effect, and our inference automatically targets these parameters.

\subsection{Literature Review.}

This paper builds on three bodies of research within the semiparametric literature: orthogonal(debiased) machine learning, least squares series estimation, and treatment effects/missing data problems. Orthogonal machine learning (\cite{LRSP}, \cite{chernozhukov2016double}) proposes  inference on a fixed-dimensional target parameter $\beta_0$ in the presence of a high-dimensional nuisance function $\eta$ in a semiparametric moment problem. If the moment condition is orthogonal to  perturbations of $\eta$, estimating $\eta$ by ML methods has no first-order effect on the asymptotic distribution of the target parameter $\beta_0$. In particular, plugging in an estimate of $\eta$ obtained on a separate sample results in a $\sqrt{N}$-consistent asymptotically normal estimate whose asymptotic variance is the same as if the econometrician knew $\eta=\eta_0$. This result makes it possible to use highly complex machine learning methods to estimate the nuisance function $\eta$, such as $\ell_1$ penalized methods in sparse models (\cite{vandergeer}, \cite{orthogStructural}), $\ell_2$ boosting  in sparse linear models (\cite{Luo}), and other methods for classes of neural nets, regression trees, and random forests. This paper extends the orthogonal machine learning literature by allowing the target parameter to be a function, that is, an infinite-dimensional parameter. Next, our paper contributes to a large body of work on ``debiased'' inference for parameters following regularization or model selection (\cite{Program}, \cite{BCH}, \cite{orthogStructural}, \cite{geer}, \cite{JM}, \cite{ZhangZhang}), with the crucial difference that our (ultimate) target parameter $g(x)$ is infinite-dimensional, whereas in those papers the target parameter is finite-dimensional.

The second building block of our paper is the literature on least squares series estimation (\cite{Newey97}, \cite{Newey09}, \cite{NewOls}, \cite{ChenChris}), which establishes pointwise and uniform limit theory for least squares series estimation. We extend this theory by allowing the dependent variable of the series projection to depend upon an unknown nuisance parameter $\eta$. We show that series properties continue to hold without  any additional strong assumptions on the problem design.

Finally, we also contribute to the  literature on estimating the conditional average treatment effects and group average treatment effects with missing data (\cite{Robins}, \cite{Hahn98}, \cite{Graham},   \cite{Graham2003}, \cite{HIR2003}, \cite{AbrevayaHsu},  \cite{AtheyImbens},  \cite{GrimmerMessing}, \cite{OprescuWu}) among others. After we released the working paper version of this article (\cite{CherSemWP}), many methods (\cite{Hardle}, \cite{Lieli},  \cite{ZimLech}, \cite{Colangelo}) have been proposed for estimating group, incremental or heterogeneous treatment effects in the presence of high-dimensional controls. Our framework covers many more examples than just Conditional Average Treatment Effect or Continuous Treatment Effects  and uses series estimators, as opposed to kernels, to localize the structural function.

In a related paper, \cite{chernozhukovdemirer} study Conditional Average Treatment Effects in randomized control trials with a known propensity score. Recognizing a widespread interest in estimating CATE by modern machine learning techniques, this paper studies the best linear projection of the true CATE function onto an arbitrary ML estimator of CATE  under consideration, constructed on an auxiliary sample. \cite{chernozhukovdemirer}'s analysis does not require  any assumptions on the ML estimator; however, that paper targets  a specific feature of CATE---the best linear projection of CATE---rather than CATE itself.  In contrast, our work operates in a classic observational setting with many potential controls and targets the true CATE function.  In our setting, modern regularized methods are used to estimate the propensity score, but are required to approximate this parameter  sufficiently well.   To sum up, this paper  delivers a sharper characterization of CATE in a more challenging setting, in  exchange for stronger assumptions about the first-stage machine learning estimate.

\section{Set-Up}

\subsection{Examples.}
\label{sec:examples}

In this section, we describe our main examples. For each example we provide an orthogonal signal $Y(\eta)$ obeying \eqref{eq:main}--\eqref{eq:orhog}.


\begin{example}[Continuous Treatment Effects]
\label{CATEcont}
Let $X \in \mathbb{R}$ be  a continuous treatment variable, $Z$ be a vector of the controls,   $Y^x$ stand for the potential outcomes corresponding to the subject's response after receiving $x$ units of treatment, and $Y=Y^{X}$ be the observed outcome. The observed data $V$ is $V=(X,Z,Y)$. For a given value $x$, the target function is the average potential outcome 
\begin{align}
 \label{eq:targetcont}
 g(x)=\E[Y^{x}].
 \end{align}
 
 A standard way to identify the function $g(x)$ is to assume unconfoundedness. Suppose all of the potential outcomes $\{ Y^{x}, x \in \mathbb{R} \}$ are independent of $X$ conditional on $Z$,
\begin{align}
    \label{eq:contoutcome:ident}
    \{ Y^{x}, x \in \mathbb{R} \} \perp X | Z.
\end{align}
Then, $g(x)$ is identified as
\begin{align*}
g(x) &=  \E \mu_0(x,Z),
\end{align*}
where $\mu_0(x,z)=\E[Y | X=x,Z=z]$ is the regression function. 
Lemma \ref{thrm:conttreat} shows that the doubly robust signal from \cite{Kennedy},
\begin{align}
\label{eq:contoutcome}
Y (\eta) &:=\dfrac{Y - \mu(X,Z)}{s(X|Z)} w (X) + \int \mu(X,z) d P_{Z}(z) ,
\end{align}
is conditionally orthogonal with respect to the nuisance parameter $$\eta_0(x,z):= \{s_0(x|z),  \mu_0(x,z),  w_0(x) \}$$ consisting of the conditional treatment density (a.k.a. generalized propensity score)  $$s_0(x|z) = \dfrac{d P(X \leq t| Z=z)}{dt}\bigg|_{t=x},$$  regression function $\mu_0(x,z) = \E [Y|X=x,Z=z],$ and the marginal treatment density $$ w_0(x)=\dfrac{d \Pr (X \leq t)}{dt}\bigg|_{t=x}=\E_{Z} s_0(x|Z).$$ Theorems \ref{thrm:cte:pointwise} and \ref{thrm:cte:uniform} establish pointwise and uniform asymptotic normality for the Orthogonal Estimator of   Continuous Treatment Effects.

\end{example}

In the examples below, the vector $X$ represents a low-dimensional subset of potential controls $Z$.  

\begin{example}[Conditional Average Treatment Effect]
\label{CATE}
Let $ Y_1$ and $Y_0$ be the potential outcomes corresponding to a subject's response with and without receiving a binary treatment, respectively. Let $D=1$  be a dummy for whether a subject is treated.  The object of interest is the Conditional  Average Treatment Effect $$g(x):=\Ep [Y_1 - Y_0 | X=x].$$ Since an individual cannot be treated and untreated at the same time, the econometrician only observes the actual outcome  $Y= DY_1 + (1-D)Y_0$, but not the treatment effect  $Y_1-Y_0$. 

A standard way to make progress in this problem is to assume unconfoundedness (\cite{Rosenbaum}). Suppose there exists an observable control vector $Z$ such that treatment status $D$ is independent of the potential outcomes $Y_1,Y_0$  conditional on $Z$,
\begin{align}
    \label{eq:MAR}
    Y_1, Y_0 \perp D | Z.
\end{align}
Define the conditional probability  of   treatment receipt as  $s_0(z) = \Pr (D=1|Z=z)$.  Consider a \cite{Robins} type  orthogonal signal $Y(\eta)$,
\begin{align}
\label{eq:catesig}
Y (\eta) &:= \mu (1,Z) - \mu (0,Z)   + \dfrac{D [Y - \mu (1,Z) ]}{s (Z)} - \dfrac{(1-D)[Y- \mu (0,Z)] }{1-s (Z)},
\end{align}
where $\mu_0(d,z) = \Ep [Y|D=d,Z=z]$ is the conditional expectation function of $Y$. Corollary \ref{thrm:CATE} shows that \eqref{eq:catesig} is  orthogonal with respect to the nuisance parameter $\eta_0(z):=\{ s_0(z),\mu_0(1,z),\mu_0(0,z) \}$ and establishes pointwise and uniform asymptotic theory for the Orthogonal Estimator of  Conditional Average Treatment Effect.
\end{example}

\begin{example}[Regression Function with Partially Missing Outcome]
\label{MD}
Suppose a researcher is interested in the conditional expectation of a variable $Y^{*}$ given $X$: $$g(x):= \Ep [Y^{*}| X=x],$$ where $Y^{*}$ is partially missing. Let $D=1$ be a dummy for whether the outcome $Y^{*}$ is observed, $Z$ be a control vector, $Y = DY^{*}$ be the observed outcome, and $V=(D,Z,Y)$  be the data vector.  Since the researcher does not control $D$, a standard way to make progress is to assume there exists an observable control vector $Z$ such that $Y^{*}$ is independent of $D$ given $Z$,
\begin{align*}
    Y^{*} \perp D | Z.
\end{align*}
Corollary \ref{thrm:MD} shows that the  signal  $Y(\eta)$, defined as
\begin{align}
\label{eq:mdsig}
Y(\eta) &:= \mu (Z)+ \dfrac{D[Y- \mu (Z)]}{s (Z)},   
\end{align}
 is  orthogonal  with respect to the nuisance parameter $$\eta_0(z):=\{ s_0(z), \mu_0(z) \},$$  where $ \mu_0(z) =  \Ep [Y|Z=z, D=1]$ is the conditional expectation function of the observed outcome $Y$.

\end{example}

\begin{example}[Conditional Average Partial Derivative]
\label{CAPD}

Let $D\in \mathbb{R}$ be  a continuous treatment variable, $Z$ be a vector of the controls,   $Y^d$ stand for the potential outcomes corresponding to the subject's response after receiving $d$ units of treatment, $Y=Y^{D}$ be the observed outcome, and  $V=(D,Z,Y)$ be the data vector. Let $X$ be a subvector of controls $Z$. The target function is the average partial derivative conditional on a covariate vector $X$,
\begin{align}
 \label{eq:targetincrement}
 g(x)= \partial_{d} \E [Y^{D} | X=x].
 \end{align}
A standard way to identify the function $g(x)$ is to assume unconfoundedness. Suppose the potential outcome $Y^{d}$ is independent of $D$ conditional on $Z$,
\begin{align}
    \label{eq:contoutcome:ident2}
    \{ Y^{d}, d \in \mathbb{R} \} \perp D | Z.
\end{align}
Then, $g(x)$ is identified as
\begin{align*}
g(x) &= \E [\partial_{d} \mu_0(D,Z)  | X=x],
\end{align*}
where $\mu_0(d,z)=\E[Y | D=d,Z=z]$ is the regression function.  Corollary \ref{thrm:CAPD} shows that the signal 
\begin{align}
\label{eq:capdsig}
	Y(\eta):= -\partial_{d}\log s(D|Z) [ Y - \mu(D,Z)] + \partial_{d} \mu(D,Z),
\end{align}
 is  orthogonal  with respect to the nuisance parameter $$\eta_0(d,z) = \{ \mu_0(d,z),  s_0(d|z)\},$$ where $s_0(d|z) $ is the conditional density of $D$ given $Z$.

\end{example}

\subsection{Overview of Main Results.}
The first main contribution of this paper is to provide sufficient conditions for pointwise and uniform asymptotic Gaussian approximation of the target function. We approximate the target function $g(x)$ by a linear form $p(x)'\beta_0$:   $$ g(x) = p(x)'\beta_0+ r_{g} (x), $$ 
where $p(x)$ is a $d$-vector of basis functions of $x$,  $r_{g}(x)$ is the linear approximation error, and $\beta_0$ is the Best Linear Predictor/Approximation parameter, defined by the normal equation
\begin{align}
\label{eq:BLP}
\Ep p(X)[g(X) - p(X)'\beta_0] = \Ep p(X)r_g(X) = 0.
\end{align}

We construct the \emph{Orthogonal Estimator} $\widehat{\beta}$, the two-stage estimator of  $\beta_0$, as follows. In the first stage, we construct an estimate $\widehat{\eta}$ of the nuisance parameter $\eta_0$, using a high-quality machine learning estimator capable of dealing with the high-dimensional covariate vector $Z$. In the second stage we construct an estimate $\widehat{Y}_i$ of the signal $Y_i$ as $\widehat{Y}_i:=Y_i (\widehat{\eta})$ and run ordinary least squares of $\widehat{Y}_i$ on $p(X_i)$. We use different samples to estimate $\eta$ in the first stage and $\beta_0$ in the second stage in a form of cross-fitting. 
\begin{definition}[Cross-Fitting]
\mbox{}
 	\label{sampling} 	\begin{compactenum} 
	\item  For a random sample of size $N$, denote a $K$-fold random partition of the sample indices $[N]=\{1,2,...,N\}$ by $(J_k)_{k=1}^K$, where $K$ is  the number of partitions and the sample size of each fold is $n = N/K$. For each $k \in [K] = \{1,2,...,K\}$ define $J_k^c = \{1,2,...,N\} \setminus J_k$.
	\item For each $k \in [K]$, construct 	an estimator $\widehat{\eta}_k = \widehat{\eta}( V_{i \in J_k^c})$   of the nuisance parameter $\eta_0$ using only the data $\{ V_{j}: j \in J_k^c \}$. For any observation $i \in J_k$, define $\widehat{Y}_i := Y_i(\widehat{\eta}_k) $.
		\end{compactenum}
\end{definition}

\begin{definition}[Orthogonal  Estimator]
	\label{def:OS}
 	Given  $(\widehat{Y}_i)_{i=1}^N$, define
	\begin{align}
	\label{eq:OLS}
	\widehat{\beta} &:=\left( \dfrac{1}{N} \sum_{i=1}^N p(X_i) p(X_i)' \right)^{-1} \dfrac{1}{N} \sum_{i=1}^N  p(X_i)\widehat{Y}_i.  	\end{align}
	
\end{definition}

Under  mild conditions on $\eta$,  the Orthogonal Estimator delivers a high-quality estimate $p(x)'\widehat{\beta}$ of the pseudo-target function $p(x)'\beta_0$ with the following properties:

\begin{enumerate}

\item With probability (w.p.) $\rightarrow 1$, the mean squared error of $p(x)'\widehat{\beta} $ is bounded by $$ \bigg(\dfrac{1}{N} \sum_{i=1}^N (p(X_i)' (\widehat{\beta} - \beta_0))^2 \bigg)^{1/2}  = O_{P} \bigg(\sqrt{\dfrac{d}{N}} \bigg) .$$
\item The estimator $p(x)'\widehat{\beta}$ of the pseudo-target function $p(x)'\beta_0$ is asymptotically linear:
$$\sqrt{N} \dfrac{ p(x)' (\widehat{\beta} - \beta_0)}{\sqrt{p(x)' \Omega p(x)}} = G_N(x) + o_{P} (1/\sqrt{\log N }),$$ where the empirical process $G_N(x)$ is approximated by a Gaussian process
$$
G(x) = \dfrac{ p(x)'}{{\sqrt{p(x)' \Omega p(x)}} } N(0, \Omega)
$$
uniformly over $x \in \mathcal{X}$, and the covariance matrix $\Omega$  can be consistently estimated by a sample analog $\widehat{\Omega}$. 

\item If the misspecification error $r_g(x)$ is small, the pseudo-target function $p(x)'\beta_0$ can be replaced by the target function $g(x)$:
$$\sqrt{N} \dfrac{ p(x)' \widehat{\beta} - g(x) }{\sqrt{p(x)' \Omega p(x)}} = G_N(x) + o_{P} (1/\sqrt{\log N }).$$ 

\item Simultaneous inference is facilitated by Gaussian bootstrap, which relies on simulating
the empirical Gaussian process:
$$
G^\star(x) = \dfrac{ p(x)'}{{\sqrt{p(x)' \Omega p(x)}} } N(0, \widehat{\Omega}).
$$
The quantiles of the suprema of $x \mapsto G(x)$ can be consistently approximated by simulation, which makes it possible to construct uniform confidence bands for $x \mapsto p(x)'\beta$ and $x \mapsto g(x)$.

\end{enumerate}

Our results accommodate high-dimensional/highly complex modern machine learning (ML) methods to estimate $\eta$, such as random forests, neural networks, and $\ell_1$-shrinkage estimators, as well as procedures that estimate $\widehat{\eta}$ by classic nonparametric methods. The only requirement we impose on the estimation of  $\widehat{\eta}$  is that it converges to the true nuisance parameter $\eta_0$ at a fast enough rate $o_P(N^{-1/4-\delta})$ for some $\delta \geq 0$. This requirement is satisfied under structural assumptions on $\eta_0$, such as approximate sparsity of $\eta_0$ with respect to some dictionary, or if $\eta_0$ is well approximated by trees or by sparse neural and deep neural nets. For the Conditional Average Treatment Effect,  it is straightforward to apply our method using the \url{best_linear_predictor} command in the $R$ \url{grf} package, available from \cite{grf}. 

\begin{figure}[ht]
\begin{center}
\begin {tikzpicture}[-latex, auto, node distance =2.5cm and 3cm, on grid, thick, 
  rect/.style ={rectangle, top color=blue, bottom color = white, draw, white, text=black , minimum width =1 cm, minimum height =1.5 cm},
  ml/.style ={circle, top color=yellow, bottom color = white, draw, black, text=black , minimum width =2 cm},
   data/.style ={rectangle, top color=blue, bottom color = white, draw, black, text=black , minimum width =2 cm, minimum height =1 cm},
    ols/.style ={rectangle, top color=white, bottom color = white, draw, white, text=black , minimum width =2 cm, minimum height =1.5 cm}]

\node[rect]   (Y1) {$Y(\widehat{\eta}_1)$};

\node[rect]   (Y2)[below =of Y1]  {$Y(\widehat{\eta}_2)$};

\node[rect]   (Y3) [below =of Y2]  {$Y(\widehat{\eta}_3)$};

\node[ml]   (ML1) [left =of Y1] { \begin{tabular}{c} ML1 \\ $\widehat{\eta}_1$ \end{tabular}};
\node[ml]   (ML2) [below  =of ML1] {\begin{tabular}{c} ML2 \\ $\widehat{\eta}_2$ \end{tabular}};
\node[ml]   (ML3) [below  =of ML2] {\begin{tabular}{c} ML3 \\ $\widehat{\eta}_3$ \end{tabular}};
\node[data] (Data1) [left =of ML1] {$\text{Data}_{1}$};
\node[data] (Data2) [left =of ML2] {$\text{Data}_{2}$};
\node[data] (Data3) [left =of ML3] {$\text{Data}_{3}$};
\node[ols] (OLS) [below =of ML3] {\begin{tabular}{c} OLS \\ $\widehat{\beta} $ \end{tabular}};
\node[ols] (Input) [left =of OLS] {\begin{tabular}{c} Input \\ $ x$ \end{tabular}};
\node[ols] (Output) [right =of OLS] {\begin{tabular}{c} Output \\ $p(x)'\widehat{\beta}$ \end{tabular}};


\path[draw, dotted] (Data1) edge  (ML2);
\path[draw, dotted] (Data1) edge  (ML3);
\path[draw, dotted] (Data2) edge  (ML1);
\path[draw, dotted] (Data2) edge  (ML3);
\path[draw, dotted] (Data3) edge  (ML1);
\path[draw, dotted] (Data3) edge  (ML2);
\path[draw, dotted] (ML1) edge  (Y1);
\path[draw, dotted] (ML2) edge  (Y2);
\path[draw, dotted] (ML3) edge  (Y3);
\path[draw] (Input) edge  (OLS);
\path[draw] (OLS) edge  (Output);

\path [draw, dotted] (Data1.east) edge[bend  right=270] (Y1.west) {};
\path [draw, dotted] (Data3.east) edge[bend  right=90] (Y3.west) {};
\path [draw, dotted] (Data2.east) edge[bend  right=90] (Y2.west) {};

\draw[draw, dashdotted] (-1.5,-6) rectangle (1.5,1);
\draw[draw, dashdotted] (-7.5,-6) rectangle (-4.5,1);
\draw[draw, dashdotted] (-4.5,-6) edge (-3,-7.1);
\draw[draw, dashdotted] (-1.5,-6) edge (-3,-7.1);

\end{tikzpicture} \end{center}
\caption{Graphical representation of  the Orthogonal Estimator (OE) with cross-fitting. \emph{First Stage}. The rectangles represent the   partition of  the data into $K=3$ subsets. The circles represent $K=3$ instances of machine learning algorithm ML, whose training sets are indicated by straight arrows. For each partition $k \in \{1,2,3\}$, the signal $Y(\widehat{\eta}_k)$ is estimated using the $\text{Data}_k$ and ML instance $\widehat{\eta}_k$. \emph{Second Stage}. The O.L.S. estimator $\widehat{\beta}= (\sum_{i=1}^N p(X_i) p(X_i)')^{-1} \sum_{i=1}^N p(X_i) Y_i(\widehat{\eta})$ is estimated on the full data set with covariate vector $p(X)$ and outcome variable $Y(\widehat{\eta})$, where $p(X)$ is a vector of series terms. For a given point of interest $x$ (input), the function $g(x)$ is estimated by $p(x)'\widehat{\beta}$ (output).  }
\end{figure}

%
Define the covariance matrix of the basis functions as
$$ Q = \E p(X) p(X)', $$
and define its empirical analog as
$$ \widehat{Q} = \dfrac{1}{N} \sum_{i=1}^N p(X_i) p(X_i)'. $$
In Examples \ref{CATE}--\ref{CAPD},  the asymptotic covariance matrix of  Orthogonal Estimator is
$$\Omega = Q^{-1} \Ep p(X) p(X)' ( U + r_g(X ))^2  Q^{-1},$$
and its empirical analog is
\begin{equation} \label{eq:omegahat}\widehat{\Omega}:= \widehat{Q}^{-1} \EN p(X_i) p(X_i)' ( Y_i (\widehat{\eta}) - p(X_i)' \widehat{\beta} )^2  \widehat{Q}^{-1}. \end{equation}
In the case of Continuous Treatment Effects (Example \ref{CATEcont}), the asymptotic variance  contains an additional component that we describe in Section \ref{sec:apps}.

\begin{definition}[Pointwise and Uniform Confidence Bands]
    \label{def:bands}
Let $\widehat{g}(x) = p(x)'\widehat{\beta}$. Denote the $t$-statistic as
\begin{align}
\label{eq:tN}
t_N(x):= \dfrac{\widehat{g}(x) - g(x)}{\widehat{\sigma}_N(x)},
\end{align}
where $\widehat{\sigma}_N(x) = \sqrt{p(x)' \widehat{\Omega} p(x)/N}$, and denote the bootstrapped $t$-statistic as
\begin{align}
\label{eq:bootcritvalue}
\widehat{t}^{b}_N(x) := \dfrac{p(x)' \widehat{\Omega}^{1/2} /\sqrt{N} }{\widehat{\sigma}_N(x)}\mathcal {N}^b_{d}, 
\end{align}
where $\mathcal {N}^b_{d}$ is a bootstrap draw from $N(0, I_d)$. Define the confidence bands for $g(x)$ as
\begin{align}
\label{eq:bootconfbands}
[ \underline{i}(x), \overline{i}(x) ]  :=  [\widehat{g}(x) -c_N (1-\alpha)\widehat{\sigma}_N(x),  \quad \widehat{g}(x) +c_N (1-\alpha)\widehat{\sigma}_N(x) ], \quad x \in \mathcal{X},
\end{align}
where  the  critical value $c_N (1-\alpha)$ is the $(1-\alpha)$-quantile of $N(0,1)$ for the pointwise bands and  $c_N (1-\alpha)$ is the $(1-\alpha)$-quantile of $\sup_{x \in \mathcal{X}} |\widehat{t}^{b}_N(x)  | $ for the uniform bands. 

\end{definition}

In the case of Continuous Treatment Effects (Example \ref{CATEcont}), the orthogonal signal \eqref{eq:contoutcome} involves an auxiliary nuisance parameter---the expectation of $\mu(x,Z)$ with respect to $Z$ for each value of $x$. We estimate this parameter by a leave-one-out sample average
\begin{align}
\label{eq:feassignal}
Y^{\dagger}_i(\widehat{\eta}) =\dfrac{Y_i^o - \widehat{\mu} (X_i, Z_{i})}{ \widehat{s} (X_i | Z_{i})} \widehat{w}(X_i)+ \dfrac{1}{n-1} \sum_{j \in J_k, j \neq i} \widehat{\mu} (X_i, Z_{j}), \quad i \in J_k,
\end{align}
where  $\widehat{\eta} (x,z)= \widehat{\eta}_k (x,z)$ is estimated on $J_k^c$ for each $k$ and  the sample average in the second summand is taken over data $(V_j)_{j \in J_k}$ excluding observation $i$.  Having replaced $Y_i(\widehat{\eta})$ by $Y^{\dagger}_i(\widehat{\eta})$ in Definition \ref{def:OS}, we obtain an asymptotically linear estimator  $p(x)'\widehat{\beta}^{\dagger}$ of the pseudo-target function:
\begin{align}
\label{eq:ALintro}
\dfrac{\sqrt{N} p(x)' (\widehat{\beta}^{\dagger} - \beta_0) }{\sqrt{p(x)' \Omega^{\dagger} p(x)}} = G_N(x) + o_{P} (1/\sqrt{\log N }),
\end{align}
where $\Omega^{\dagger}$ is the asymptotic variance of  $\widehat{\beta}^{\dagger}$ and the empirical process $G_N(x)$ is approximated by a Gaussian process 
$$
G(x) = \dfrac{ p(x)'}{{\sqrt{p(x)' \Omega^{\dagger} p(x)}} } N(0, \Omega^{\dagger})
$$
uniformly over $x \in \mathcal{X}$.

\section{Main Theoretical Results.}

We use the empirical process notation. For a generic function $f$ and a generic sample $(V_i)_{i=1}^N$, denote the empirical sample average by $ \EN f(V_i) := \dfrac{1}{N}  \sum_{i=1}^N f(V_i) $ and the
scaled, demeaned sample average by  $$\GN f(V_i) := 1/\sqrt{N} \sum_{i=1}^N  [f(V_i) - \int f(v) d P(v)]$$ 

\label{sec:asymp}

The following assumptions impose regularity conditions on the covariate distribution, error terms, and the estimator of the nuisance parameter.

\begin{assumption}[Identification]
	\label{ass:identification}
	
	Let $Q := \Ep p(X) p(X)' =Q_d$ denote population covariance matrix of $p(X)$. Assume that there exist $0 < C_{\min} < C_{\max} < \infty$ that do not depend on $d$ s.t. $C_{\min} \leq \min \eig (Q) \leq \max \eig (Q) \leq C_{\max}$ for all $d$.
\end{assumption}
\begin{assumption}[Growth Condition]
	\label{ass:growth} We assume that the $\sup$-norm of the basis functions $\xi_d := \sup_{x \in \mathcal{X}} \| p(x) \| = \sup_{x \in \mathcal{X}}  (\sum_{j=1}^d p_j(x)^2)^{1/2} $ grows sufficiently slow:$$ \sqrt{\dfrac{\xi_d^2 \log N}{N}} = o(1).$$
\end{assumption}

\begin{assumption}[Misspecification Error] 
	\label{ass:approx}
There exists a sequence of finite constants ${l_d, r_d}$ such that the norms of the misspecification error are controlled as follows:
$$ \| r_{g} \|_{P,2} := \sqrt{\int r_{g}(x)^2  dP(x)} \lesssim r_d \text{ and }  \| r_{g} \|_{P,\infty} := \sup_{x \in \mathcal{X}} | r_{g}(x)| \lesssim l_dr_d.$$
\end{assumption}
Assumption \ref{ass:approx} introduces the rate of decay of the misspecification error. Specifically, the sequence of constants $r_d$  bounds the mean squared misspecification error. In addition, the sequence $l_d r_d$ bounds the worst-case misspecification error uniformly over the domain $\mathcal{X}$, where $l_d$ is the modulus of continuity of the worst-case error with respect to mean squared error.

Define the stochastic error $U$ as   $$U:=Y - g(X)$$ 
and the lower and upper bounds on its second moment conditional on $X$ as
\begin{align*}
   \underbar{$\sigma$}^2:=\inf_{x \in \mathcal{X}} \E [ U^2 | X= x], \quad \overline {\sigma}^2:=\sup_{x \in \mathcal{X}} \E [ U^2 | X= x].
\end{align*}

\begin{assumption}[Error Assumption]
\label{ass:boundederror}
The second moment of the sampling error $U$ conditional on $X$ is bounded from above: $\overline{\sigma}^2 \lesssim 1$. 
\end{assumption}

To describe the first-stage rate requirement, Assumption \ref{ass:smallbias} introduces a sequence of nuisance realization sets $\mathcal{T}_N$ for the nuisance parameter $\eta_0$. As sample size $N$ increases, the sets $\mathcal{T}_N$ shrink around the true value $\eta_0$.   The shrinkage speed is described in terms of the statistical rates $B_{N}$ and $\Lambda_N$.
\begin{assumption}[Small Bias Condition]
\label{ass:smallbias}
There exists a sequence $\epsilon_N = o(1)$, such that with probability at least $1-\epsilon_N$, for all $k \in [K]$, the first stage estimate $\widehat{\eta}_k $, obtained by cross-fitting (Definition \ref{sampling}), belongs to a shrinking neighborhood of $\eta_0$, denoted by $\mathcal{T}_N$. Uniformly over $\mathcal{T}_N$, the following mean square convergence holds:
	\begin{align}
	B_{N}:= \sqrt{N} \sup_{\eta \in \mathcal{T}_N} \| \Ep  p(X) [Y(\eta) - Y(\eta_0)] & \|= o(1),\\
	\Lambda_N:=  \sup_{\eta \in \mathcal{T}_N} (\E \| p(X) [Y(\eta) - Y(\eta_0)]  \|^2)^{1/2} &= o(1 ). 
	\end{align}
In particular, $\Lambda_N$ can be bounded as $\Lambda_N \lesssim \xi_d \sup_{\eta \in \mathcal{T}_N} (\E  (Y(\eta) - Y(\eta_0))^2)^{1/2}.$

\end{assumption}

\begin{remark}[Sufficient Conditions for Assumption \ref{ass:smallbias}]
\label{rm:rate}
Assumption \ref{ass:smallbias} is stated in a high-level form in order to accommodate various machine learning estimators. We demonstrate the plausibility of Assumption \ref{ass:smallbias}  for a high-dimensional sparse model for Example \ref{MD} in Appendix B, adapting the work of   \cite{Program}. Furthermore, one can also use deep neural networks (\cite{Schmidt}, \cite{Farrell}), random forest in small (\cite{WagerWalther}) dimensions and high (\cite{fastRF}) dimensions with sparsity structure. 

\end{remark}

\subsection{Pointwise Limit Theory}

In this section, we establish pointwise asymptotic properties for the Orthogonal Estimator. Our first result concerns with mean square convergence rate and pointwise linearization. 

\begin{lemma}[$L_2$ Rate and Pointwise Linearization]
\label{lem:pointwise}
 Let  Assumptions \ref{ass:identification}-\ref{ass:smallbias}  hold. Then, the following statements hold:	
 \begin{enumerate}
		\item[(a)] The $\ell_2$-norm of the estimation error is bounded as: $$ \|\widehat{\beta}-\beta_0 \|_{2}  \lesssim_{P} \sqrt{ \dfrac{d}{N} } + \big[ \sqrt{\dfrac{d}{N}} l_d r_d  \wedge \xi_d r_d /\sqrt{N} \big], $$ which implies a bound on the mean squared error of the estimate $p(x)'\widehat{\beta}$ of the pseudo-target function $p(x)'\beta_0$: $$ (\EN (p(X_i)' (\widehat{\beta} - \beta_0))^2)^{1/2}  \lesssim_{P}  \sqrt{ \dfrac{d}{N} } + \big[\sqrt{\dfrac{d}{N}}l_d r_d  \wedge \xi_d r_d /\sqrt{N} \big]. $$
		        \item[(b)] For any $\alpha \in \mathcal{S}^{d-1} := \{ \alpha \in \mathbb{R}^d: \| \alpha \| = 1\}$, the estimator $\widehat{\beta}$ is approximately linear:
		         \begin{align*}
        \sqrt{N} \alpha' (\widehat{\beta} - \beta)  = \alpha' Q^{-1}  \GN p(X_i) (U_i + r_g(X_i)) + R_{1,N} (\alpha),
        \end{align*}
        where the remainder term $ R_{1, N} (\alpha)$ is bounded as  $$R_{1,N} (\alpha)  \lesssim_{P}  B_N+ \Lambda_N+
 \sqrt{\dfrac{\xi_d^2 \log N}{N} } \bigg( 1 + \min \bigg\{ l_d r_d \sqrt{d} ,  \xi_d r_d \bigg\} \bigg) .$$ 
 \end{enumerate}
\end{lemma}

Under small bias condition, Lemma \ref{lem:pointwise} states that the Orthogonal  Estimator converges at the oracle rate and achieves oracle asymptotic linearity representation, where oracle knows the true value of the nuisance parameter $\eta_0$.

\begin{theorem}[Pointwise Normality of the Orthogonal Estimator] 
 Suppose  Assumptions \ref{ass:identification}-\ref{ass:smallbias}  hold. In addition, suppose $ (\xi_d^2 \log N/N)^{1/2} \cdot ( 1 + l_d r_d \sqrt{d})=o(1)$, $1 \lesssim \underbar{$\sigma$}^2$ and the  Lindeberg condition holds: $ \sup_{x \in \mathcal{X}} \Ep [U^2 1_{|U|>M} |X=x] \rightarrow 0, \quad M \rightarrow \infty$. Then for any $\alpha \in \mathcal{S}^{d-1}$, the Orthogonal Estimator is asymptotically normal:
		\begin{align}
	\lim_{N \rightarrow \infty} \sup_{t \in \mathbb{R}} \left| \Pr \left(\dfrac{\sqrt{N} \alpha' (\widehat{\beta} - \beta_0)}{\sqrt{\alpha' \Omega \alpha }}  < t \right) - \Phi(t) \right| = 0.
	\end{align} 
	Moreover, for any $x_0=x_{0,N} \in \mathcal{X}$ the estimator $p(x_0)'\widehat{\beta}$ of the pseudo-target value $p(x_0)'\beta_0$ is asymptotically normal:
		\begin{align}
	\lim_{N \rightarrow \infty} \sup_{t \in \mathbb{R}} \left| \Pr \left(\dfrac{\sqrt{N} p(x_0)' (\widehat{\beta} - \beta_0)}{\sqrt{p(x_0)' \Omega p(x_0) }}  < t \right) - \Phi(t) \right| = 0,
	\end{align}

	and if the approximation error is negligible relative to the estimation error, namely $\sqrt{N}r_g(x_0)=o( \| \Omega^{1/2} p(x_0) \|)$, then $\widehat{g}(x)=p(x_0)'\widehat{\beta}$ is asymptotically normal:
		\begin{align}
	\lim_{N \rightarrow \infty} \sup_{t \in \mathbb{R}} \left| \Pr \left(\dfrac{\sqrt{N} (\widehat{g}(x_0) - g(x_0))}{\sqrt{p(x_0)' \Omega p(x_0) }}  < t \right) - \Phi(t) \right| = 0.
	\end{align}
	\label{thrm:OS} 
\end{theorem}

Theorem \ref{thrm:OS}  delivers the pointwise convergence in distribution of the Orthogonal  Estimator for any point $x_0$ that can depend on $N$.

\subsection{Uniform Limit Theory}

In this section, we establish uniform asymptotic properties for the  Orthogonal Estimator. Not surprising, stronger conditions are required for our results to hold when compared to the pointwise case. Let $m>2$. The following assumption controls the tails of the regression errors.

\begin{assumption}[Tail Bounds]
 \label{ass:merror}
 	There exists a constant $m>2$ such that the upper bound of the $m$'th moment of $|U|$ is bounded conditional on $X$:  $$ \sup_{x \in \mathcal{X}} \Ep [| U|^m|X = x] \lesssim 1.$$
\end{assumption}

Denote by $\alpha(x) := p(x)/\| p(x)\|$  the normalized value of basis functions vector $p(x)$. Define  Lipschitz constant  as:
$$\xi_d^L = \sup_{x,x' \in \mathcal{X}, x \neq x' } \dfrac{\|  \alpha(x) - \alpha(x') \|}{\| x- x'\|}.$$

\begin{assumption}[Basis]
\label{ass:basis}
	Basis functions are well-behaved, namely (i) $(\xi_d^L)^{2m/(m-2)} \log {N}/N \lesssim 1$  and $\log \xi_d^L \lesssim \log d$ for the same $m$ as in Assumption \ref{ass:merror}. 
\end{assumption}

\begin{assumption}[Condition for Matrix Estimation]
\label{ass:matrix}
Let $\mathcal{T}_N$ be as in Assumption \ref{ass:smallbias}. Uniformly over $\mathcal{T}_N$, the following convergence holds:
       \begin{align}
       \kappa_N^1 := \sup_{\eta \in \mathcal{T}_N} \E [  \max_{1 \leq i \leq N} |Y_i(\eta) - Y_i(\eta_0) | ] &= o(1), \\
        \kappa_N :=\sup_{\eta \in \mathcal{T}_N} (\E \max_{1 \leq i \leq N} (Y_i(\eta) - Y_i(\eta_0))^2)^{1/2} &= o(1).
	\end{align}

\end{assumption}

Lemma \ref{lem:uniform} establishes  asymptotic linearity representation uniformly over the domain $\mathcal{X}$ and  uniform convergence rate.
	
\begin{lemma}[Uniform Rate and Uniform Linearization]
    \label{lem:uniform}
  Suppose Assumptions \ref{ass:identification}-\ref{ass:basis} hold. 
  \begin{enumerate}
		\item[(a)]  The Orthogonal Estimator is approximately linear uniformly over  $\mathcal{X}$:

		\begin{align*} 		 | \sqrt{N} \alpha(x)' (\widehat{\beta} - \beta_0) - \alpha'(x) Q^{-1} \GN p(X_i)[U_i + r_g(X_i)] | \leq R_{1,N}(\alpha(x)),
		 \end{align*} where $ R_{1,N}(\alpha(x))$, summarizing the impact of unknown design and the first stage error, obeys \begin{align*} \sup_{x \in \mathcal{X}} R_{1,N}(\alpha(x)) \lesssim_{P} 
		B_N+ \Lambda_N +\sqrt{\dfrac{\xi_d^2 \log N}{N}}(N^{1/m}  \sqrt{\log N}+ \sqrt{d} l_d r_d )   =: \bar{R}_{1N} \end{align*} uniformly over $x \in \mathcal{X}$. Moreover, $$ | \sqrt{N} \alpha(x)' (\widehat{\beta} - \beta_0) - \alpha'(x) Q^{-1} \GN p(X_i) U_i | \leq R_{1,N}(\alpha(x)) + R_{2,N}(\alpha(x)),$$ where $R_{2,N}(\alpha(x))$, summarizing the impact of misspecification error, obeys  $$R_{2,N}(\alpha(x)) \lesssim_{P} \sqrt{\log N} l_d r_d=: \bar{R}_{2N} $$ uniformly over $x \in \mathcal{X}$.
	\item[(b)] The estimator $p(x)'\widehat{\beta}$ of the pseudo-target $p(x)'\beta_0$ converges uniformly over $\mathcal{X}$: $$\sup_{x \in \mathcal{X}} | p(x)' (\widehat{\beta} - \beta_0) | \lesssim_{P} \dfrac{\xi_d}{\sqrt{N}} [\sqrt{\log N} + \bar{R}_{1N} + \bar{R}_{2N}].$$
	\end{enumerate}
\end{lemma}

\begin{remark}[Optimal Uniform Rate in Holder class]
Suppose the true function $g(x)$ belongs to the Holder smoothness class of order $k$, denoted by  $\Sigma_k(\mathcal{X})$. Suppose $l_d r_d \lesssim d^{-k/\text{dim}(X)}, \xi_d \lesssim \sqrt{d}, \bar{R}_{1N} + \bar{R}_{2N} \lesssim (\log N)^{1/2}$.  Then, the optimal number $d$ of technical regressors that comprise a vector $p(x)$ obeys $$d \asymp (\log N/N )^{-\text{dim}(X)/(2k+\text{dim}(X))}.$$ This choice of $d$ yields the optimal uniform rate: $$ \sup_{x \in \mathcal{X}} | \widehat{g}(x) - g(x)| \lesssim_{P} \Big( \dfrac{\log N}{N} \Big)^{k/(2k+\text{dim}(X))} . $$
\end{remark}

Theorem \ref{thrm:ult} establishes a strong approximation of  the Orthogonal Estimator's series process by a sequence of zero-mean Gaussian processes.

\begin{theorem}[Strong Approximation by a Gaussian Process] 	 
\label{thrm:ult}
Suppose Assumptions \ref{ass:identification}-\ref{ass:basis} hold with $m \geq 3$.  Let $\bar{a}_N$ be a sequence of positive numbers s.t. $\bar{a}_N^{-1}=o(1)$. Suppose (i)  $\bar{R}_{1N}=o(\bar{a}^{-1}_N)$, (ii) 
$1 \lesssim \underbar{$\sigma$}^2$ and (iii) $d^4 \bar{a}_N^6 \xi_d^2 (1 +l_d^3 r_d^3)^2 \log^2 N /N=o(1)$. Then, for some
$\mathcal{N}_d \sim N(0, I_d)$,  the following statement holds for $e(x)=\Omega^{1/2}p(x)$
\begin{align}
       \sqrt{N}\dfrac{p(x)'(\widehat{\beta} -\beta_0)}{\| e(x) \|}&=_d \dfrac{e(x)}{\| e(x) \|}\mathcal{N}_{d}+ o_{P} (\bar{a}_N^{-1})  \text{ in } \ell^{\infty}(\mathcal{X}).  
\end{align}
In addition, if  $\sup_{x \in \mathcal{X}} \sqrt{N} |r(x)|/ \| e(x) \|=o(\bar{a}_N^{-1})$, then, for $\widehat{g}(x)=p(x)'\widehat{\beta}$
\begin{align}
    \sqrt{N}\dfrac{\widehat{g}(x) -g(x)}{\| e(x) \|}&=_d \dfrac{e(x)}{\| e(x) \|}\mathcal{N}_{d} + o_{P} (\bar{a}_N^{-1}) \text{ in } \ell^{\infty}(\mathcal{X}).  
\end{align}

\end{theorem}



Theorem \ref{lem:matrix} establishes the convergence rate for the covariance matrix estimator $\widehat{\Omega}$. It extends Theorem 4.6 of \cite{NewOls}, allowing the signal $Y(\eta_0)$ to depend upon an unknown nuisance parameter $\eta_0$.

\begin{theorem}[Matrices Estimation]
\label{lem:matrix}
Suppose Assumptions \ref{ass:identification}-\ref{ass:matrix} hold.  In addition, suppose (i)  $\bar{R}_{1N} + \bar{R}_{2N} \lesssim \sqrt{\log N}$, (ii) $(N^{1/m}+l_d r_d) (\sqrt{\frac{\xi_d^2 \log N}{N}} + \kappa_N^1  )= o(1)$. Then, the estimator $\widehat{\Omega}$, defined in \eqref{eq:omegahat},  converges in the matrix operator norm with the following rate:
\begin{align*}
\| \widehat{\Omega} - \Omega \| &\lesssim_{P} \bigg(N^{1/m} +l_d r_d \bigg) \cdot \bigg(\sqrt{\frac{\xi_d^2 \log N}{N}} +\kappa_N^1 \bigg)+ \kappa_N^2 =:a_N.
\end{align*}
Moreover, for $ \sigma_N(x) = \sqrt{ p(x)' \Omega p(x)/N}$ and  $ \widehat{\sigma}_N(x) = \sqrt{ p(x)' \widehat{\Omega} p(x)/N}$, the following bound holds:
\begin{align}
\label{eq:lemma51}
\sup_{x \in \mathcal{X}}  \bigg|  \dfrac{\widehat{\sigma}_N(x)}{ \sigma_N(x)}-1 \bigg| \lesssim_{P} \| \widehat{\Omega} - \Omega \|  \lesssim_{P}  a_N.
\end{align}
%
\end{theorem}

Theorem \ref{cor:bootstrap} establishes validity of empirical (Gaussian) bootstrap.

\begin{theorem}[Validity of Gaussian Bootstrap]
 \label{cor:bootstrap}
 Suppose Assumptions of Theorem \ref{thrm:ult} hold with $\bar{a}_N = \log N$ and
 Assumptions of Theorem \ref{lem:matrix} hold with  $a_N=O(N^{-b})$ for some $b>0$.  
In addition, suppose (i) $1 \lesssim \underbar{$\sigma$}^2$ and there exists a sequence $\xi_N'$ obeying $1 \lesssim \xi_N' \lesssim \| p(x) \|$ uniformly for all $x \in \mathcal{X}$ so that $\| p(x) - p(x') \|/\xi_N' \leq L_N \| x - x' \|$, where $\log L_N \lesssim \log N$. Let $\mathcal{N}^b_{d}$ be a bootstrap draw from $N(0, I_d)$ and $P^{*}$ be a probability conditional on data $(V_i)_{i=1}^N$.  Then, the following approximation holds uniformly in $\ell^{\infty}(\mathcal{X})$:
 \begin{align}
 \label{eq:nb1}
    \dfrac{p(x)'\widehat{\Omega}^{1/2} }{\| \widehat{\Omega}^{1/2} p(x) \|}   \mathcal{N}^b_{d}  =^d   \dfrac{p(x)'\Omega^{1/2} }{\| {\Omega}^{1/2} p(x) \|}\mathcal{N}^b_{d} + o_{P^{*}}(\log^{-1}N).
 \end{align}

 \end{theorem}

Theorem \ref{cor:confbands} establishes the asymptotic validity of uniform confidence bands. It also shows that the uniform width of the bands is of the same order as the uniform rate of
convergence.
 
\begin{theorem}[Validity of Uniform  Confidence Bands]
\label{cor:confbands}
Let Assumptions \ref{ass:identification}-\ref{ass:matrix} hold  with $m \geq 4$. In addition, suppose (i) $\bar{R}_{1N} + \bar{R}_{2N} \lesssim \log^{-1/2} N$, (ii) $\xi_d \log^2 N/N^{1/2-1/m}=o(1)$, (iii) $1 \lesssim \underbar{$\sigma$}^2$, (iv) $\sup_{x \in \mathcal{X}} \sqrt{N} |r_g(x)|/\|p(x) \|=o(\log^{-1/2} N)$, (v) $d^4 \xi_d^2 (1+l_d^3r_d^3)^2 \log^5 N/N=o(1)$.  Then,
\begin{align*}
\Pr \big( \sup_{x \in \mathcal{X}} | t_N(x) | \leq c_N (1-\alpha) \big) = 1 - \alpha + o (1)
\end{align*}
for $t_N$ defined in \eqref{eq:tN}. As a consequence, the confidence bands defined in  \eqref{eq:bootconfbands} satisfy
\begin{align*}
\Pr \big( g(x) \in [ \underline{i}(x), \overline{i}(x) ]  \quad  \forall x \in \mathcal{X}  \big) = 1 - \alpha + o (1).
\end{align*}
The width of the confidence bands $2  c_N (1-\alpha) \widehat{\sigma}_N(x) $ obeys
\begin{align*} 
2  c_N (1-\alpha) \widehat{\sigma}_N(x) \lesssim_{P} \sigma_N(x)\sqrt{ \log N}  \lesssim \sqrt{\dfrac{ \xi_d^2 \log N}{N }}
\end{align*}
uniformly over $x \in \mathcal{X}$.
\end{theorem}

\section{Applications}
\label{sec:apps}
In this section, we apply the results of Section \ref{sec:asymp} for empirically relevant settings, described in Examples \ref{CATEcont}-\ref{CAPD}.

\subsection{Continuous Treatment Effects}
\label{ex:conttreatment}
Consider the setup of Example \ref{CATEcont}. We provide sufficient low-level conditions on the first-stage nuisance parameter $\eta_0(x,z)$ such that the pointwise and uniform Gaussian approximations established in Section \ref{sec:asymp} hold.

Assume that there exists a sequence of numbers $\epsilon_N = o(1)$ and sequences of neighborhoods $S_N$ of  $s_0(\cdot|\cdot) $, $M_N$ of $\mu_0(\cdot,\cdot)$, $W_N$ of $w_0(\cdot)$ such that the first-stage estimate $$\{\widehat{s}(\cdot|\cdot), \widehat{\mu}(\cdot,\cdot),\widehat{w}(\cdot)\} $$ belongs to the set $\{S_N \bigtimes M_N \bigtimes  W_N \}$ w.p. at least $1-\epsilon_N$. The shrinkage speed of this set is measured by the following statistical rates:
\begin{align*}
\textbf{s}_{N,q} &:= \sup_{ s \in S_N} (\E ( s(X|Z) - s_0(X|Z))^q)^{1/q}, \\
 \textbf{m}_{N,q} &:= \sup_{\mu \in M_N } (\E ( \mu(X,Z) - \mu_0(X,Z))^q)^{1/q}, \\
 \textbf{w}_{N,q} &:= \sup_{ w \in W_N} (\E ( w(X) - w_0(X))^q)^{1/q},
\end{align*}
where $q$ is either a positive number $q \geq 2$ or $q=\infty$, which corresponds to $\ell_{\infty}$-norm (sup-norm). For $q=2$,  we will refer to $\textbf{s}_{N}:=\textbf{s}_{N,2}$ as the conditional density mean square rate, $\textbf{m}_{N}:=\textbf{m}_{N,2}$ as the regression function mean square rate, and $ \textbf{w}_{N}=\textbf{w}_{N,2}$ as the marginal density mean square rate.

\begin{assumption}[First-Stage Rate of CTE]
\label{ass:fsratecont} 
Assume that mean square rates $\textbf{s}_N, \textbf{m}_N,  \textbf{w}_N$ decay sufficiently fast: $$\xi_d (\textbf{s}_N \vee  \textbf{m}_N \vee  \textbf{w}_N)  = o(1)$$ 
and one of two alternative conditions hold. (1) Bounded basis. There exist $\bar{B}< \infty$ so that $\sup_{x \in \mathcal{X}} \|p(x)\|_{\infty} \leq \bar{B}$ and $ \sqrt{N}\sqrt{d} (\textbf{m}_N \textbf{s}_N \vee \textbf{m}_N \textbf{w}_N )= o(1)$. (2) Unbounded basis. There exist $\kappa, \gamma \in [1, \infty], \quad 1/\kappa + 1/\gamma=1$ so that   $ \sqrt{N}\sqrt{d} (\textbf{m}_{N,2\gamma} \textbf{s}_{N,2\kappa} \vee \textbf{m}_{N,2\gamma} \textbf{w}_{N,2\kappa} )= o(1)$. Furthermore, there exists a constant $\bar{\mathcal{C}}< \infty$ such that
$ \sup_{\mu \in M_N } \sup_{(x,z) \in \mathcal{X} \bigtimes \mathcal{Z}}   | \mu (x,z) | < \bar{\mathcal{C}}$,  $ \sup_{s \in S_N } \sup_{(x,z) \in \mathcal{X} \bigtimes \mathcal{Z}} s^{-1}(x|z)   < \bar{\mathcal{C}}$, $ \sup_{s \in S_N } \sup_{(x,z) \in \mathcal{X} \bigtimes \mathcal{Z}} s(x|z)   < \bar{\mathcal{C}}$, $$ \sup_{w \in W_N }  \sup_{x \in \mathcal{X}}   w^{-1}(x)  < \bar{\mathcal{C}}.$$
\end{assumption}
\begin{lemma}[Orthogonal Signal for CTE]
\label{thrm:conttreat}
Suppose Assumption \ref{ass:fsratecont} holds. Then, the orthogonal signal $Y(\eta)$, defined in \eqref{eq:contoutcome}, satisfies Assumption \ref{ass:smallbias}. 
\end{lemma}
As discussed in the Introduction, the  estimator $\widehat{\beta}^{\dagger}$  takes the form
\begin{align}
\label{eq:betadagger}
\widehat{\beta}^{\dagger} = \widehat{Q}^{-1} \dfrac{1}{N} \sum_{i=1}^N p(X_i) Y^{\dagger}_i(\widehat{\eta}),
\end{align}
where $Y^{\dagger}_i(\widehat{\eta})$ is as in \eqref{eq:feassignal}. Since $\mu$ enters linearly in \eqref{eq:contoutcome}, the error term $Y^{\dagger}_i(\widehat{\eta})-  Y_i(\widehat{\eta})$  does not introduce bias in $\widehat{\beta}^{\dagger}$, but introduces an extra term in asymptotic variance, which we characterize below. For a function $\mu(x,z)$, define its demeaned analog $\mu^0(x,z)$ as
\begin{align*}
	\mu^{0}(x,z) := \mu(x,z) - \E \mu(x,Z)
\end{align*}
and the kernel function as 
\begin{align}
\label{eq:tau}
\tau (v_1, v_2; \mu) =\dfrac{1}{2} (p(x_1) \mu^0(x_1, z_2) +  p(x_2) \mu^0(x_2, z_1)),
\end{align}
where $v_1=(x_1,z_1)$ and $v_2=(x_2,z_2)$. Finally, define the Hajek projection of $\tau (v_1, v_2; \mu)$ as
\begin{align}
\label{eq:tau1}
\tau_1(v;\mu):=  \E \tau (v, V; \mu) =\E p(X) \mu^0(X, z) =  \tau_1(z;\mu).
\end{align}
Decompose the Orthogonal Estimator $\widehat{\beta}^{\dagger}$  into the sum of its infeasible analog $\widehat{\beta}$ and a mean-zero $U$-statistic with the kernel function $\tau(v_1, v_2;\mu)$
\begin{align*}
\widehat{\beta}^{\dagger} &=:\widehat{\beta} + \widehat{Q}^{-1} \dfrac{1}{K} \sum_{k=1}^K \dfrac{1}{n(n-1)} \sum_{i,j \in J_k, i \neq j} \tau (V_i, V_j; \widehat{\mu}).
\end{align*}

Theorem \ref{thrm:cte:pointwise} establishes pointwise asymptotic normality of the Orthogonal Estimator.  Its asymptotic variance  is $$
    \Omega^{\dagger} = Q^{-1} \Sigma^{\dagger}  Q^{-1},$$ where $\Sigma^{\dagger}$ is
\begin{align}
    \label{eq:covmat:cte}
    \Sigma^{\dagger} = \E [p(X) (U + r_g(X)) + \tau_1(Z;\mu_0)][p(X) (U + r_g(X)) + \tau_1(Z;\mu_0)]'.
\end{align}


 \begin{theorem}[Pointwise Asymptotic Theory for Continuous Treatment Effects]
 \label{thrm:cte:pointwise}
 Suppose Assumptions \ref{ass:identification}-\ref{ass:boundederror} and \ref{ass:fsratecont} hold. Let  $C_{\min}^{\dagger}>0$ be an absolute constant. In addition, suppose (i) $ (\xi_d^2 \log N/N)^{1/2} \cdot ( 1 + l_d r_d \sqrt{d})=o(1)$, (ii) $\sqrt{d\xi_d^4 \log^3 N/N^2} +  \sqrt{d}\textbf{m}_N=o(1)$, $\xi_d l_d r_d = o(N^{1/2})$, (iii) $\min \eig \Omega^{\dagger} \geq C_{\min}^{\dagger}$,  and (iv)  the  Lindeberg condition holds: $ \sup_{x \in \mathcal{X}} \Ep [U^2 1_{|U|>M} |X=x] \rightarrow 0, \quad M \rightarrow \infty$.  Then, for any   $x_0 = x_{0,N} \in \mathcal{X}$ the estimator $p(x_0)'\widehat{\beta}^{\dagger}$ of the pseudo-target $p(x_0)'\beta_0$ is asymptotically normal:
		\begin{align}
	\lim_{N \rightarrow \infty} \sup_{t \in \mathbb{R}} \left| \Pr \left(\dfrac{\sqrt{N}  p(x_0)' (\widehat{\beta}^{\dagger} - \beta_0)}{\sqrt{ p(x_0)' \Omega^{\dagger} p(x_0)
	}}  < t \right) - \Phi(t) \right| = 0,
	\end{align} 
	and if the approximation error is negligible relative to the estimation error, namely $\sqrt{N}r_g(x_0)=o( \| (\Omega^{\dagger})^{1/2} p(x_0) \|$, then $\widehat{g}^{\dagger} (x_0) = p(x_0)'\widehat{\beta}^{\dagger}$ is asymptotically normal:
	\begin{align}
	\lim_{N \rightarrow \infty} \sup_{t \in \mathbb{R}} \left| \Pr \left(\dfrac{\sqrt{N}   ( \widehat{g}^{\dagger} (x_0) - g(x_0))}{\sqrt{ p(x_0)' \Omega^{\dagger} p(x_0)
	}}  < t \right) - \Phi(t) \right| = 0.
	\end{align}
\end{theorem}

\begin{theorem}[Uniform Asymptotic Theory for Continuous Treatment Effects]
Suppose Assumptions \ref{ass:identification}-\ref{ass:basis} hold with $\xi_d^L/C_{\text{min}} \geq e^2/16 \vee e$.  In addition, suppose 
$ \sqrt{ d \xi_d^2 \log N/N} =o(1)$ and $\xi_d \log^2 N/N=o(1)$. Then, the following statements hold.
 \begin{enumerate}
		\item[(a)]  The Orthogonal Estimator converges uniformly over $\mathcal{X}$:
		\begin{align*}
		    \sup_{x \in \mathcal{X}}| p(x)' (\widehat{\beta}^{\dagger} - \beta_0)| &\lesssim_{P} \dfrac{\xi_d}{\sqrt{N}} \bigg(\sqrt{\log N} +\bar{R}_{1N} + \bar{R}_{2N}  + \sqrt{d} \textbf{m}_N  \bigg),
		\end{align*}
		where $\bar{R}_{1N}$ and $\bar{R}_{2N}$ are as in Lemma \ref{lem:uniform}.
		\item[(b)] Suppose Assumption \ref{ass:merror}-\ref{ass:basis} hold with $m \geq 3$. Then, the statement of  Theorem \ref{thrm:ult} holds for $\widehat{\beta}^{\dagger}$ in place of $\widehat{\beta}$,
		$\Omega^{\dagger}$ in place of $\Omega$		and $e(x):=(\Omega^{\dagger})^{1/2} p(x)$. 
\end{enumerate}

 \label{thrm:cte:uniform}
 
\end{theorem}

Theorem \ref{thrm:cte:uniform} establishes uniform convergence rate and strong Gaussian approximation for Continuous Treatment Effects. We compare our Theorems \ref{thrm:cte:pointwise} and  \ref{thrm:cte:uniform} to \cite{Kennedy} who introduced the doubly robust score for the average potential outcome. First, by virtue of sample splitting, we do not impose any complexity requirements on the estimator of the first-stage nuisance parameters. In particular, unlike Theorem 2 of \cite{Kennedy}, we do  not require the function class  containing $\widehat{\mu}(\cdot, \cdot), \widehat{s}(\cdot|\cdot), \widehat{w}(\cdot)$ to have bounded uniform entropy integrals. As a result, our method accommodates a wide class of modern regularized methods to be employed to estimate first-stage parameters as discussed in Remark \ref{rm:rate}. Second, Theorem \ref{thrm:cte:uniform} offers uniform asymptotic statements, while 
 while the work by \cite{Kennedy} offers only  pointwise results. Finally, \cite{Kennedy} employs local linear regression in the second stage and delivers convergence at a $\sqrt{N h^{\text{dim}(X)}}$ rate, where $h=h(N)$ is the kernel bandwidth. In contrast, our method delivers $\sqrt{N}$- approximation for the normalized projection $p(x)' \beta_0/\| p(x) \|$.

\subsection{Conditional Average Treatment Effect}
\label{ex:CATE}
Consider the setup of Example \ref{CATE}. We provide sufficient low-level conditions on the  regression functions $\mu_0(1,\cdot), \mu_0(0,\cdot)$ and the propensity score $s_0(\cdot)$ such that the pointwise and uniform Gaussian approximations  of Section \ref{sec:asymp} hold. 
\begin{assumption}[Strong Overlap]
	\label{ass:overlap} 
	\begin{enumerate}
	\item[A] The propensity score is bounded above and below. Specifically, there exists $\bar{\pi}_0 > 0$ such that  $ 0< \bar{\pi}_0 < s_0(z) < 1-\bar{\pi}_0<1 $ for any $z \in \mathcal{Z}$. 
	\item[B] The propensity score is bounded below. Specifically, there exists $\bar{\pi}_0 > 0$ such that $0< \bar{\pi}_0 < s_0(z) <1 $ for any $z \in \mathcal{Z}$. 
	\end{enumerate}
\end{assumption}
In the context of Example \ref{CATE} Assumption \ref{ass:overlap}(a) ensures that the probability of assignment to the treatment and control group is bounded away from zero. In the context of Example \ref{MD} Assumption \ref{ass:overlap}(b) ensures that the probability of observing the response $Y^{*}$ is bounded away from zero. 

Given the true functions $s_0(\cdot), \mu_0(1,\cdot),  \mu_0(0,\cdot)$ and  sequences of shrinking neighborhoods  $S_N$ of $s_0(\cdot)$,  $M_N$ of  $\mu_0(1,\cdot)$ and of  $\mu_0(0,\cdot)$ define the following rates:
	           $$ \textbf{s}_{N,q} := \sup_{s \in \mathcal{S}_N} ( \Ep   ( s(Z) - s_0(Z))^q)^{1/q},$$  
           $$  \textbf{m}_{N,q} := \sup_{\mu  \in M_N}  ( \Ep   ( \mu (1,Z) -\mu_0(1,Z))^q)^{1/q} \vee \sup_{\mu  \in M_N}  ( \Ep   ( \mu (0,Z) -\mu_0(0,Z))^q)^{1/q}, $$ 
where $q\geq 2$ is either a positive number or $q=\infty$.  We will refer to $\textbf{s}_N:=\textbf{s}_{N,2}$ as the propensity score mean square  rate and $\textbf{m}_N:=\textbf{m}_{N,2}$ as the regression function mean square  rate.
\begin{assumption}[First-Stage Rate for CATE] 
\label{ass:fsrate} 
Assume that there exists a sequence of numbers $\epsilon_N = o(1)$ and sequences of neighborhoods $S_N$ of  $s_0(\cdot) $, $M_N$ of $\mu_0(1,\cdot)$ and  $\mu_0(0,\cdot)$  such that the first-stage estimate $\{\widehat{s}(\cdot), \widehat{\mu}(1,\cdot),\widehat{\mu}(0,\cdot)\} $ belongs to the set $\{S_N \bigtimes M_N \bigtimes M_N\}$ w.p. at least $1-\epsilon_N$.  Assume that mean square rates $\textbf{s}_N, \textbf{m}_N$ decay sufficiently fast: $$\xi_d (\textbf{s}_N \vee  \textbf{m}_N)  = o(1)$$ 
and one of two alternative conditions hold. (1) Bounded basis. There exist $\bar{B}< \infty$ so that $\sup_{x \in \mathcal{X}} \|p(x)\|_{\infty} \leq \bar{B}$ and $ \sqrt{N}\sqrt{d} \textbf{m}_N \textbf{s}_N = o(1)$. (2) Unbounded basis. There exist $\kappa, \gamma \in [1, \infty], \quad 1/\kappa + 1/\gamma=1$ so that   $ \sqrt{N}\sqrt{d} \textbf{m}_{N,2\gamma} \textbf{s}_{N,2\kappa} = o(1)$. 
	Finally, the functions in $S_N$ and $M_N$ are bounded uniformly over their domain:   $$ \sup_{\mu \in M_N } \sup_{z \in \mathcal{Z}}  \sup_{d \in \{1,0\}} | \mu (d,z) | \vee \sup_{s \in S_N } \sup_{z \in \mathcal{Z}}  s^{-1}(z)  < \bar{\mathcal{C}} < \infty.$$
\end{assumption}

\begin{corollary}[Asymptotic Theory for Conditional Average Treatment Effect]
\label{thrm:CATE}
Under Assumptions \ref{ass:overlap}[A] and \ref{ass:fsrate}, the orthogonal signal  $Y(\eta)$, given by Equation \eqref{eq:catesig}, satisfies Assumption \ref{ass:smallbias}. As a result, the statements of Theorems \ref{thrm:OS}-\ref{cor:confbands} hold for Conditional Average Treatment Effect.

\end{corollary}

\subsection{Regression Function with Partially Missing Outcome}
\label{ex:MD}
Consider the setup of Example \ref{MD}. Define the regression function rate $\textbf{m}_{N,q}$ as
\begin{align*}
    \textbf{m}_{N,q} := \sup_{\mu  \in M_N}  ( \Ep   ( \mu (Z) -\mu_0(Z))^q)^{1/q}
\end{align*}
where $q\geq 2$ is either a positive number or $q=\infty$ and let $\textbf{s}_{N,q}$ be as defined in Section \ref{ex:CATE}.  We show that pointwise and uniform Gaussian approximations  of Section \ref{sec:asymp} hold for Regression Function with Partially Missing Outcome.



\begin{corollary}[Asymptotic Theory for Regression Function with Partially Missing Outcome]
\label{thrm:MD}
Suppose Assumptions \ref{ass:overlap} [B] and Assumption \ref{ass:fsrate} hold for $\textbf{s}_{N,q}$ defined in Example \ref{ex:CATE} and $\textbf{m}_{N,q}$ redefined above. Then,  the orthogonal signal $Y(\eta)$, given by Equation \eqref{eq:mdsig}, satisfies Assumption \ref{ass:smallbias}. Then, the statements of Theorems \ref{thrm:OS}-\ref{cor:confbands} hold for the regression function with Partially Missing Outcome. 
\end{corollary}
We give an example of low-level sparse conditions for the propensity score $s(\cdot)$ and regression $\mu(\cdot)$ in Appendix B.

\subsection{Conditional Average Partial Derivative}
\label{ex:CAPD}
Consider the setup of Example \ref{CAPD}. We provide sufficient low-level conditions on the  regression functions $s(d|z), \mu(d,z)$ such that the pointwise and uniform Gaussian approximations  of Section \ref{sec:asymp} hold.

Given a true function $s_0(\cdot|\cdot), \mu_0(\cdot,\cdot)$, let $S_N, M_N$ be a sequence of shrinking neighborhoods of $s_0(\cdot|\cdot)$ and $\mu_0(\cdot,\cdot)$ constrained as follows:
	           $$ \textbf{s}_{N,q} := \sup_{s \in \mathcal{S}_N} ( \Ep   ( s(D|Z)) - s_0(D|Z))^q)^{1/q} \vee  \sup_{s \in \mathcal{S}_N} ( \Ep   ( \partial_{d} s(D|Z)) - \partial_{d} s_0(D|Z))^q)^{1/q}  $$ 
           $$  \textbf{m}_N := \sup_{\mu  \in M_N}  ( \Ep   ( \mu (D,Z) -\mu_0(D,Z))^q)^{1/q} \vee \sup_{\mu  \in M_N}  ( \Ep   ( \partial_{d} \mu (D,Z) -\partial_{d} \mu_0(D,Z))^q)^{1/q},  $$
where $q\geq 2$ is either a positive number or $q=\infty$.  We will refer to $\textbf{s}_N := \textbf{s}_{N,2}$ as mean square conditional density rate and $\textbf{m}_N:=\textbf{m}_{N,2}$ as mean square regression function rate. 
\begin{assumption}[First-Stage Rate for CAPD]
	\label{ass:fsrate3}   Assume that there exists a sequence of numbers $\epsilon_N = o(1)$ and a sequence of neighborhoods $S_N,M_N$ of functions $s_0(\cdot|\cdot)$, $\mu_0(\cdot,\cdot)$ 
	such that the first-stage estimate $\{\widehat{s} (\cdot|\cdot), \widehat{\mu} (\cdot) \} $ belongs to the set $\{S_N \bigtimes  M_N \}$ w.p. at least $1-\epsilon_N$. 
	Assume that mean square rates $\textbf{s}_N, \textbf{m}_N$ decay sufficiently fast $\xi_d (\textbf{s}_N \vee  \textbf{m}_N)  = o(1)$ and one of two alternative conditions hold. (1) Bounded basis. There exist $\bar{B}< \infty$ so that $\sup_{x \in \mathcal{X}} \|p(x)\|_{\infty} \leq \bar{B}$ and $ \sqrt{N}\sqrt{d} \textbf{m}_N \textbf{s}_N = o(1)$. (2) Unbounded basis. There exist $\kappa, \gamma \in [1, \infty], \quad 1/\kappa + 1/\gamma=1$ so that   $ \sqrt{N}\sqrt{d} \textbf{m}_{N,2\gamma} \textbf{s}_{N,2\kappa} = o(1)$.  
	The functions in $M_N$ and $S_N$ are bounded uniformly over their domain:
	$$\sup_{\mu \in M_N} \sup_{(d,z) \in \mathrm{R} \bigtimes \mathcal{Z}  }   | \mu (d,z) |\vee \sup_{s \in \mathcal{S}_N}   \sup_{ (d,z) \in  \mathrm{R} \bigtimes \mathcal{Z} }   \max\{  s (d|z) , s^{-1}(d|z) \}  < \bar{\mathcal{C}}.$$ 
	
\end{assumption}

\begin{corollary}[Asymptotic Theory for Conditional Average Partial Derivative]
\label{thrm:CAPD}
Suppose Assumption \ref{ass:fsrate3} holds. Then,  the orthogonal signal $Y(\eta)$, given by Equation \eqref{eq:mdsig}, satisfies Assumption \ref{ass:smallbias}. Then, the statements of Theorems \ref{thrm:OS}-\ref{cor:confbands} hold for Conditional Average Partial Derivative.

\end{corollary}

\section{Empirical Application: Inference on Conditional
Average Elasticity}
\label{sec:empapp}
To show the immediate usefulness of the method, we consider an important problem of inference on structural derivatives.  We apply our methods to study the household demand for gasoline, a question studied in \cite{HausmanNewey},  \cite{Stoker}, \cite{YatchewNo} and  \cite{Blundell}. These papers estimated the demand function and the average price elasticity for various demographic groups. The dependence of the price elasticity on the household income was highlighted in \cite{Blundell}, who have estimated the elasticity by low, middle, and high-income groups and found its relationship with income to be non-monotonic. To gain more insight into this question, we estimate the average price elasticity as a function of income and provide simultaneous confidence bands for it.

The data for our analysis are the same as in  \cite{YatchewNo}, coming from the National Private Vehicle Use Survey, conducted by Statistics  Canada between October 1994 and September 1996. The data set is based on fuel purchase diaries and contains detailed information about fuel prices, fuel consumption patterns, vehicles, and demographic characteristics.  We employ the same selection procedure as in   \cite{YatchewNo} and \cite{CherChetQuant}, focusing on a sample of the households with non-zero licensed drivers, vehicles, and distance driven which leaves us with 5001 observations.

The object of interest is the average predicted percentage change in the demand due to a unit percentage change in the price, holding the observed demographic characteristics fixed, conditional on income. In the context of Example \ref{CAPD}, this corresponds to the conditional average derivative \begin{align*} g(x) &= \Ep [\partial_{d} \mu(D,Z) | X=x], \\
\mu(d,z) &= \Ep [Y|D=d,Z=z], 
\end{align*}
where $Y$ is the logarithm of  gas consumption, $D$ is  the logarithm of  price per liter, $X$ is log income, and $Z$ are the observed subject characteristics such as household size and composition, distance driven, and the type of fuel usage, including income. We use the orthogonal signal $Y(\eta)$ of \eqref{eq:capdsig}.

The choice of the estimators in the first and the second stages is as follows. To estimate the conditional expectation function $  \mu(d,z) $ and its partial derivative  $ \partial_{d} \mu(d,z) $, we consider a linear model 
$$ \mu(d,z) = b(d,z)' \omega,$$
where the basis function $b(d,z)$ includes price, price squared, income, income squared, their interactions with 28  time, geographical, and household composition dummies. All in all, we have 91 explanatory variables. We estimate the coefficient vector $ \omega$  using  Lasso with the penalty level chosen as in \cite{BCH} and plug the estimate $ \widehat{\omega}$ into the expression for the derivative:
$$ \partial_{d} \mu(d,z)  = \partial_{d} b(d,z)' \omega $$
to estimate $\partial_{d} \mu(d,z)$.

To estimate the conditional density $ s(d|z) $, we consider a model:
$$ D = l(Z) + U, \quad U \perp Z,$$
where $l(z) = \Ep [D|Z=z]$ is the conditional expectation of price variable  $D$ given covariates $Z$, and $U$ is an independent  continuously distributed shock with univariate density $\phi(\cdot)$. Under this assumption, the log density $\partial_{d} \log s(d|z)$ equals to $$\partial_{d} \log s(d|z) = \dfrac{\phi'(d-l(z) )}{\phi(d-l(z))}.$$  We estimate $\phi(u): \mathbb{R} \rightarrow \mathbb{R}^{+}$ by an adaptive kernel density estimator of \cite{PortnoyKoenker} with Silverman choice of bandwidth. Finally, we plug in the estimates of $  \mu(d,z)$, $ \partial_{d} \mu(d,z)$ , $s(d|z)$ into the Equation \eqref{eq:signal} to get an estimate of $\widehat{Y}$ and estimate $g(x)$ by least squares series regression of $\widehat{Y}$ on $X$. We try both polynomial basis function and B-splines to construct technical regressors.

Figures \ref{fig:hd} and \ref{fig:hdhh}  report the estimate of the target function (the black line), the pointwise (the dashed blue lines) and the uniform confidence (the solid blue lines) bands for the average price elasticity conditional on income, where the significance level $\alpha = 0.05$. The panels of Figure  \ref{fig:hd} correspond to different choices of the first-stage estimates of the nuisance functions $\mu(d,z)$ and $s(d|z)$ and dictionaries of technical regressors. The panels of Figure  \ref{fig:hdhh} correspond to the subsamples of large and small households and to different choices of the dictionaries.

The summary of our empirical findings based on Figure \ref{fig:hd} and \ref{fig:hdhh}  is as follows. We find the elasticity to be in the range $(-1,0)$ and significant for majority of income levels. The estimates based on $B$-splines (Figures \ref{fig:lassokernelsp}, \ref{fig:forestkernelsp}) are monotonically increasing in income, which is intuitive. The estimates based on polynomial functions are non-monotonic in income. For every algorithm in Figure \ref{fig:hd} we cannot reject the null hypothesis of constant price elasticity for all income levels: for each estimation procedure, the uniform confidence bands contain the constant function.  Figure \ref{fig:hdhh} shows the average price elasticity conditional on income for small and large households.  For the majority of income levels, we find large households to be more price elastic than the small ones, but the difference is not significant at any income level.

 To demonstrate the relevance of demographic data $Z$ in the first stage estimation, we also show the average predicted effect of the price change on the gasoline consumption (in logs), without accounting for the covariates in the first stage. In particular, this effect equals to $\Ep [\partial_{d} \mu(D,X)|X=x]$, where $\mu(d,x) = \Ep [Y|D=d,X=x]$ is the conditional expectation of gas consumption given income and price.  Figure \ref{fig:noz} shows this predictive effect, approximated by the polynomials of degree $d \in \{1,2\}$, conditional on income. By contrast to the results in Figure   \ref{fig:hd}, the slope of the polynomial of degree $d=1$ has a negative relationship between income and price elasticity, which present evidence that the demographics strongly confound the relationship between income and price elasticity.

\FloatBarrier

\begin{figure}
	\begin{subfigure}[b]{.4\textwidth}
		\includegraphics[width=\textwidth]{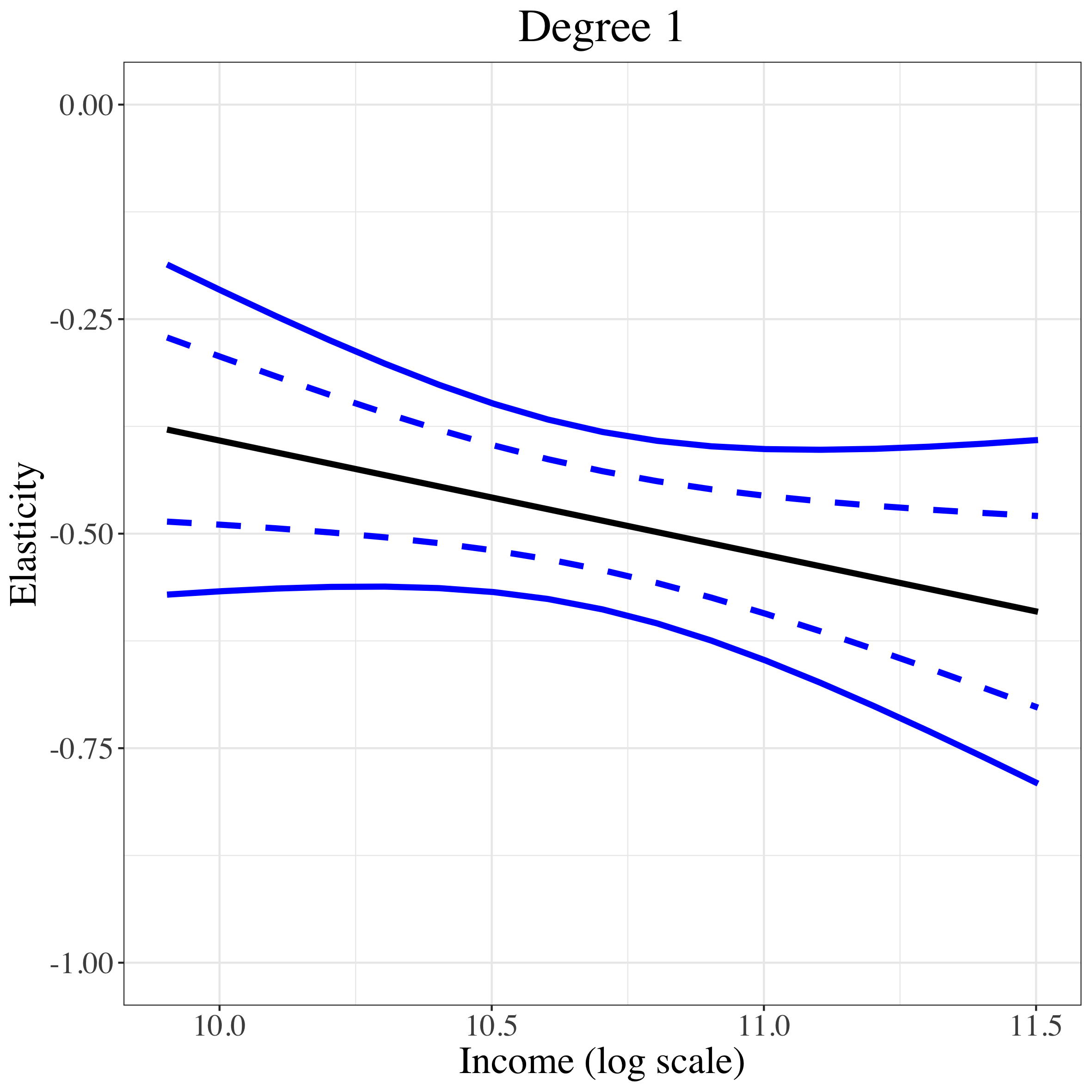}
		\caption{Polynomial degree $d=1$}
		\label{fig:level1}
	\end{subfigure}
	\centering
	\begin{subfigure}[b]{.4\textwidth}
		\includegraphics[width=\textwidth]{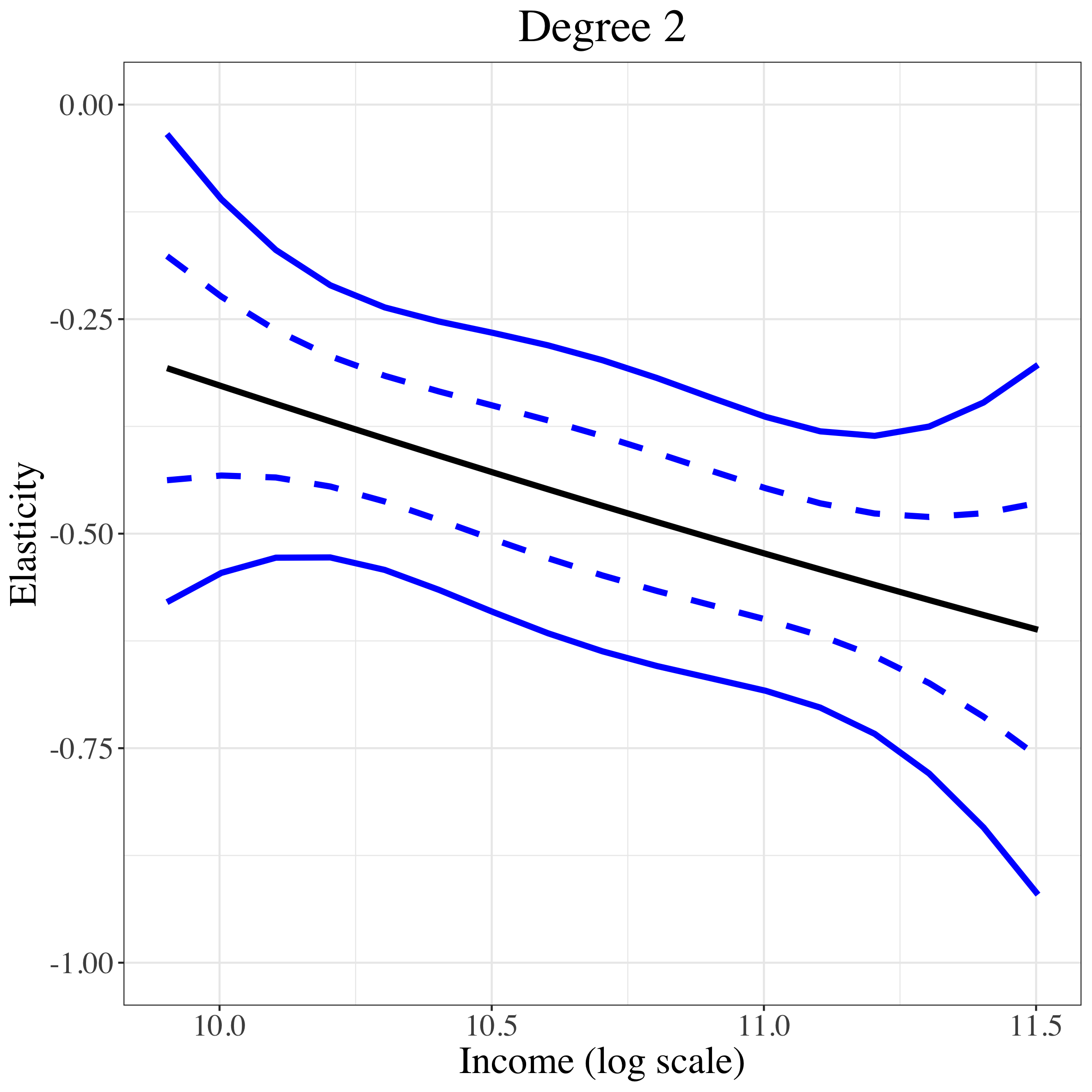}
		\caption{Polynomial degree $d=2$}
		\label{fig:level2}
	\end{subfigure}
	\caption{$95 \%$ confidence bands for the best linear approximation of the average price elasticity conditional on income without accounting for the demographic controls in the first stage. The black line is the estimated function, the dashed blue lines and the solid blue lines are the pointwise and the uniform confidence bands.  The estimation algorithm has three steps: (1) first-stage estimation of the conditional expectation function $\mu(d,x) = \Ep [Y|D=d,X=x] $, (2) second-stage estimation of the conditional density $ s(d|x)$, and (3) third-stage estimation of the target function $g(x)$ by least squares series. Step 1 is performed using least squares series regression using polynomial functions $\{1, x, \dots, x^q\}, q=3$ whose power $q$ is chosen by cross-validation out of $\{1,2,3\}$. Step 2 is performed by kernel density estimator with the Silverman choice of bandwidth.  Step 3 is performed using  polynomial functions $\{1, x, \dots, x^d\}$ and is shown for $d=1$ and $d=2$. Uniform confidence bands are based on $B=200$ repetitions of weighted (Bayes) bootstrap algorithm, described in \cite{NewOls}.  }
		\label{fig:noz}
\end{figure}

\begin{figure}[h]
\begin{subfigure}[b]{.4\textwidth}
		\includegraphics[width=\textwidth]{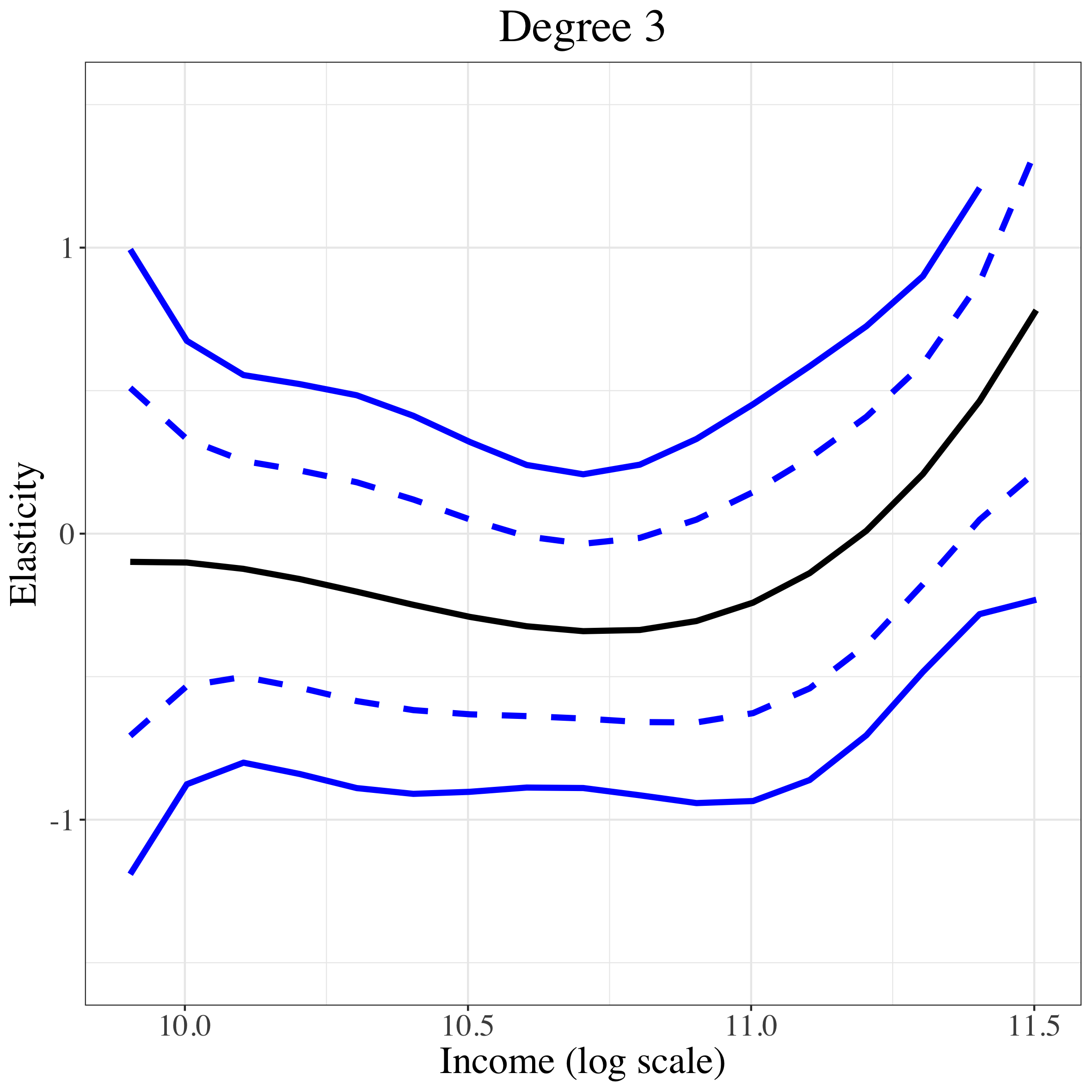}
		\caption{Step 2: $l(z)$ is estimated by Lasso. Step 3: polynomials of degree $3$. }
		\label{fig:lassokernelpoly}
	\end{subfigure}
	\centering
	\begin{subfigure}[b]{.4\textwidth}
		\includegraphics[width=\textwidth]{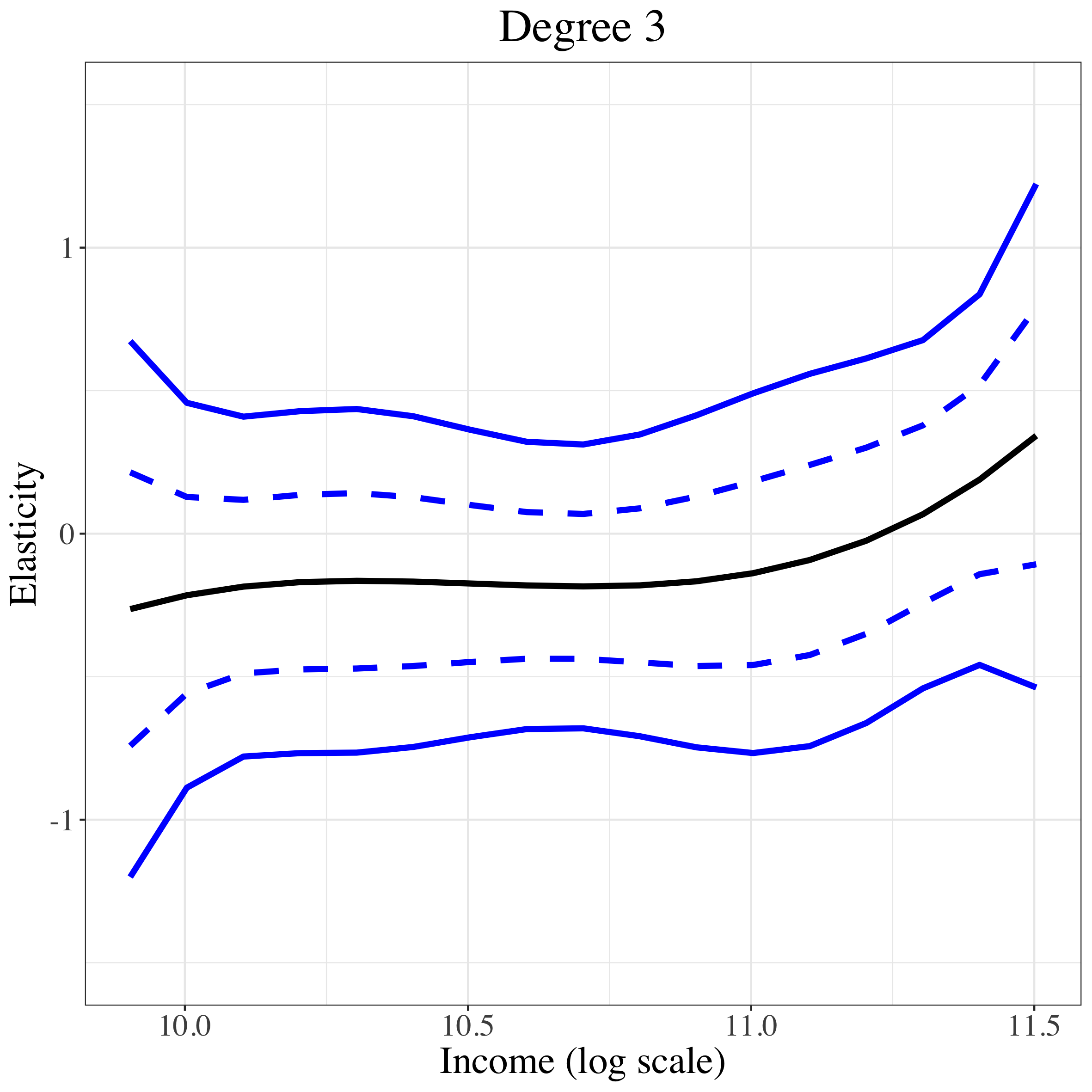}
		\caption{Step 2: $l(z)$ is estimated by random forest. Step 3: polynomials of degree $3$.}
		\label{fig:forestkernelpoly}
	\end{subfigure}
	\centering
		\begin{subfigure}[b]{.4\textwidth}
		\includegraphics[width=\textwidth]{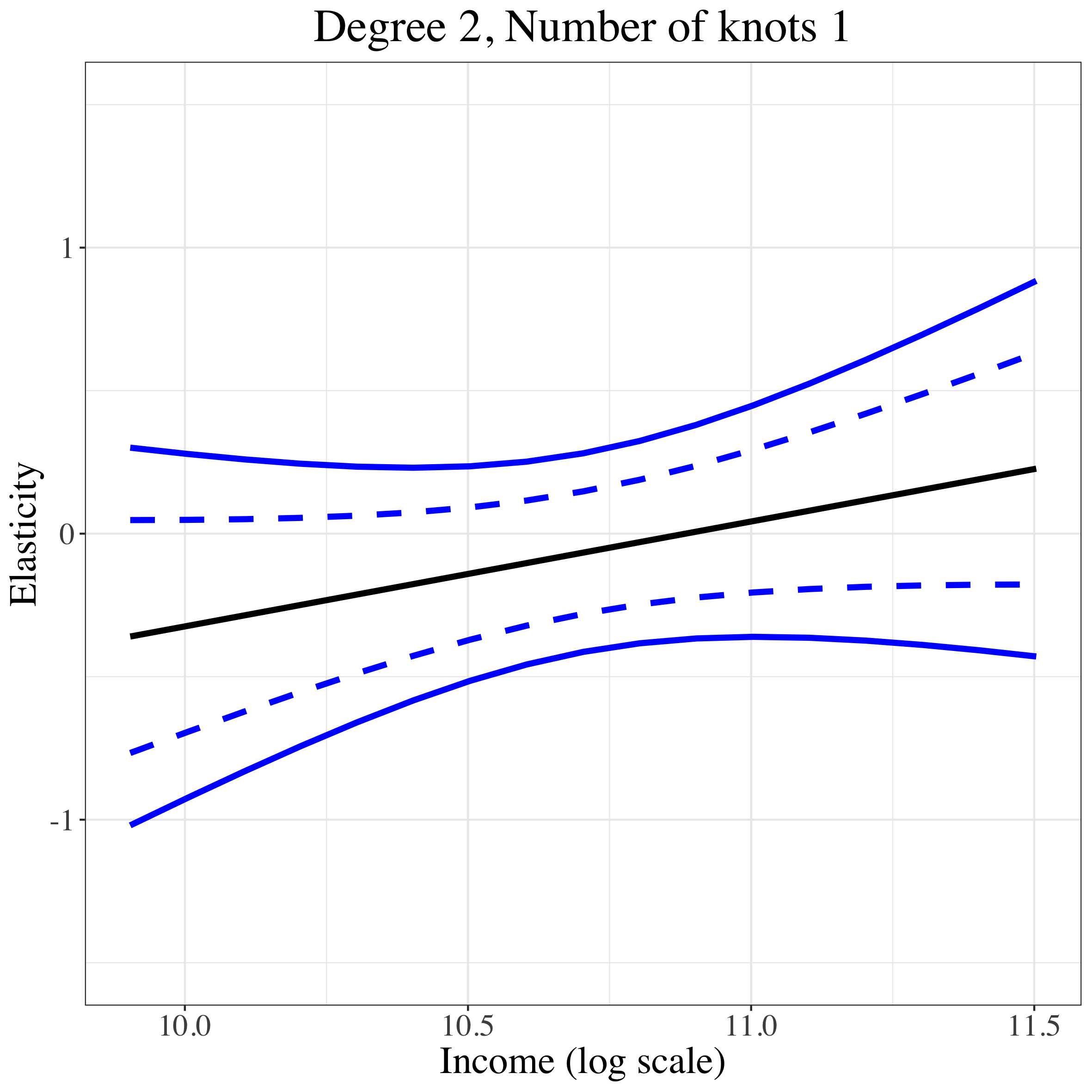}
		\caption{Step 2: $l(z)$ is estimated by Lasso. Step 3: $B$-splines of order $2$ with $1$ knot.}  
		\label{fig:lassokernelsp}
	\end{subfigure}
	\centering
	\begin{subfigure}[b]{.4\textwidth}
		\includegraphics[width=\textwidth]{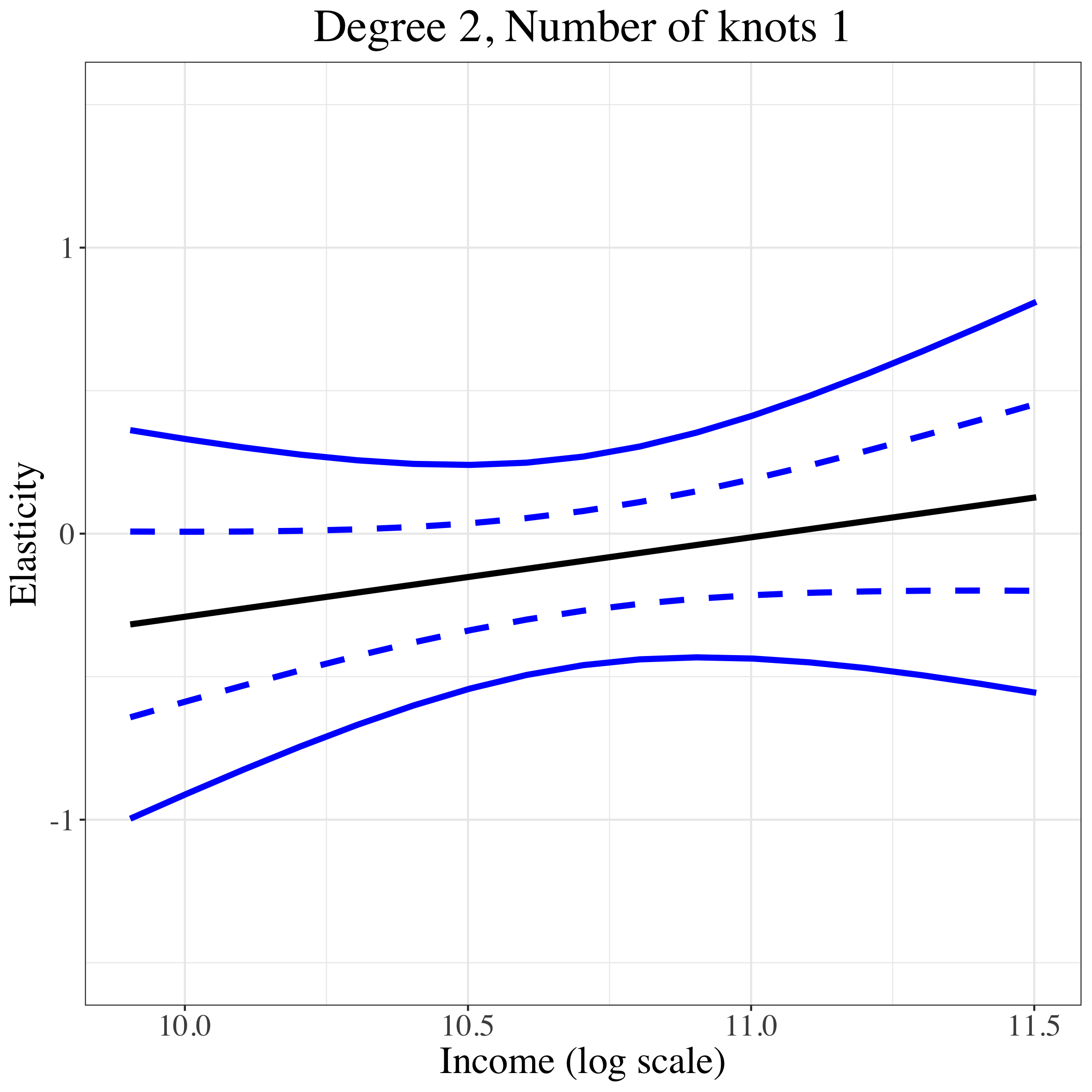}
		\caption{Step 2: $l(z)$ is estimated by random forest. $B$-splines of order $2$ with $1$ knot. }
		\label{fig:forestkernelsp}
	\end{subfigure}
	\caption{ $95 \%$ confidence bands for the best linear approximation of the average price elasticity conditional on income with accounting for the demographic controls in the first stage. The black line is the estimated function, the dashed(solid) blue lines are the pointwise (uniform) confidence bands.   The estimation algorithm has three steps: (1) first-stage estimation of the conditional expectation function $\mu(d,z) $, (2) second-stage estimation of the conditional density $ s(d|z)$, and (3) third-stage estimation of the target function $g(x)$ by least squares series. Step 1 is performed using Lasso with standardized covariates and the penalty choice $\lambda =2.2  \sqrt{n} \widehat{\sigma} \Phi^{-1} (1- \gamma/2p)$, where $\gamma= 0.1/\log n$ and $\widehat{\sigma}$ is the estimate of the residual variance. Step 2 is performed by estimating the regression function of $l(z)=\Ep [D|Z=z]$ and estimating the density $ \phi(d-l(z))$ of the residual $d-l(z)$ by adaptive kernel density estimator of \cite{PortnoyKoenker} with the Silverman choice of bandwidth. The regression function $l(z)$ is estimated lasso (\ref{fig:lassokernelpoly}, \ref{fig:lassokernelsp}) and random forest (\ref{fig:forestkernelpoly}, \ref{fig:forestkernelsp}). Step 3 is performed using B-splines of order $2$ with the number of knots equal to one (\ref{fig:lassokernelsp}, \ref{fig:forestkernelsp}) and  polynomial functions of order $3$.  (\ref{fig:lassokernelpoly}, \ref{fig:forestkernelpoly}). Uniform confidence bands are based on $B=200$ repetitions of weighted (Bayes) bootstrap algorithm, described in \cite{NewOls}.  }
	\label{fig:hd}
\end{figure}

\begin{figure}[h]
\begin{subfigure}[b]{.4\textwidth}
		\includegraphics[width=\textwidth]{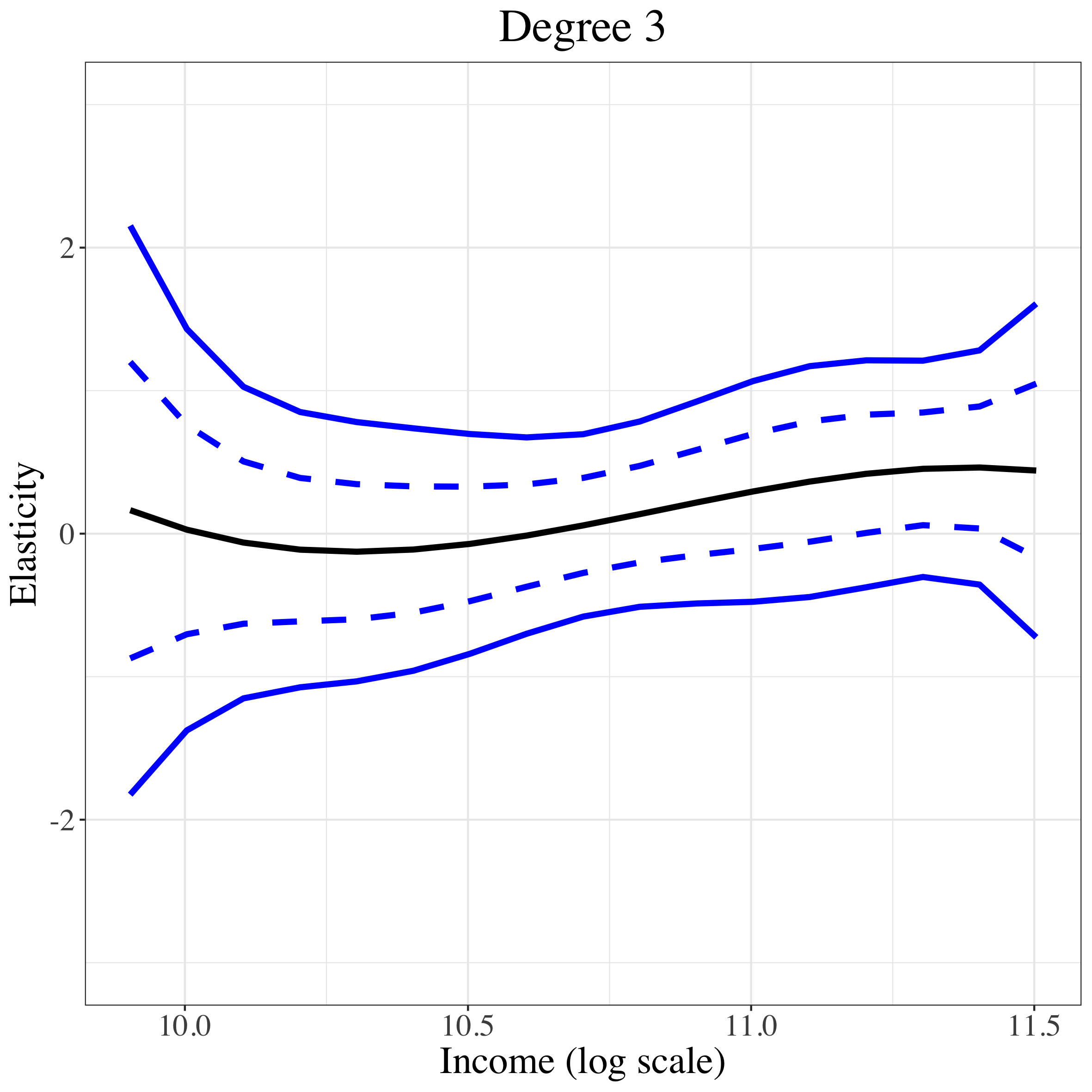}
		\caption{Large Households, Polynomials of degree $3$. }
		\label{fig:largelassokernelpoly}
	\end{subfigure}
	\centering
	\begin{subfigure}[b]{.4\textwidth}
		\includegraphics[width=\textwidth]{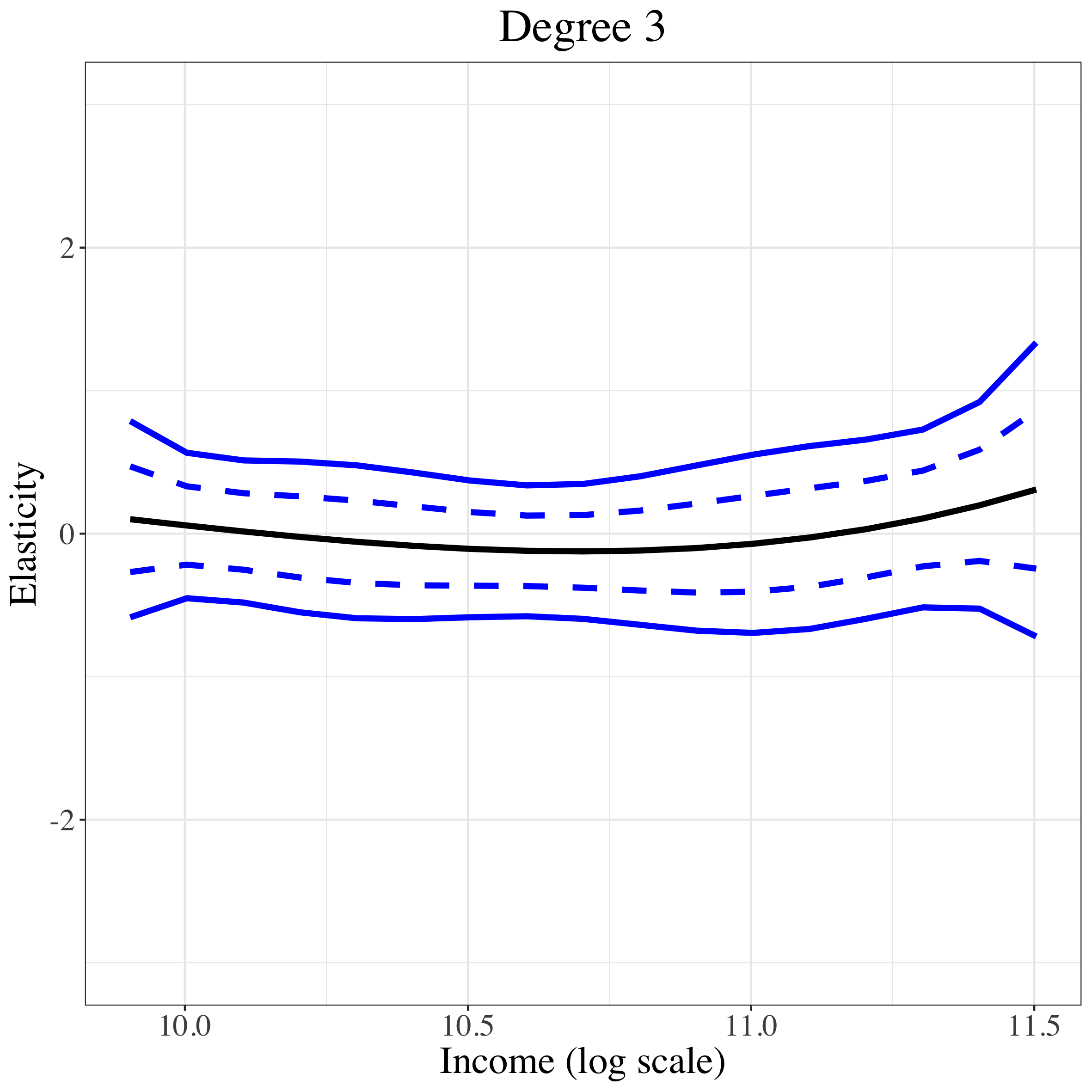}
		\caption{Small Households, Polynomials of degree $3$. }  
		\label{fig:smalllassokernelpoly}
	\end{subfigure}
	\begin{subfigure}[b]{.4\textwidth}
		\includegraphics[width=\textwidth]{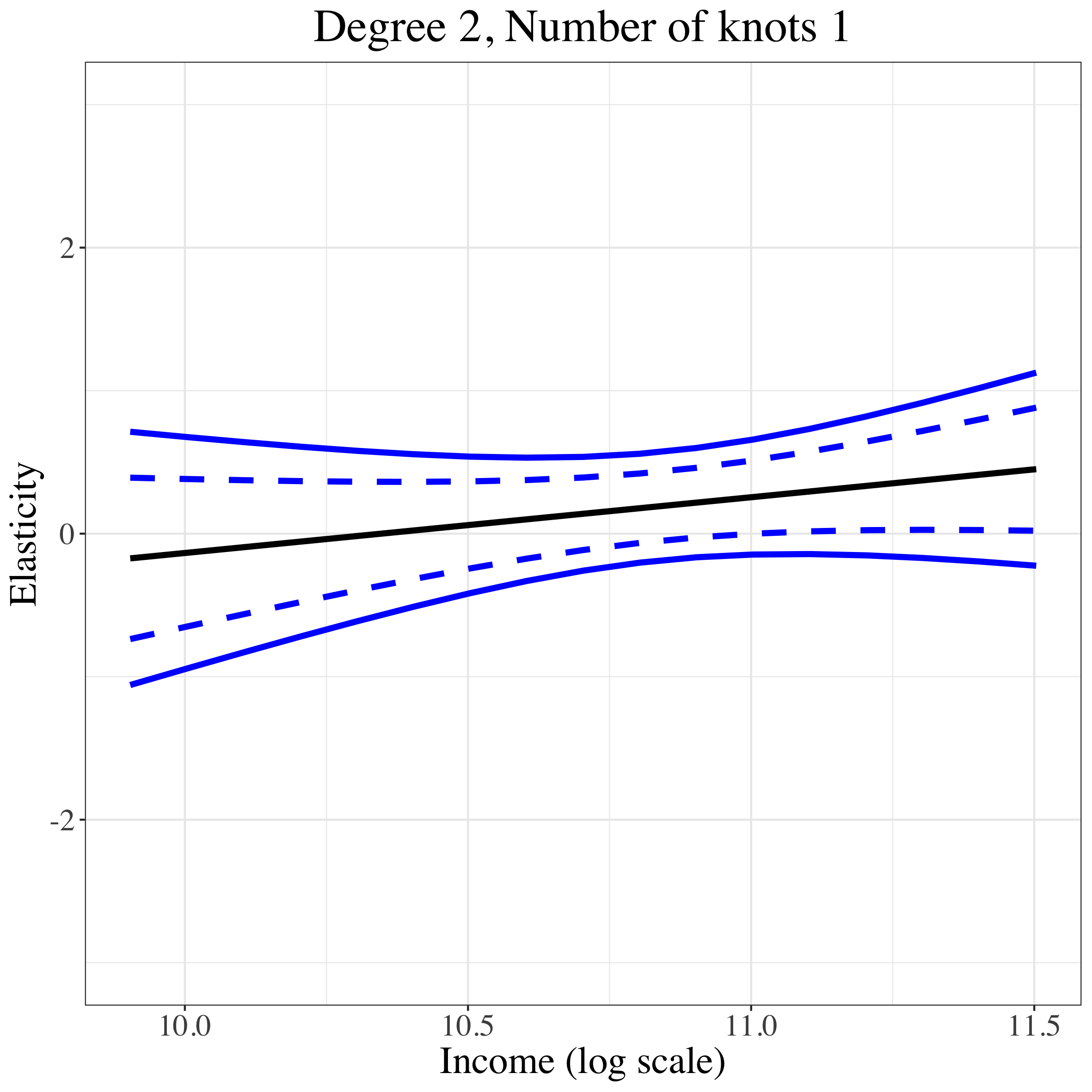}
		\caption{Large Households, B-splines of degree $2$ with $1$ knot. }
		\label{fig:largelassokernelsp}
	\end{subfigure}
	\centering
	\begin{subfigure}[b]{.4\textwidth}
		\includegraphics[width=\textwidth]{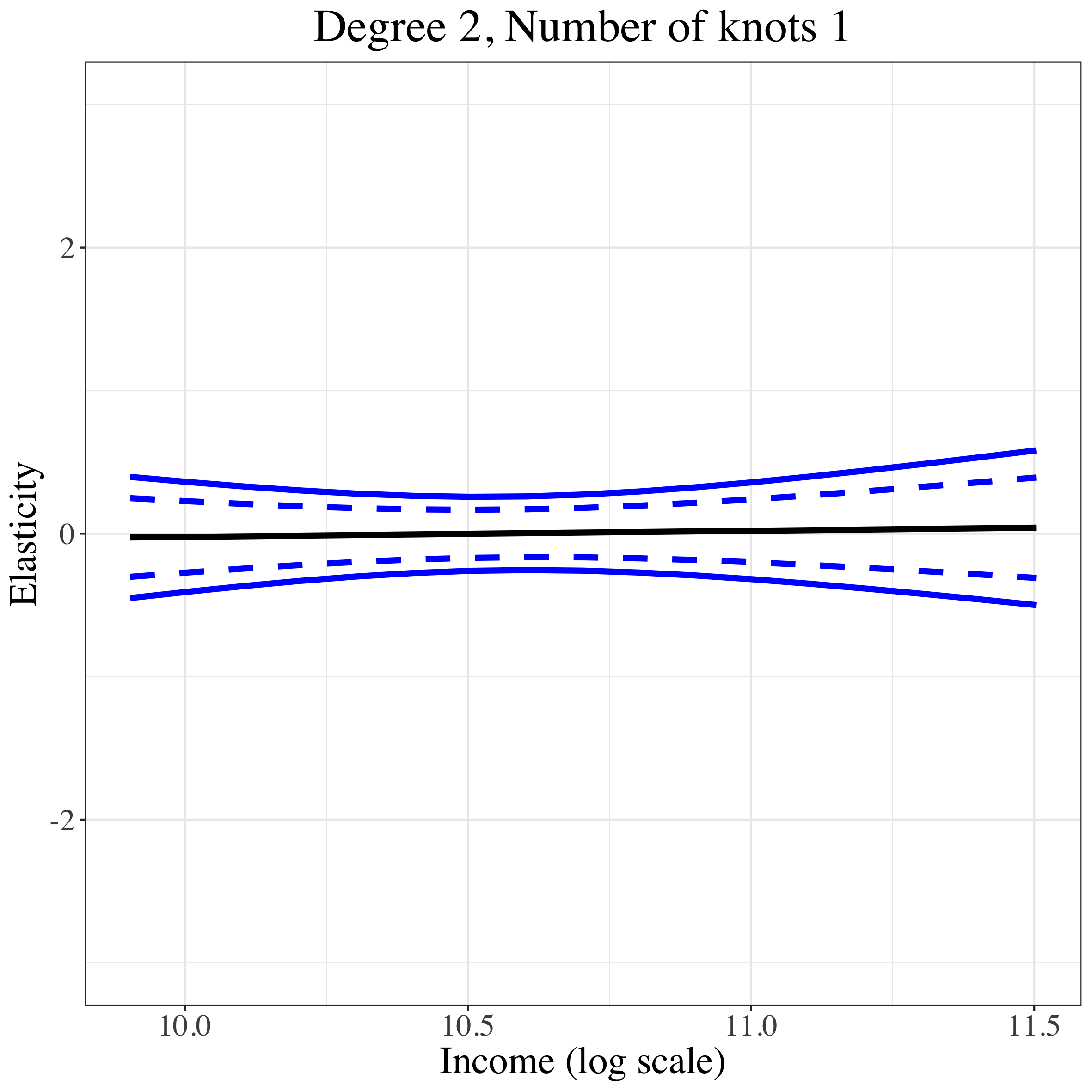}
		\caption{Small Households, B-splines of degree $2$ with $1$ knot. }
		\label{fig:smalllassokernelsp}
	\end{subfigure}
	\caption{$95 \%$ confidence bands for the best linear approximation of the average price elasticity conditional on income  with accounting for the demographic controls in the first stage by household size. The black line is the estimated function, the dashed(solid) blue lines are the pointwise (uniform) confidence bands.    The estimation algorithm has three steps: (1) first-stage estimation of the conditional expectation function $\mu(d,z) $, (2) second-stage estimation of the conditional density $ s(d|z)$, and (3) third-stage estimation of the target function $g(x)$ by least squares series. Step 1 is performed using Lasso with standardized covariates and the penalty choice $\lambda =2.2  \sqrt{n} \widehat{\sigma} \Phi^{-1} (1- \gamma/2p)$, where $\gamma= 0.1/\log n$ and $\widehat{\sigma}$ is the estimate of the residual variance. Step 2 is performed by estimating the regression function of $l(z)=\Ep [D|Z=z]$ and estimating the density $ \phi(d-l(z))$ of the residual $d-l(z)$ by adaptive kernel density estimator of \cite{PortnoyKoenker} with the Silverman choice of bandwidth. The regression function $l(z)$ is estimated lasso. Step 3 is performed using B-splines of order $2$ with the number of knots equal to one (\ref{fig:largelassokernelsp}, \ref{fig:smalllassokernelsp})  and using non-orthogonal polynomial functions of degree $3$  (\ref{fig:largelassokernelpoly}, \ref{fig:smalllassokernelpoly}). Uniform confidence bands are based on $B=200$ repetitions of weighted (Bayes) bootstrap algorithm, described in \cite{NewOls}. }
	\label{fig:hdhh}
\end{figure}

\FloatBarrier


\bibliography{/Users/virasemenova/Desktop/my_new_bibtex}
\bibliographystyle{apalike}


\newpage
\section*{Online Supplement}
    \begin{abstract}

   This document contains supplementary materials for the article "Debiased Machine Learning of Conditional Average Treatment Effects and Other Causal Functions" submitted to the Econometrics Journal. Appendix A contains proofs.   Appendix B contains verification of high-level conditions.

    \end{abstract}
    
\section*{Appendix A: Proofs }
\renewcommand{\theequation}{A.\arabic{equation}}
\renewcommand{\thelemma}{A.\arabic{lemma}}
\renewcommand{\thesection}{A}
\setcounter{equation}{0}
\setcounter{section}{0}
\setcounter{lemma}{0}
\medskip
\label{sec:otherproofs}

\paragraph{Glossary.} For two sequences of random variables  $a_N, b_N, N \geq 1: a_N \lesssim_{P}  b_n$ means    $a_N = O_{P} (b_N)$. For two sequences of numbers $a_N, b_N, N \geq 1$, $a_N \lesssim  b_N$ means $a_N = O (b_N)$. Let $a \wedge b = \min \{ a, b\}, a \vee b = \max \{ a, b\} $. The $\ell_2$ norm of a vector is denoted by $\| \cdot \|$, the $\ell_1$ norm is denoted by $\| \cdot \|_1$, the $\ell_{\infty}$ is denoted by $\| \cdot \|_{\infty}$, and $\ell_0$ is denoted by $\| \cdot \|_{0}$. Given a vector $\delta \in \mathbb{R}^p$ and a set of indices $T \subset \{ 1,...,p\}$, we denote by $\delta_T$ the vector in $\mathbb{R}^p$ for which $\delta_{Tj} = \delta_j , j \in T$ and $\delta_{Tj} = 0, j \not \in T$. Let $\xi_d := \sup_{x \in \mathcal{X}} \| p(x) \| = \sup_{x \in \mathcal{X}}  (\sum_{j=1}^d p_j(x)^2)^{1/2} $. For a matrix $Q$, let $\|Q\|$ be the maximal eigenvalue of $Q$. For a random variable $V$, let $\| V \|_{P,q}:= (\int |V|^q d P)^{1/q}$. The random sample $(V_i)_{i=1}^N$ is a sequence of independent copies of a random element $V$ taking values in a measurable space $(\mathcal{V}, \mathcal{A}_{\mathcal{V}})$  according to a probability law $P$. Let $p_i:=p(X_i)$ and $r_i:= r_g(X_i)$. Let $w_X(x)$ be the marginal density of $x$, $w_{X,Z}(x,z)$ be the joint density of $x,z$, $w_Z(z)$ be the marginal density of $z$, and  $s(x|z)$ be the conditional density of $x$ given $z$. 

Define an event $\mathcal{E}_N := \{ \widehat{\eta}_{k} \in T_N \quad \forall k \in [K] \}$, such that the nuisance parameter estimate $\widehat{\eta}_k$ belongs to the realization set $T_N$ for each fold $k \in [K]$. By union bound, this event holds w.h.p. $$\Pr (\mathcal{E}_N ) \geq 1- K \epsilon_N = 1-o(1).$$ 
For a given partition $k$ in $\{1,2, \dots, K\}$, define the partition-specific averages $$\Enk f(V_i) := \dfrac{1}{n} \sum_{i \in J_k} f(V_i), \quad \Gnk f(V_i) := \dfrac{1}{\sqrt{n}} \sum_{i \in J_k} [f(V_i) - \int f(v) dP (v)] .$$

\subsection{Useful Technical Lemmas}

\begin{lemma}[LLN for Matrices]
Let $Q_i, i = 1,2, \dots, N$ be i.n.i.d symmetric non-negative $d \bigtimes d$-matrices
such that $d \geq e^2$ and $\| Q_i \| \leq M$ a.s.  Let $Q= \frac{1}{N} \sum_{i=1}^N \Ep Q_i$ denote average value of population covariance matrices. Then, the following statement holds:
 $$ \Ep \| \widehat{Q}- Q \| \leq  \sqrt{\frac{M (1 + \| Q \|) \log N}{N}} .$$
 In particular, if $Q_i = p_i p_i' $ with $\| p_i \| \leq \xi_d$, then
 $$ \Ep \| \widehat{Q}- Q \| \leq  \sqrt{\frac{\xi_d^2 (1 + \| Q \|) \log N}{N}} .$$
\label{lem:LLN}
\end{lemma}

Proof can be found in \cite{RudVersh}. The statement is taken from \cite{NewOls}.

\begin{lemma}[Conditional Convergence Implies Unconditional]
\label{lem:cond}
Let $\{X_m\}_{m \geq 1}$ and $\{Y_m\}_{m \geq 1}$ be sequences of random vectors. (i) If for $\epsilon_m \rightarrow 0, \Pr( \| X_m \| > \epsilon_m | Y_m ) \rightarrow_{P} 0,$ then  $\Pr( \| X_m \| > \epsilon_m) \rightarrow 0$. In particular, this occurs if  $\Ep [\| X_m \|^q /\epsilon_m^q |Y_m]  \rightarrow_P 0$ for some $q \geq 1$, by Markov inequality. (ii) Let $\{A_m\}_{m \geq 1}$ be a sequence of positive constants. If $\|X_m \| = O_{P} (A_m) $   conditional on $Y_m$, namely, that for any $\ell_m \rightarrow \infty, \Pr( \| X_m \| >  \ell_m A_m|Y_m) \rightarrow_P 0$ , then $ X_m = O_{P} (A_m)$ unconditionally, namely, that for any $\ell_m \rightarrow \infty, \Pr( \| X_m \| >  \ell_m A_m) \rightarrow 0$.
\end{lemma}
Lemma \ref{lem:cond} is a restatement of Lemma 6.1  of \cite{chernozhukov2016double}.

\subsection{Proofs of Lemmas from Section \ref{sec:asymp}}

\begin{lemma}[No Effect of First Stage Error]
\label{lem:smallerror}
Suppose Assumption \ref{ass:smallbias} holds. Then, 
\begin{align}
 \sqrt{N} \| \EN p_i[Y_{i}(\widehat{\eta}) -Y_i(\eta_0)] \| &= O_{P} (   B_N + \Lambda_N) = o(1).
\end{align}
\end{lemma}
\begin{proof}[Proof of Lemma \ref{lem:smallerror}]
The total sum is the sum of $K$ partial sums, defined for each partition $(J_k)_{k=1}^K$:
 \begin{align*}
	\EN p_i[Y_{i}(\widehat{\eta}) -Y_i(\eta_0)] &= \frac{1}{K} \sum_{k=1}^K  \Enk p_i[Y_{i}(\widehat{\eta}) -Y_i(\eta_0)] - \Ep \big[ p_i[Y_{i}(\widehat{\eta}) -Y_i(\eta_0)] | (V_i)_{i \in J_k^c} \big]  \\
	 &+  \Ep \big[ p_i[Y_{i}(\widehat{\eta}) -Y_i(\eta_0)] | (V_i)_{i \in J_k^c} \big]  \\
	 &=: \frac{1}{K} \sum_{k=1}^K [ I_{1,k} + I_{2,k}].
\end{align*}
Conditional on $(V_i)_{i \in J_k}$, the estimator   $\widehat{\eta} = \widehat{\eta}_{k}$ is non-stochastic. On the event $\mathcal{E}_N $, \begin{align*}
\Ep[ \| \sqrt{n} I_{1,k} \|^2 |  \mathcal{E}_N, (V_i)_{i \in J_k^c} ]&\leq \Ep [ \| p_i[Y_{i}(\widehat{\eta}) -Y_i(\eta_0)] \|^2 | \mathcal{E}_N,  (V_i)_{i \in J_k^c} ] \\
&\leq  \sup_{  \eta \in T_N} \Ep  \| p_i[Y_{i}(\eta) -Y_i(\eta_0)] \|^2 = \Lambda_N^2 \tag{Assumption \ref{ass:smallbias}}
\end{align*}
Hence, $\sqrt{n} I_{1,k} = O_{P} (\Lambda_N)$ by Lemma \ref{lem:cond}. To bound $I_{2,k}$, recognize that on the event $\mathcal{E}_N $
\begin{align*}
	\Ep [ \| \sqrt{n} I_{2,k} \|  |  \mathcal{E}_N, (V_i)_{i \in J_k^c} ]&\leq \sup_{\eta \in T_N} \sqrt{n} \| \Ep  p_i (Y_{i}(\eta) -Y_i(\eta_0)) | (V_i)_{i \in J_k} \| \\
	&\leq  \sup_{\eta \in T_N} \sqrt{n} \| \Ep  p_i (Y_{i}(\eta) -Y_i(\eta_0)) \| \leq B_N.
\end{align*}
Therefore, $\sqrt{n}I_{2,k} = O_{P}(B_N)$.
\end{proof}

\begin{proof}[Proof of Lemma \ref{lem:pointwise} (a)]

\begin{align*}
	\| \widehat{\beta}- \beta_0\| &= \| \widehat{Q}^{-1} \EN p_i Y_i(\widehat{\eta}) - \beta_0 \| 
	\leq  \| \widehat{Q}^{-1}\| \| \EN p_i[Y_i(\widehat{\eta}) - Y_i(\eta_0) ]\| + \| \widehat{Q}^{-1} \| \| \EN p_i [Y_i(\eta_0) - p_i'\beta_0 ] \|  \\
	&=: S_1 + \| \widehat{Q}^{-1}\| \| \EN p_i [Y_i(\eta_0) - g(X_i) ] \|   +  \| \widehat{Q}^{-1}\| \| \EN p_i [g(X_i) - p_i'\beta_0 ] \| \\
	&=: S_1 + \| \widehat{Q}^{-1}\|  \| \EN p_i U_i \| + \| \widehat{Q}^{-1}\| \| \EN p_i r_i\|.
\end{align*}
The effect of sampling error (second term) is bounded as
\begin{align}
\| \EN p_i U_i \| &\lesssim_{P} (\Ep \big[ \| \EN p_i U_i \|^2 \big] )^{1/2} \leq (\Ep U_i^2 p_i' p_i / N )^{1/2} 
\lesssim \overline{\sigma} \sqrt{d /N}.
\end{align}
The effect of approximation error (third term) is bounded as
\begin{align}
\| \EN p_i r_i \| &\lesssim_{P} (\Ep \big[ \|\EN p_i r_i \|^2 \big] )^{1/2} \lesssim l_d r_d \sqrt{\dfrac{\Ep \| p_i \|^2 }{N}}  = l_d r_d \sqrt{\dfrac{d}{N}}.
\end{align}
Alternatively, the third term can be also bounded as
\begin{align}
\| \EN p_i r_i \| &\lesssim_{P} (\Ep \big[ \| \EN p_i r_i \|^2 \big] )^{1/2} \leq \xi_d \sqrt{\dfrac{\Ep r_i^2 }{N}}  = \xi_d r_d/ \sqrt{N}.
\end{align}
With high probability, $\| \widehat{Q}^{-1}\| \leq  2 \| Q^{-1} \| \leq 2 /C_{\min}   $. Lemma \ref{lem:smallerror} implies \begin{align*}
 \| S_1 \| & \leq   \| \widehat{Q}^{-1} \|  \| \EN p_i [Y_i(\widehat{\eta}) - Y_i(\eta_0)] \| \leq  \| \widehat{Q}^{-1} \|  \big[ B_N/\sqrt{N} +   \Lambda_N/\sqrt{N}  \big] =  o_{P} (N^{-1/2}).
\end{align*}
\end{proof}

\begin{proof}[Proof of Lemma \ref{lem:pointwise} (b)]
Decomposing  \begin{align*}
\sqrt{N}  [ \EN p_i Y_i(\widehat{\eta}) - \widehat{Q} \beta_0]&=  \sqrt{N} \EN p_i[Y_i(\widehat{\eta}) - Y_i (\eta_0)]  +\sqrt{N} \EN p_i[Y_i(\eta_0) - p_i'\beta_0] \\
  &=  \sqrt{N}\EN p_i[Y_i(\widehat{\eta}) - Y_i (\eta_0)]  + \GN p_i[Y_i(\eta_0) - p_i'\beta_0]   \\
  &=  \sqrt{N}\EN p_i[Y_i(\widehat{\eta}) - Y_i (\eta_0)]  + \GN p_i U_i +  \GN p_i r_i 
  \end{align*} 
   we obtain:
\begin{align}
\label{eq:decomp}
\sqrt{N}\alpha' &[\widehat{\beta} - \beta_0] = \sqrt{N}\alpha' \widehat{Q}^{-1}  [ \EN p_i Y_i(\widehat{\eta}) - \widehat{Q} \beta_0] \\
 &= \sqrt{N} \alpha' \widehat{Q}^{-1}[  \EN p_i [Y_i(\widehat{\eta})  - Y_i(\eta_0)]]  +  \alpha'\widehat{Q}^{-1} \GN [ p_i [r_i + U_i]] \nonumber  \\
&=  \alpha'Q^{-1}  \GN [ p_i [r_i + U_i]] \nonumber \\
&+\alpha^\top [\widehat{Q}^{-1} - Q^{-1}]  \GN [ p_i (U_i+ r_i)] +  \sqrt{N} \alpha^\top Q^{-1} \EN p_i [Y_i(\widehat{\eta}) - Y_i(\eta_0) ]  \\
&+ \sqrt{N}\alpha^\top [\widehat{Q}^{-1} - Q^{-1}] \EN p_i [Y_i(\widehat{\eta}) - Y_i(\eta_0) ]\\
&=: \alpha'Q^{-1}  \GN [ p_i [r_i + U_i]]  + R_{1,N} (\alpha) \nonumber,
\end{align} where the remainder term $R_{1,N}$ is $$ R_{1,N} (\alpha) = I_1 + I_2 +I_3. $$ Decomposing $I_1 $ into sampling and approximation parts gives
\begin{align*}
I_1 &= \sqrt{N} \alpha^\top [\widehat{Q}^{-1} - Q^{-1}]  \EN p_i U_i + \alpha^\top [\widehat{Q}^{-1} - Q^{-1}]  \GN p_i r_i =: I_{1,a} + I_{1,b}. 
\end{align*}
Observe that $	\Ep [I_{1,a} | (X_i)_{i=1}^N] = 0$.  As shown in the proof of Lemma 4.1 in \cite{NewOls}, its second moment is bounded as
\begin{align*}
	\Ep [I_{1,a}^2 | (X_i)_{i=1}^N] &\leq  \alpha^\top  [\widehat{Q}^{-1} - Q^{-1}] Q [\widehat{Q}^{-1} - Q^{-1}] \alpha  \overline{\sigma}^2   \lesssim_{P} \dfrac{\xi_d^2 \log N}{N}  \overline{\sigma}^2,
\end{align*}
which implies $I_{1,a} = O_{P}(  \sqrt{\dfrac{\xi_d^2 \log N}{N } }  )$. Likewise,  by the proof of Lemma 4.1 in \cite{NewOls}, 
\begin{align*}
|I_{1,b} | &\lesssim_{P} \sqrt{\dfrac{\xi_d^2 \log N}{N} }[ l_d r_d \sqrt{d} \wedge \xi_d r_d].
\end{align*}
By Lemma \ref{lem:smallerror},
 \begin{align*}
|I_{2}| & \lesssim_{P} \| Q \|^{-1} \| \sqrt{N} \EN  p_i [Y_i(\widehat{\eta}) - Y_i(\eta_0) ] \| \lesssim_{P} 1/C_{\min} [\Lambda_N + B_N] = o (1),\\ 
|I_{3}| & \lesssim_{P} \|  \alpha \| \| [\widehat{Q}^{-1} - Q^{-1}] \| \|\sqrt{N}  \EN p_i [Y_i(\widehat{\eta}) - Y_i(\eta_0) ] \|  \lesssim_{P}  \sqrt{\dfrac{\xi_d^2 \log N}{N}} [ \Lambda_N +B_N]  = o(1).
\end{align*} 
Therefore, with probability approaching one, $$\sup_{\eta \in T_N} \| R_{1,N} (\alpha) \| \lesssim_{P}  \Lambda_N +B_N + 
 \sqrt{\dfrac{\xi_d^2 \log N}{N} } \left( 1 + \min \bigg\{ l_d r_d \sqrt{d} ,  \xi_d r_d \bigg\} \right). $$

\end{proof}

\begin{proof}[Proof of Lemma \ref{lem:uniform} (a)] Define 
\begin{align*}
I_1(x) &= \alpha(x) ^\top [\widehat{Q}^{-1} - Q^{-1}]  \GN [ p_i (U_i+ r_i)], \\
I_2(x) &= \sqrt{N} \alpha(x) ^\top Q^{-1} \EN  p_i [Y_i(\widehat{\eta}) - Y_i(\eta_0) ],\\
I_3(x) &= \sqrt{N}\alpha(x)^\top [\widehat{Q}^{-1} - Q^{-1}] \EN  p_i [Y_i(\widehat{\eta}) - Y_i(\eta_0) ].
\end{align*}
Decompose  $$\sqrt{N} \alpha(x)^\top (\widehat{\beta} - \beta_0)  = \sqrt{N} \alpha(x)^\top Q^{-1} \GN   [ p_i [r_i + U_i]] + I_1(x) + I_2(x) + I_3(x)$$ Step 1. According to Step 1 of the proof of Lemma 4.2 in \cite{NewOls}, the following bound  holds 
\begin{align*}
    \sup_{x \in \mathcal{X}} | \alpha(x) ^\top [\widehat{Q}^{-1} - Q^{-1}]  \GN p_iU_i  | \lesssim_{P} N^{1/m} \sqrt{\dfrac{\xi_d^2 \log N}{N}}
\end{align*}
by Assumptions \ref{ass:growth}, \ref{ass:merror}, \ref{ass:basis}.

Step 2. According to Step 2 of the proof of Lemma 4.2 in \cite{NewOls}, the following bound  holds 
$$ \sup_{x \in \mathcal{X}} | \alpha(x) ^\top [\widehat{Q}^{-1} - Q^{-1}]  \GN  p_i  r_i | \lesssim_{P}  \sqrt{\dfrac{\xi_d^2 \log N}{N} } l_d r_d \sqrt{d}$$
by Assumptions \ref{ass:growth} and \ref{ass:approx}. Steps 1 and 2 give the bound on $ \sup_{x \in \mathcal{X}} I_1(x)$.

Step 3. The bound on $I_2(x) +I_3(x)$ follows from Lemma \ref{lem:smallerror}
\begin{align*}
 \sup_{x \in \mathcal{X}}  | I_2 (x) | &\leq  \sup_{x \in \mathcal{X}}  \| \alpha(x) \| C_{\text{min}}^{-1} \| \E_N p_i [Y_i(\widehat{\eta}) - Y_i(\eta_0) ] \| \\
 &\lesssim_{P} B_N + \Lambda_N = o (1) \\
  \sup_{x \in \mathcal{X}}  | I_3 (x) | &\leq  \sup_{x \in \mathcal{X}}  \| \alpha(x) \| \| Q^{-1} - \widehat{Q}^{-1}  \| \| \E_N p_i [Y_i(\widehat{\eta}) - Y_i(\eta_0) ] \| \\
 &\lesssim_{P} \sqrt{\dfrac{\xi_d^2 \log N}{N}} (B_N + \Lambda_N) =o (1).
 \end{align*}

Lemma \ref{lem:uniform}(b) follows from Theorem 4.3 in \cite{NewOls}. 
\end{proof}

\subsection{Proofs of Main Results from Section \ref{sec:asymp}}

\begin{proof}[Proof of Theorem \ref{thrm:OS}]
Proof of Theorem \ref{thrm:OS}  follows from Lemma \ref{lem:pointwise} and the proof of  Theorem 4.2 in \cite{NewOls}.

\end{proof}

\begin{proof}[Proof of Theorem \ref{thrm:ult}]
Proof of Theorem \ref{thrm:ult} follows from Lemma \ref{lem:uniform} and the proof of   Theorem 4.4 in \cite{NewOls}.

\end{proof}

\begin{proof}[Proof of Theorem \ref{lem:matrix}]
Define the estimator of the matrix $\widehat{\Sigma}$ as
\begin{align*}
\widehat{\Sigma}:= \E_N p_ip_i^\top  \widehat{U}^2_i := \E_N p_ip_i^\top  ( Y_i(\widehat{\eta}) - p(X_i)'\widehat{\beta})^2,
\end{align*} 
and $\widehat{\Omega}:= \widehat{Q}^{-1} \widehat{\Sigma}  \widehat{Q}^{-1}$. Decompose the difference $\widehat{\Sigma} - \Sigma$ as
\begin{align*}
\widehat{\Sigma} - \Sigma &=  \E_N p_i p_i'[\widehat{U}_i^2 - \{ U_i + r_i \}^2] +\E_N p_i p_i'\{ U_i + r_i \}^2 - \Sigma =: I_1 + I_2. 
\end{align*}
In \cite{NewOls}, it was shown that $I_2 \lesssim_{P} (N^{1/m}  +l_d r_d ) \sqrt{\frac{\xi_d^2 \log N}{N}}$. Therefore, it suffices to prove a bound on $I_1$. Recognize that
\begin{align*}
\widehat{U}_i &=  Y_i(\eta_0) - p_i'\beta_0 + p_i'\beta_0 - p_i'\widehat{\beta} +  Y_i(\widehat{\eta}) -Y_i(\eta_0) \\
&= (U_i + r_i) + (p_i' (\beta_0- \widehat{\beta})) +  (Y_i(\widehat{\eta}) -Y_i(\eta_0)) =  a + b + c.
\end{align*}
Plugging $|(a+b+c)^2 - a^2|  = |2 a(b+c) + (b+c)^2 | \leq 2 (c^2 + b^2  + |ac|+ |ab| )$ into $I_1$ gives:
\begin{align*}
I_1:= \| \E_N p_i p_i'[\widehat{U}_i^2 - \{ U_i + r_i \}^2] \| &\leq 2 \| \E_N p_i p_i'(Y_i(\widehat{\eta}) -Y_i(\eta_0))^2    \| \\
&+ 2 \| \E_N p_i p_i' ( U_i + r_i) (Y_i(\widehat{\eta}) -Y_i(\eta_0)) \|  \\
&+ 2 \| \E_N p_i p_i'(p_i'(\widehat{\beta} - \beta_0))^2    \|  \\
&+ 2 \| \E_N p_i p_i' ( U_i + r_i) p_i'(\widehat{\beta} - \beta_0) \|    \\
&:= 2(C^2 + AC + B^2 + AB).
\end{align*}
For each $k \in \{1,2,\dots, K \}$, conditional on $(V_j)_{j \in J_k^c}$ and the event $\mathcal{E}_N$,
\begin{align*}
 \max_{i \in J_k} ( Y_i(\widehat{\eta}) - Y_i (\eta_0))^2 &\lesssim_{P} \E [ \max_{i \in J_k} ( Y_i(\widehat{\eta}) - Y_i (\eta_0))^2 | (V_j)_{j \in J_k^c}, \mathcal{E}_N ] + o_{P} (1) \\
 &\lesssim \sup_{\eta \in \mathcal{T}_N} \E  \max_{i \in J_k} ( Y_i(\eta) - Y_i (\eta_0))^2 = \kappa^2_n.
\end{align*}
Therefore, $$ \max_{i \in \{1,2,\dots, N\}} ( Y_i(\widehat{\eta}) - Y_i (\eta_0))^2 \lesssim_{P} K \kappa^2_n + o_{P} (1).$$
\begin{align*}
C^2 &\lesssim_{P} \max_{1 \leq i \leq N} ( Y_i(\widehat{\eta}) - Y_i (\eta_0))^2 \| \E_N p_i p_i'\| \\
&\lesssim_{P} \max_{1 \leq i \leq N} ( Y_i(\widehat{\eta}) - Y_i (\eta_0))^2 O_{P}(1) \lesssim_{P} \kappa_{n}^2 \asymp \kappa_N^2
\end{align*}
The bound on $AC$ can be seen as
\begin{align*}
AC &\lesssim_{P} \max_{1 \leq i \leq N} | Y_i(\widehat{\eta}) - Y_i (\eta_0) | \max_{1 \leq i \leq N} \{ |U_i| + |r_i|  \} \| \E_N p_i p_i'\| \\
&\lesssim_{P} \max_{1 \leq i \leq N} | Y_i(\widehat{\eta}) - Y_i (\eta_0) | \max_{1 \leq i \leq N} \{ |U_i| + |r_i|  \}  O_{P} (1) \\
&\lesssim_{P}  \kappa_N^1 ( N^{1/m} +l_d r_d).
\end{align*}
The bound on $B^2+AB$ is established in the  proof of Theorem 4.6 in \cite{NewOls}. It is
\begin{align*}
B^2 + AB \lesssim_{P} (N^{1/m}  +l_d r_d) \sqrt{\frac{\xi_d^2 \log N}{N}}.
\end{align*}
Collecting the bounds on $C^2 + AC + B^2 + BC$ yields:
\begin{align}
\label{eq:i1}
I_1 \leq 2(C^2 + AC + B^2 + AB) \lesssim_{P} (N^{1/m}  +l_d r_d) (\sqrt{\frac{\xi_d^2 \log N}{N}} + \kappa_N^1 ) + \kappa_N^2.
\end{align}
Thus, the statement of Theorem follows from the proof of Theorem 4.6 of \cite{NewOls} and (\ref{eq:i1}). Finally, (\ref{eq:lemma51}) follows from the proof of Lemma 5.1 in \cite{NewOls}. 
\end{proof}

\begin{proof}[Proof of Theorem \ref{cor:bootstrap}]
Theorem \ref{cor:bootstrap} coincides with Condition C.2(b) in \cite{CLR} with $\mathcal{V} = \mathcal{X}$ and 
\begin{align*}
    Z_N^{\star}(x) &= \dfrac{p(x)'\widehat{\Omega}^{1/2}}{\| p(x)'\widehat{\Omega}^{1/2} \|} \mathcal{N}^b_d, \quad Z_N^{*} (x)=\dfrac{p(x)'\Omega^{1/2}}{\| p(x)'\Omega^{1/2} \|} \mathcal{N}^b_d, \quad \delta_N=o(1/\log N)
\end{align*}
for some $\ell_N$ chosen below. Condition C.2 (b) holds by Lemma 5 of \cite{CLR}, whose assumptions we verify below. First, condition NS 1 (a) of Lemma 5 is verified by  Theorem \ref{thrm:ult} with $\bar{a}_N=\log N$. Second, the eigenvalues of matrix $\Omega$ are bounded below by the minimal eigenvalue of $\underbar{$\sigma$}^2 Q^{-1}$, which is bounded away from zero by Assumption \ref{ass:identification}. Finally, the last two conditions of C.2 are directly assumed in the theorem.  Therefore, all conditions of Lemma 5 are verified, and the statement of Theorem holds. 

%
\end{proof}
  
 We approximate the $t$-statistic process in \eqref{eq:tN} by the  following Gaussian coupling process conditional on the data:
\begin{align}
\label{eq:tnstar}
   \bigg\{  t^{*}_N(x) &= \dfrac{p(x)'\Omega^{1/2} \mathcal{N}_d/\sqrt{N}}{\sigma_N(x)}, x \in \mathcal{X} \bigg\}
\end{align}
\begin{lemma}[Strong Approximation  of the $T$-statistic process]
\label{lem:strong}
Suppose Assumptions \ref{ass:identification}-\ref{ass:matrix} hold. In addition, suppose the assumptions of Theorem \ref{thrm:ult} hold with $\bar{a}_N^{-1}=o(1)$. Then, 
\begin{align*}
    t_N(x) &= t^{*}_N(x) + o_{P}(\bar{a}_N) \text{ in } \ell^{\infty}(\mathcal{X}),
\end{align*}
where $t_N(x)$ is as in \eqref{eq:tN} and  $t^{*}_N(x)$ is as in \eqref{eq:tnstar}.
\end{lemma}
The statement of Lemma \ref{lem:strong} follows from the proof of Theorem 5.4 from \cite{NewOls}.  Specifically, the proof of Theorem 5.4 is invoked with $a_N = \bar{a}_N$, $\ell(w) = p(x)$,  $\mathcal{I} = \mathcal{X}$, Theorem \ref{lem:matrix}(b) (in place of Lemma 5.1), Lemma \ref{lem:uniform} (a) (in place of Lemma 5.2), Lemma \ref{lem:uniform} (b) (in place of Theorem 5.3) and  Theorem \ref{thrm:ult} (in place of Theorem 4.4), where in parentheses we refer to statements from  \cite{NewOls}.


\begin{lemma}[Strong Approximation of Gaussian Process Suprema]
\label{lem:strongsup}
Suppose Assumptions \ref{ass:identification}-\ref{ass:matrix} hold with $m \geq 4$. Assume that (i) $\bar{R}_{1N} + \bar{R}_{2N} \lesssim \log^{-1/2} (N)$, (ii) $$\xi_d \log^2 N/ N^{1/2-1/m} =o(1),$$ (iii) $1 \lesssim \underbar{$\sigma$}^2$, and (iv) $\sup_{x \in \mathcal{X}} \sqrt{N} | r(x) |/ \| p(x) \| = o(1/\sqrt{\log N})$. Then,
\begin{align*}
    \sup_{x \in \mathcal{X}} | t_N(x)| =_d \sup_{x \in \mathcal{X}} | t_N^{*} (x) | + o_{P} (1/\sqrt{ \log N}).
\end{align*}
\end{lemma}
The statement of Lemma \ref{lem:strongsup} follows from the proof of Theorem 5.5 from \cite{NewOls} with $a_N = \sqrt{\log N}$, $\ell(w)=p(x)$, $\mathcal{I}=\mathcal{X}$.  Specifically, the proof of Theorem 5.5 is invoked with $a_N = \bar{a}_N$, $\ell(w) = p(x)$,   $\mathcal{I} = \mathcal{X}$, Lemma \ref{lem:strong} (in place of Theorem 5.4), Lemma \ref{lem:uniform} (a) (in place of Lemma 5.2), where in parentheses we refer to statements from  \cite{NewOls}.

\begin{proof}[Proof of Theorem \ref{cor:confbands}]
The proof follows by invoking Lemma \ref{lem:strongsup} and the same steps as the proof
of Theorem 5.5 in \cite{NewOls}, replacing  the bound on $\| \widehat{\Omega} - \Omega \|$ by the bound from Theorem \ref{lem:matrix}. The validity of the critical value obtained by Gaussian bootstrap follows exactly the same steps as in the proof of Theorem 5.6 in \cite{NewOls}.

\end{proof}

\subsection{Proofs from Section \ref{sec:apps}}

\begin{lemma}[Maximal inequality for canonical $U$-statistics]
\label{lem:Chen}
Let $(V_i)_{i=1}^N$  be a sample of i.i.d random variables. Let $\widetilde{\tau}(\cdot, \cdot) : \mathcal{V} \bigtimes \mathcal{V} \rightarrow \mathcal{R}^d$ be a symmetric $d$-dimensional kernel function whose $m$'th coordinate is denoted by $\widetilde{\tau}_m(\cdot, \cdot)$. The kernel function is canonical:
$$ \E[ \widetilde{\tau}_m(V_1, V_2) ] =0 \quad \text{ for any } m  \in \{1,2,\dots, d\}.$$
Assume that $\E | \widetilde{\tau}_m(V_1, V_2) |<\infty $ for all $m \in \{1,2,\dots, d\}$.  Let  $$\bar{V}_N = \frac{1}{N(N-1)} \sum_{1 \leq i \neq j \leq N} \widetilde{\tau}(V_i, V_j),$$ $$M =  \max_{1 \leq i \neq j \leq N}  \max_{1 \leq m \leq d} | \widetilde{\tau}_m (V_i, V_j) |, $$ and $D_q= \max_{1 \leq m \leq d} (\E |\widetilde{\tau}_m (V_1, V_2) |^q)^{1/q}, \quad q>0$. If $2 \leq d \leq \exp (bN)$ for some constant $b>0$, then there exists an absolute constant $K>0$ such that
\begin{align*}
\E [ \|\bar{V}_N\|_{\infty} ] \leq K (1 + b^{1/2})  \bigg\{ \dfrac{\log d}{N }D_2  + \bigg(\dfrac{\log d}{N } \bigg)^{5/4} D_4 + \bigg(\dfrac{\log d}{N } \bigg)^{3/2} \| M \|_4  \bigg\}
\end{align*}
\end{lemma}
Lemma \ref{lem:Chen} is a restatement of Theorem 5.1 of \cite{Chen}. In this paper, $d$ is less than $N$. Therefore, the R.H.S. is bounded by $K (1 + b^{1/2}) (D_2 \vee D_4 \vee \| M \|_4 )\dfrac{\log d}{N }$.

For a symmetric function $f: \mathcal{V}^r \rightarrow \mathcal{R}$, denote
\begin{align*}
    P^{r-k} f(x_1,x_2, \dots, x_k) = \E f(x_1, x_2, \dots, x_k, X_{k+1}, \dots, X_{r} )
\end{align*}
and let $P^0 f =f$.  Let $\mathcal{F}$ be a class of symmetric  functions and denote  $\| \cdot \|_{\mathcal{F}} = \sup_{f \in \mathcal{F}} | \cdot  |$. Assume that there exists an envelope $F$ for $\mathcal{F}$ such that $P^r F^2 < \infty$. Consider the associated $U$-process
\begin{align*}
U_N^{(r)}(f) &= \dfrac{1}{| I_{N,r} |} \sum_{(i_1,i_2,\dots, i_r ) \in I_{N,r}} f(V_{i,1}, V_{i,2}, \dots, V_{i,r}), \quad f \in \mathcal{F},
\end{align*}
where $ I_{N,r}$ is a collection of permutations. For each $ k \in \{1,2,\dots, r\}$, the Hoeffding projection (with respect to $P$) is defined by
\begin{align*}
\pi_k(f) (x_1,x_2,\dots x_k) &:= (\delta_{x_1} - P) \dots  (\delta_{x_k} - P) P^{r-k} f,
\end{align*}
where $\delta_{x}$ is the delta-function at $x$. Let $\sigma_k$ be any positive constant such that $\sup_{f \in \mathcal{F}} \| P^{r-k} f  \|_{P^k,2} \leq  \sigma_k \leq   \| P^{r-k} F  \|_{P^k,2} $ whenever $  \| P^{r-k} F  \|_{P^k,2}>0$ and let $\sigma_k=0$ otherwise. Define $M_k$ as
\begin{align*}
M_k &= \max_{1 \leq i \leq [N/k]} (P^{r-k}F) (V^{ik}_{i(k-1)+1}),
\end{align*}
where $V^{ik}_{i(k-1)+1} = (V_{i(k-1)+1},V_{i(k-1)+2}, \dots V_{ik})$. Finally, let $J_k(\delta)$ be a uniform entropy integral 
\begin{align*}
    J_k(\delta) := \int_0^{\delta} \sup_{Q} \bigg[ (1 + \log N(P^{r-k} \mathcal{F}, \| \cdot \|_{Q,2}, \| P^{r-k} F \|_{Q,2} ) ) \bigg]^{k/2} d \tau,
\end{align*}
where $P^{r-k} \mathcal{F} = \{ P^{r-k} f: f \in \mathcal{F} \} $ and 
$\sup_{Q}$ is taken over all finitely discrete distributions on $\mathcal{V}^r$. A function class $\mathcal{F}$ with envelope $F$ is said to be VC-type
with characteristics $(A, v$) if $\sup_{Q} N(\mathcal{F}, \| \|_{Q,2}, \epsilon \| F \|_{Q,2} ) \leq (A/\epsilon)^v $ for  all $0< \epsilon \leq 1$, where 
$\sup_{Q}$ is taken over all finitely discrete distributions on $\mathcal{V}^r$.

\begin{lemma}[Local maximal inequalities for $U$-processes]
\label{lem:ChenKato}
Suppose $J_k(1) < \infty$ for $k=1,2,\dots, r$. Let $\delta_k = \sigma_k/\| P^{r-k} F \|_{P^k,2}$ for $k=1,2,\dots, r$. Then, 
\begin{align*}
N^{k/2} \E [ \| U_N^{(k)} (\pi_k f) \|_{\mathcal{F}}] \lesssim J_k(\delta_k) \| P^{r-k} F \|_{P^k,2} + \dfrac{J_k^2(\delta_k) \| M_k \|_{\mathrm{P},2}}{\delta_k^2 \sqrt{N}}.
\end{align*}

\end{lemma}
Lemma \ref{lem:ChenKato} is a restatement of Theorem 5.1 of \cite{ChenKato}.

\begin{lemma}[Local maximal inequalities for $U$-processes indexed by VC-type classes]
\label{cor:ChenKato}
Let $\mathcal{F}$  be VC-type with characteristics $A\geq e^{2(r-1)}/16 \vee e$ and $v \geq 1$. Then, for $k=1,2,\dots, r$, 
\begin{align*}
N^{k/2} \E [ \| U_N^{(k)} (\pi_k f) \|_{\mathcal{F}}] &\lesssim \sigma_k \bigg\{ v \log (A \| P^{r-k} F \|_{P^k,2}  / \sigma_k) \bigg\}^{k/2} +  \dfrac{\| M_k \|_{\mathrm{P},2}}{\sqrt{N}} \bigg\{ v \log (A \| P^{r-k} F \|_{P^k,2}  / \sigma_k)  \bigg\}^{k}.
\end{align*}
\end{lemma}
Lemma \ref{cor:ChenKato} is a restatement of Corollary 5.3. of \cite{ChenKato}.

\begin{lemma}[Useful Properties for Continuous Treatment Effects]
\label{lem:cte:properties}
Suppose Assumptions \ref{ass:identification} and \ref{ass:fsratecont} hold. Let $\alpha$ be a vector on a unit sphere $\mathcal{S}^{d-1} = \{ \alpha \in \mathcal{R}^d: \| \alpha \| =1 \}$ and let $v=Q^{-1} \alpha$. Then, the following properties hold for $\tau_1(z,\mu)$ in \eqref{eq:tau1}:
 \begin{enumerate}
 \item[(a)] The function $\mu^0(x,z)$ is bounded uniformly over its domain:
$$\sup_{(x,z) \in \mathcal{X} \bigtimes \mathcal{Z}} |\mu^0(x,z)| \leq 2\bar{\mathcal{C}}.$$
 \item[(b)] The function $\tau_1(z; \mu)$ is linear in $\mu$ for each $z \in \mathcal{Z}$
\begin{align*}
\tau_1(z; \mu - \mu_0)= \tau_1(z; \mu) - \tau_1(z; \mu_0).
\end{align*}
		\item[(c)] For each $z \in \mathcal{Z}$  the following inequality holds:
		\begin{align*}
| v' \tau_1(z; \mu) | \leq  2 C_{\min}^{-1} \bar{\mathcal{C}}.
\end{align*}
\item[(d)] For any non-negative function $\phi(x,z)$, 
\begin{align*}
    \int_{z \in \mathcal{Z}}  \phi(x,z) w_Z(z) dz \lesssim \E [ \phi(X,Z) | X=x] \quad \text{ for any } x,
\end{align*}
and 
\begin{align*}
     \int_{x \in \mathcal{X}}  \bigg(\int_{z \in \mathcal{Z}}  \phi (x,z) w_Z(z) \bigg)^2 w_X(x) dz
     &\leq  \int_{x \in \mathcal{X}}  \int_{z \in \mathcal{Z}}  \phi^2(x,z) w_Z(z) w_X(x) dz dx \lesssim \E  \phi^2(X,Z). 
\end{align*}
\end{enumerate}
%
\end{lemma}

\begin{proof}[Proof of Lemma \ref{lem:cte:properties}]
Lemma \ref{lem:cte:properties}(a). Since $\sup_{(x,z) \in \mathcal{X} \bigtimes \mathcal{Z}} |\mu(x,z)| \leq \bar{\mathcal{C}}$, 
\begin{align*}
    |\mu^0(x,z)|= |\mu(x,z) - \E \mu(x,Z)| \leq |\mu(x,z)| + | \E \mu(x,Z) | \leq 2\bar{\mathcal{C}}.
\end{align*}
Lemma \ref{lem:cte:properties}(b) follows from the definition of $\tau_1(z;\mu)$. 
Lemma \ref{lem:cte:properties}(c) follows from the following inequality:
	\begin{align*}
| v' \tau_1(z; \mu) | &= | \E_{X} [(v' p(X)) \cdot (\mu^0(X;z) )) ] | \\
&\leq  v' Q v (\E_{X} [\mu^0(X;z)]^2)^{1/2} \leq  2 C_{\min}^{-1} \bar{\mathcal{C}},
\end{align*}
where the last inequality follows from Lemma \ref{lem:cte:properties}(a) and $v'Qv = \alpha' Q^{-1} \alpha$. 
Lemma \ref{lem:cte:properties}(d) holds since
\begin{align*}
    \int_{z \in \mathcal{Z}}  \phi(x,z) w_Z(z) dz &= \E \bigg[ \phi(X,Z) \dfrac{w_X(X)}{s(X|Z)} \bigg| X=x \bigg] \leq  \bar{\mathcal{C}}^2 \E [ \phi(X,Z)  | X=x]\\
      \int_{x \in \mathcal{X}}  \bigg(\int_{z \in \mathcal{Z}}  \phi (x,z) w_Z(z)\bigg)^2 w_X(x) dz &\leq      \int_{x \in \mathcal{X}}  \int_{z \in \mathcal{Z}}   \phi^2 (x,z) w_Z(z) w_X(x) dz dx \\
     &\leq \int_{x \in \mathcal{X}}\int_{z \in \mathcal{Z}}  \phi^2 (x,z) w_{Z,X} (x,z) \dfrac{w_X(x)}{s(x|z)} dx dz  \\
     &\leq \bar{\mathcal{C}}^2  \E \phi^2(X,Z). 
\end{align*}

\end{proof}

\begin{proof}[Proof of Lemma \ref{thrm:conttreat}]

\textbf{ Step 1}. Let us show that $Y(\eta)$ given by Equation \eqref{eq:contoutcome} satisfies Assumption \ref{ass:smallbias}. Bayes rule implies that $w(z |x) = \dfrac{s_0(x|z) w_Z(z)}{w_X(x) }$,
where $w_X(x)=w_0(x)$ denotes the true values the marginal density of $X$, and $w(x)$ is a candidate value of the nuisance parameter. Likewise, let $s(x|z)$ and $s_0(x|z)$ be the true and the candidate value of the conditional density. Observe that 
\begin{align*}
    \E[Y(\eta_0) | X=x] &=  \int_{z \in \mathcal{Z}}  \mu_0(x,z) dP_Z(z) \\
    \E[Y(\eta) | X=x] &= \E[\E[ Y(\eta)| X, Z] | X=x] =  \int_{z \in \mathcal{Z}}  \mu(x,z) dP_Z(z) + \E \bigg[ \dfrac{w(x)(\mu_0(x,z) - \mu(x,z))}{s(x|z)} \bigg| X=x \bigg].
\end{align*}
Plugging both statements in  $ \E(Y(\eta) - Y(\eta_0) | X=x)$ gives
\begin{align*}
    \E \left(Y(\eta) - Y(\eta_0) | X=x\right) &= \int_{z \in \mathcal{Z}}  \dfrac{\mu_0(x,z) - \mu(x,z) }{s(x|z)/w(x) } \dfrac{s_0(x|z)}{w_X(x)} w_Z(z) d z \\
    &+  \int_{z \in \mathcal{Z}}   [ \mu(x,z) -  \mu_0(x,z)] w_Z(z) d z   \\
    &= \dfrac{w(x)}{w_X(x)} \int_{z}   \dfrac{(\mu_0(x,z) - \mu(x,z)) (s_0(x|z) - s(x|z)) }{s(x|z) } w_Z(z) d z
  \\
    &+  \dfrac{w(x) - w_0(x)}{w_0(x)} \int_{z}  [\mu_0(x,z) - \mu(x,z)] w_Z(z) d z \\
    &+ \int_{z  \in \mathcal{Z}} [ \mu_0(x,z) - \mu(x,z) ] w_Z(z) d z   +\int_{z \in \mathcal{Z}}   [ \mu(x,z) -  \mu_0(x,z)] w_Z(z) d z  \\
    &=: S_1(x) + S_2(x) + 0.
\end{align*}
Recognize that 
\begin{align*}
\| \E p(X) (S_1(X) + S_2(X)) \| \leq  \| \E p(X) S_1(X) \| + \| \E p(X) S_2(X) \|.
\end{align*}
For each $j \in \{1,2, \dots, d\}$,  Cauchy-Scwartz inequality implies
$$ (\E  | p_j(X)  S_k(X) | )^2  \leq  \sup_{x \in \mathcal{X}} p_j^2(x) (\E |S_k(X)|)^2, \quad k \in \{1, 2 \}.$$
Therefore,
\begin{align*}
    \| \E p(X) (S_1(X) + S_2(X)) \|^2 &\leq 2 (\| \E p(X) S_1(X) \|^2 + \| \E p(X) S_2(X) \|^2) \\
    &\leq 2  \sum_{j=1}^d \sup_{x \in \mathcal{X}}  p_j(x)^2 ( (\E |S_1(X)|)^2 + (\E |S_2(X)|)^2)\\
    &\lesssim d ( (\E |S_1(X)|)^2 + (\E |S_2(X)|)^2)
\end{align*}
if the dictionary of basis function is bounded. Alternatively, 
\begin{align*}
    \| \E p(X) (S_1(X) + S_2(X)) \|^2 &\leq 2 (\| \E p(X) S_1(X) \|^2 + \| \E p(X) S_2(X) \|^2) \\
    &\lesssim d ( \E S_1^2(X)  + \E S_2^2(X) ),
\end{align*}
where the bounds on $(\E |S_k(X)|)^2$ and $\E S_k^2(X)$ are established in Steps 2 and 3. Finally, Assumption \ref{ass:fsratecont} implies the following inequalities:
\begin{align}
\label{eq:ineq1}
   \sup_{w \in \mathcal{W}_N} \sup_{x \in \mathcal{X}} \dfrac{w(x)}{w_0(x)} &\leq \bar{\mathcal{C}}^2, \\
   \sup_{s \in \mathcal{S}_N}  \sup_{(x,z) \in (\mathcal{X} \bigtimes \mathcal{Z})} s^{-1} (x|z) &\leq \bar{\mathcal{C}}, \label{eq:ineq2} \\
     \sup_{s \in \mathcal{S}_N}       \sup_{w \in \mathcal{W}_N} \sup_{(x,z) \in \mathcal{X} \bigtimes \mathcal{Z}}  w (x) s^{-1} (x,z) &\leq \bar{\mathcal{C}}^2. \label{eq:ineq3} 
\end{align}
Joint, conditional, and marginal densities are related as
 \begin{align}
 \label{eq:rel}
     s_0(x|z) w_Z(z) = w_{X,Z}(x,z) = w(z|x) w_X(x).
 \end{align}
\textbf{ Step 2. } The following bound holds:
\begin{align*}
    |S_1(x)| &\leq \sup_{x \in \mathcal{X}} \dfrac{w(x)}{w_0(x) } \sup_{(x,z) \in \mathcal{X} \bigtimes \mathcal{Z}} s^{-1}(x|z) \int_{z \in \mathcal{Z}} | \mu_0(x,z) - \mu(x,z) | |s_0(x|z) - s(x|z) |w_Z(z) dz \\
    &\leq^{i} \bar{\mathcal{C}}^3 \int_{z \in \mathcal{Z}} | \mu_0(x,z) - \mu(x,z) | |s_0(x|z) - s(x|z) | w_Z(z) dz \\
    &\leq^{ii} \bar{\mathcal{C}}^5  \E [ | \mu_0(X,Z) - \mu(X,Z) | |s_0(X|Z) - s(X|Z) | |X=x ],
\end{align*}
where $i$ holds by \eqref{eq:ineq1} and $ii$ holds by Lemma \ref{lem:cte:properties}(d). Therefore, $\E | S_1(X) | \lesssim \textbf{s}_N \textbf{m}_N$ and
$\E  S_1^2(X)  \lesssim \textbf{s}_{N,2\kappa} \textbf{m}_{N,2\gamma}$ for any $\kappa, \gamma \geq 1$ so that $1/\kappa+1/\gamma=1$.

 \textbf{ Step 3. } By \eqref{eq:ineq1} and \eqref{eq:rel},
\begin{align*}
    |S_2(x)| &\leq^{i}  \bar{\mathcal{C}} | w(x) -w_0(x)| \int_{z \in \mathcal{Z}} | \mu(x,z) - \mu_0(x,z) | w_Z(z) dz \\
    &\lesssim^{ii}  \bar{\mathcal{C}}^3 | w(x) -w_0(x)|  \E [| \mu(X,Z) - \mu_0(X,Z) | | X=x].
\end{align*}
where $i$ holds by \eqref{eq:ineq1} and $ii$ holds by Lemma \ref{lem:cte:properties}(d).
Therefore, $\E | S_2(X) | \lesssim \textbf{w}_N \textbf{m}_N$ and
$\E  S_2^2(X)  \lesssim \textbf{w}_{N,2\kappa} \textbf{m}_{N,2\gamma}$ for any $\kappa, \gamma \geq 1$ so that $1/\kappa+1/\gamma=1$.

\textbf{ Step 4. } To choose $\Lambda_N$, we decompose $Y(\eta) - Y(\eta_0)$ as follows
\begin{align*}
Y(\eta) - Y(\eta_0) &= (\mu_0(X,Z) - \mu(X,Z)) \left( \dfrac{w(X)}{s(X|Z)} \right) \\
&+ (Y - \mu_0(X,Z))\left( \dfrac{w(X)}{s(X|Z)} - \dfrac{w_0(X)}{s_0(X|Z)} \right)\\
&+ \int_{z \in Z} (\mu(X,z)  - \mu_0(X,z)) d P_{Z}(z) \\
&= M_1 + M_2  + M_3.
\end{align*} 
Recall that $\Lambda_N \leq \xi_d \sup_{\eta \in \mathcal{T}_N} (\E (Y(\eta) - Y(\eta_0))^2)^{1/2}$. Therefore, it suffices to bound $\E M_k^2$ for $k=1,2,3$. By \eqref{eq:ineq3},  $$\E M_1^2 = \E  (\mu_0(X,Z) - \mu(X,Z))^2  \sup_{(x,z) \in \mathcal{X} \bigtimes \mathcal{Z}}  \bigg(\dfrac{w(x)}{s(x|z)} \bigg)^2  \lesssim \textbf{m}_N^2 \bar{\mathcal{C}}^4.$$ By \eqref{eq:ineq1} and \eqref{eq:ineq2},
\begin{align*}
\E M_2^2 &\leq \overline{\sigma}^2  \E \bigg[ \dfrac{w(X)-w_0(X)}{s(X|Z)}+ \dfrac{w_0(X)}{s_0(X|Z) s(X|Z)} ( s_0(X|Z)-s(X|Z)) \bigg]^2  \lesssim 2\overline{\sigma}^2   \bar{\mathcal{C}}^2 \textbf{w}_N^2  +  2 \overline{\sigma}^2 \bar{\mathcal{C}}^6  \textbf{s}_N^2.
\end{align*} 
By Lemma \ref{lem:cte:properties}(d) and Jensen inequality, 
\begin{align*}
\E M_3^2 & = \E_{X} \bigg[ \int_{z \in Z} (\mu(X,z)  - \mu_0(X,z)) d P_Z(z) \bigg]^2 \\
 &\leq  \E_{X} \bigg[  \E [  (\mu(X,Z)  - \mu_0(X,Z)) \frac{w_0(X)}{s_0(X|Z)} | X]  \bigg]^2 \lesssim \textbf{m}_N^2.
\end{align*}
Therefore, Assumption \ref{ass:smallbias}   hold with $B_N:= \sqrt{N} \sqrt{d} (\textbf{w}_N \textbf{m}_N + \textbf{s}_N \textbf{m}_N)$ if the dictionary is bounded or 
$B_N:= \sqrt{N} \sqrt{d} (\textbf{w}_{N,2\kappa} \textbf{m}_{N,2\gamma} + \textbf{s}_{N,2\kappa} \textbf{m}_{N,2\gamma})$ for any $\kappa, \gamma$ so that $1/\kappa + 1/\gamma=1$. Furthermore, $\Lambda_N := \xi_d( \textbf{m}_N \vee \textbf{w}_N \vee  \textbf{s}_N)$ obeys Assumption \ref{ass:smallbias}. By Assumption \ref{ass:fsratecont}, $\Lambda_N=o(1)$ and $B_N=o(1)$.
\end{proof}

\begin{lemma}[Pointwise Linearization of  Orthogonal Estimator]
\label{lem:pointwise:cte}
Suppose Assumptions \ref{ass:identification}-\ref{ass:boundederror} and 
\ref{ass:fsratecont} hold, and $\xi_d N^{-1} = o(1)$. Then,  for any $\alpha \in \mathcal{S}^{d-1} := \{ \alpha \in \mathbb{R}^d: \| \alpha \| = 1\}$ the estimator $\widehat{\beta}^{\dagger}$ is approximately linear:
		         \begin{align*}
        \sqrt{N} \alpha' (\widehat{\beta}^{\dagger} - \beta)  = \alpha' Q^{-1}  \GN [p(X_i) (U_i + r_g(X_i)) + \tau_1(Z_{i};\mu_0)]+  R_{1,N}(\alpha)+ R_{1,N}^{\dagger} (\alpha),
        \end{align*}
        where $R_{1,N}(\alpha)$ is defined in Lemma \ref{lem:pointwise} and
        \begin{align*} 
        R_{1,N}^{\dagger} (\alpha) := \dfrac{1}{n (n-1)} \dfrac{1}{K}\sum_{k=1}^K \sum_{i \in J_k} \sum_{i \neq j} \tau (V_i, V_j; \widehat{\mu}) - \dfrac{1}{N} \sum_{i=1}^N \tau_1(V_i; \mu_0),
        \end{align*}
        summarizing the remainder of the $U$-statistic projection, obeys  
         \begin{align*} R_{1,N}^{\dagger} (\alpha)   
        &\lesssim_{P} \sqrt{d/N} + \sqrt{d\xi_d^4 \log^3 N/N^2} + \sqrt{d}\textbf{m}_N (1+\sqrt{\xi_d^2 \log N/N}).
         \end{align*}
\end{lemma}

\begin{proof}[Proof of Lemma \ref{lem:pointwise:cte}]

\textbf{ Step 1. } We decompose the remainder term  $R_{1,N}^{\dagger} (\alpha)$ into $R^{\dagger}_{1,k}$ and $R^{\dagger}_{2,k}$:
\begin{align*}
R_{1,N}^{\dagger} (\alpha) &:= \sqrt{N} \alpha' (\widehat{\beta}^{\dagger} - \widehat{\beta} - \frac{1}{N} \sum_{i=1}^N \tau_1(Z_{i}; \mu_0)) \\
&=\sqrt{N} \alpha' \widehat{Q}^{-1} \frac{1}{K} \sum_{k=1}^K \bigg[  \frac{1}{n(n-1)} \sum_{i,j \in J_k, i \neq j} \tau(V_i, V_j; \widehat{\mu}) -  \frac{1}{n} \sum_{i \in J_k}  \tau_1(Z_{i};\widehat{\mu}) \bigg]  \\
&+\sqrt{N} \alpha' \widehat{Q}^{-1} \frac{1}{K} \sum_{k=1}^K \bigg[\frac{1}{n} \sum_{i \in J_k}  \tau_1(Z_{i};\widehat{\mu} ) - \tau_1(Z_{i};\mu_0) \bigg]  \\
&=: \sqrt{N}\alpha' \widehat{Q}^{-1}  \frac{1}{K} \sum_{k=1}^K [R^{\dagger}_{1,k}+R^{\dagger}_{2,k}] \\
&= \frac{1}{K} \sum_{k=1}^K [\mathcal{R}_{1,k} + \mathcal{R}_{2,k} + \mathcal{R}_{1,k}' + \mathcal{R}_{2,k}'],
\end{align*}
where the latter terms are defined as:
\begin{align*}
\mathcal{R}_{1,k}:=\sqrt{N} \alpha' Q^{-1} R^{\dagger}_{1,k}, \quad \mathcal{R}_{2,k}:=\sqrt{N} \alpha' Q^{-1} R^{\dagger}_{2,k}, \\
\mathcal{R}_{1,k}':=\sqrt{N} \alpha'  (\widehat{Q}^{-1}-Q^{-1}) R^{\dagger}_{1,k}, \quad \mathcal{R}_{2,k}':=\sqrt{N} \alpha'  (\widehat{Q}^{-1}-Q^{-1}) R^{\dagger}_{2,k} \\
\end{align*}

\textbf{Step 2}. \textbf{Term $\mathcal{R}_{1,k}$}. Conditional on the data $(V_j)_{j \in J_k^c} $, $\widehat{\mu}$ is treated as fixed, and $R^{\dagger}_{1,k}$ is a degenerate $U$-statistic of order 2. To bound  $R^{\dagger}_{1,k}$, we invoke  Lemma \ref{lem:ChenKato} with the function class $$\mathcal{F}_{\alpha} = \{ \alpha' Q^{-1} \tau (\cdot, \cdot; \widehat{\mu} ) \}$$ consisting of a single element. The entropy integral of $\mathcal{F}_{\alpha}$ $\mathcal{J}_2(\delta)=\delta$.  Its envelope $F_{\alpha}(\cdot, \cdot): (V, V') \rightarrow C_{\text{min}}^{-1} \| \tau (V,V' ;\widehat{\mu})\|$ is bounded in $L_{P,2}$ norm: 
\begin{align*}
\E  \bigg[ \| \tau (V, V'; \widehat{\mu})\|^2 | \mathcal{E}_N, (V_j)_{j \in J_k^c}    \bigg] &\leq \sup_{\mu \in  M_N} \E \| \tau (V, V'; \mu)\|^2 \\
&\leq \bar{\mathcal{C}}^2 \frac{1}{2} (\| p(X) \|_{P,2}^2 + \| p(X') \|_{P,2}^2) \leq d \bar{\mathcal{C}}^2.
\end{align*}
Furthermore, for each $k \in \{1,2,\dots, K\}$
\begin{align*}
\max_{ i \in J_k} \| \tau (V_{2(i-1)+1}, V_{2i};\widehat{\mu} ) \| \leq \sup_{x \in \mathcal{X}} \| p(x) \| \bar{\mathcal{C}} =  \xi_d  \bar{\mathcal{C}}.
\end{align*}
Applying Lemma \ref{lem:ChenKato} to $ \alpha' Q^{-1} R^{\dagger}_{1,k}$ with $r=k=2$ and $\sigma_2= \| F\|_{P,2} \lesssim d  \bar{\mathcal{C}}$ gives
\begin{align*}
n \E \bigg[  |  \alpha' Q^{-1} R^{\dagger}_{1,k} | \bigg| \mathcal{E}_N,  (V_j)_{j \in J_k^c}  \bigg]  &\lesssim  \bar{\mathcal{C}} \sqrt{d}  + \bar{\mathcal{C}} n^{-1/2} \xi_d.
\end{align*}
By Markov inequality, $\mathcal{R}_1:= \sqrt{N}   \alpha' Q^{-1}  R^{\dagger}_{1,k}   \lesssim_{P} N^{-1/2} \bar{\mathcal{C}} \sqrt{d}  + \xi_d N^{-1} \bar{\mathcal{C}} = o_{P} (1)$
conditional on $(V_j)_{j \in J_k^c}$. By Lemma \ref{lem:cond}, $\mathcal{R}_{1,k}=\sqrt{N}   \alpha' Q^{-1}  R^{\dagger}_{1,k}    = o_{P} (1)$ unconditionnally.

\textbf{ Step 3}. \textbf{Term $\mathcal{R}_{1,k}'$}. By Assumptions \ref{ass:identification} and  \ref{ass:growth}, $\| \widehat{Q}^{-1} - Q^{-1} \| \lesssim_{P} \sqrt{\dfrac{\xi_d^2 \log N}{N}}$. Therefore,
$\mathcal{R}_{1,k}' \lesssim_{P} \sqrt{\frac{\xi_d^2 \log N}{N}} \| R^{\dagger}_{1,k}  \|$, where 
\begin{align}
\| R^{\dagger}_{1,k}  \| \leq \sqrt{d} \| R^{\dagger}_{1,k} \|_{\infty}.
\end{align}
We invoke  Lemma \ref{lem:Chen} for the canonical two-sample $U$-statistic $R^{\dagger}_{1,k} $ conditional on the event $\mathcal{E}_N$ and the sample $(V_j)_{j \in J_k^c}$.
On this event, $\tau(V_i, V_j; \widehat{\mu})$ is bounded in absolute norm:
\begin{align*}
\| R^{\dagger}_{1,k} \|_{\infty} \leq 2 \max_{1 \leq j \leq d} \sup_{x \in \mathcal{X}} | p_j(x) | \mathcal{C} \leq 2 \xi_d  \mathcal{C}.
\end{align*}
Therefore, the constants $D_2, D_4, M_4$ are bounded by $\xi_d  \mathcal{C}$. Thus,
\begin{align*}
\E [ \| R^{\dagger}_{1,k} \|_{\infty} ] &\lesssim \xi_d  \mathcal{C}  \dfrac{\log d}{N}  
\end{align*}
and $\| R^{\dagger}_{1,k} \|_{\infty} =  O_{P} \bigg(  \dfrac{\xi_d \log d}{N} \bigg)= o(1)$.  Invoking $\log d \leq \log N$, we obtain
\begin{align*}
\mathcal{R}_{1,k}' \lesssim_{P} \sqrt{N}  \sqrt{\frac{d \xi_d^2 \log N}{N }}   \dfrac{\xi_d \log d}{N}  = O\bigg( \sqrt{\frac{d \xi_d^4 \log^2 N}{N^2 }} \bigg).
\end{align*}

\textbf{ Step 4}.  \textbf{Term $\mathcal{R}_{2,k}$}. Conditional on the data $(V_j)_{j \in J_k^c} $, the remainder term $R^{\dagger}_{2,k}$ is a sample average of mean zero i.i.d r.v.:
\begin{align*}
R^{\dagger}_{2,k}:=  \E_{n,k} [\tau_1(Z_{i}; \widehat{\mu}) -  \tau_1(Z_{i}; \mu_0) ] =^{i}  \E_{n,k} \tau_1(Z_{i}; \widehat{\mu}- \mu_0), 
\end{align*}
where $i$ follows from Lemma \ref{lem:cte:properties}(b). The  bound on the norm follows as
\begin{align*}
\E [ \| \sqrt{n} R^{\dagger}_{2,k} \|^2 | \mathcal{E}_N, (V_j)_{j \in J_k^c} ] &\leq \E [ \|  \tau_1(Z_{i}; \widehat{\mu}) -  \tau_1(Z_{i}; \mu_0) \|^2  |  \mathcal{E}_N,  (V_j)_{j \in J_k^c} ] \\
&\leq \sup_{\mu \in  M_N} \E \| \tau_1(Z_{i}; \mu - \mu_0) \| ^2 \\
&=  \sup_{\mu \in  M_N} \sum_{j=1}^d \E_{Z} ( \E_{X} [p_j(X) (\mu^0(X;z)-\mu_0^0(X;z)) ] |_{z=Z} )^2 \\
&\leq  \sup_{\mu \in  M_N} \sum_{j=1}^d \E_{Z}   \E_{X} p_j^2(X) \E_{X} [(\mu^0(X;z)-\mu_0^0(X;z))^2]|_{z=Z} \lesssim d \textbf{m}_N^2.
\end{align*}
Therefore, $\| \sqrt{n} R^{\dagger}_{2,k} \| = O_{P} (\sqrt{d} \textbf{m}_N)$ and $\mathcal{R}_{2,k}'\lesssim_{P} \sqrt{d} \textbf{m}_N \sqrt{\frac{\xi_d^2\log N}{N }}$.



\textbf{ Step 5}. Collecting the terms gives the bound in Lemma \ref{lem:pointwise:cte}. The bound on $R_{1,N}(\alpha)$ is established in Lemma \ref{lem:pointwise}.
\end{proof}

\begin{lemma}[Uniform Linearization of   Orthogonal Estimator]

\label{lem:uniform:cte}
Suppose Assumptions  \ref{ass:identification}-\ref{ass:basis} hold with $\xi_d^L/C_{\text{min}} \geq e^2/16 \vee e$. Suppose $N^{-1} \xi_d \log N = o(1)$. Then, the estimator $\widehat{\beta}^{\dagger}$ is approximately linear uniformly over the domain $\mathcal{X}$:

		\begin{align*} 		 |& \sqrt{N} \alpha(x)' (\widehat{\beta}^{\dagger}  - \beta_0) - \alpha'(x)Q^{-1} \GN [p(X_i)(U_i + r_g(X_i))  + \tau_1(Z_{i}; \mu_0) ] | \\
		 &\leq R_{1,N}(\alpha(x)) + R_{1,N}^{\dagger}(\alpha(x)), 
		 \end{align*} where  $R_{1,N}(\alpha(x))$ is defined in Lemma \ref{lem:uniform} and		 $R_{1,N}^{\dagger}(\alpha(x))$, summarizing the 
		 remainder of the $U$-statistic projection
		  \begin{align*} R_{1,N}^{\dagger} (\alpha(x))    &\lesssim_{P}  N^{-1/2} \bar{\mathcal{C}} (\sqrt{d} \log N  + \xi_d N^{-1}\log^2 N)  \\
        &+ \ \sqrt{\dfrac{d \xi_d^4 \log^3 N}{N^2}} + \sqrt{d} \textbf{m}_N  \bigg(  1+ \sqrt{\dfrac{\xi_d^2 \log N}{N}}  \bigg)=:\bar{R}_{1,N}^{\dagger}
         \end{align*}
\end{lemma}

\begin{proof}[Proof of Lemma \ref{lem:uniform:cte}]
\textbf{ Step 1}. Conditional on the data $(V_j)_{j \in J_k^c} $,$\widehat{\mu}$ is treated as fixed, and $R^{\dagger}_{1,k}(\alpha(x))$ is a degenerate $U$-statistic of order 2.  Consider a function class
$$ \mathcal{F}  = \{ \alpha(x)'Q^{-1}  \tau (\cdot, \cdot; \widehat{\mu} ), \quad x \in \mathcal{X}  \} $$
whose envelope $F:= C_{\text{min}}^{-1} \|  \tau (\cdot, \cdot; \widehat{\mu})  \|$ is square integrable (Step 2, Proof of Lemma \ref{lem:pointwise:cte}). We  determine the bracket size. Recognize that 
$$  |  [\alpha(x) - \alpha(x')] ^\top Q^{-1} \tau (V, V'; \widehat{\mu}) | \leq \xi_d^L \| x - x' \|  \| Q^{-1} \|  \| \tau (V, V'; \widehat{\mu})   \|$$
and therefore $$ \sup_{Q} N(\mathcal{F}, L^2(Q), \epsilon \| F \|_{L^2_{Q}}) \leq \bigg(\frac{ \xi^L_d /C_{\min} }{\epsilon}\bigg)^r $$ 
Plugging in $\sigma = \|F \|_{L^2_P} \leq  \sqrt{d} \bar{\mathcal{C}} $, $A=  \xi^L_d /C_{\min}$, $V=r$ 
into Lemma \ref{cor:ChenKato}, we obtain:
\begin{align*}
 \sup_{x \in \mathcal{X}} | \sqrt{N} \alpha(x)^\top Q^{-1} R^{\dagger}_{1,k} |&\lesssim_{P} N^{-1/2} \sqrt{d} \bar{\mathcal{C}}  \log \xi_d^L + N^{-1} \xi_d \log^2 \xi_d^L \bar{\mathcal{C}}   \\
 &\lesssim_{P} N^{-1/2} \sqrt{d} \bar{\mathcal{C}} \log N + N^{-1} \xi_d \log^2 N \bar{\mathcal{C}} 
 \end{align*}

\textbf{ Step 2}. Let $R^{\dagger}_{1,k}$ and $R^{\dagger}_{2,k}$ be as defined in the proof of Lemma \ref{lem:pointwise:cte}. We establish the following bound
\begin{align*}
    \sup_{x \in \mathcal{X}} | \sqrt{N} \alpha(x)^\top (Q^{-1} - \widehat{Q}^{-1}) R^{\dagger}_{1,k} |&\lesssim_{P} \| \widehat{Q} - Q \|  \| R^{\dagger}_{1,k} \| = O_{P} \bigg(\sqrt{\frac{\xi_d^2 \log N}{N}}\sqrt{d} \| R^{\dagger}_{1,k} \|_{\infty} \bigg)\\
    &=^{i} O_{P} \bigg(\sqrt{\frac{\xi_d^4 \log^3 N}{N^2}}\sqrt{d}  \bigg),
\end{align*}
where $i$ follows from Step 3 of the Proof of Lemma \ref{lem:pointwise:cte}.

\textbf{ Step 3}. We establish the following bound
\begin{align*}
 \sup_{x \in \mathcal{X}} | \sqrt{N} \alpha(x)^\top Q^{-1} R^{\dagger}_{2,k} |&\lesssim_{P} \|Q^{-1} \|  \| R^{\dagger}_{2,k} \| = O_{P} (\sqrt{d} \textbf{m}_N) \\
 \sup_{x \in \mathcal{X}} | \sqrt{N} \alpha(x)^\top (Q^{-1} - \widehat{Q}^{-1}) R^{\dagger}_{2,k} |&\lesssim_{P} \| \widehat{Q} - Q \|  \| R^{\dagger}_{2,k} \| = O_{P} \bigg(\sqrt{\frac{\xi_d^2 \log N}{N}}\sqrt{d} \textbf{m}_N \bigg), \\
\end{align*}
where $i$ and $ii$ hold by Step 4 of the proof of Lemma \ref{lem:pointwise:cte}.

\end{proof}

\begin{proof}[Proof of Theorem \ref{thrm:cte:pointwise}]
In Lemma \ref{lem:pointwise:cte}, we have shown the pointwise linearization of $\widehat{\beta}^{\dagger}$. Consider a triangular array $\{ (\chi_{Ni})_{i=1}^N, \quad i=1,2, \dots, N \}$, where $\chi_{Ni}$ is defined as
\begin{align*}
\chi_{Ni} = \frac{\alpha' Q^{-1} [p_i (U_i+r_i) + \tau_1(Z_{i};\mu_0)]}{\sqrt{N} \| \alpha' (\Omega^{\dagger})^{1/2} \|}.
\end{align*}
We verify the conditions of Lindeberg's CLT for $\chi_{Ni}$. Assumption \ref{ass:approx} and Lemma \ref{lem:cte:properties}(a) imply
\begin{align*}
\text{var} (\sum_{i=1}^N \chi_{Ni}) = 1 \text { and }  |\chi_{Ni} | \leq (\xi_d | U_i |+ \xi_d l_d r_d + \bar{\mathcal{C}}  \sqrt{d})/\sqrt{N}.
\end{align*}
Fix $\delta>0$ and $C>0$.   Observe that
\begin{align*}
   &N  \E \chi_{Ni}^2 1_{\{ | \chi_{Ni} | > \delta \} } \\
   &\leq  \frac{2}{\alpha' \Omega^{\dagger} \alpha} \bigg( \E (\alpha' Q^{-1}  p_i (U_i+r_i))^2 1_{\{ | \chi_{Ni} | > \delta \} } +  \E (\alpha' Q^{-1}  \tau_1(Z_{i};\mu_0))^2 1_{\{ | \chi_{Ni} | > \delta \}} \bigg) = L_1 + L_2.
\end{align*}
Observe that the event $| \chi_{Ni} | > \delta$ implies a larger event
\begin{align*}
   \bigg \{  | \chi_{Ni} | > \delta \bigg \} \Rightarrow \bigg \{ |U_i| > \delta \sqrt{N}/\xi_d  - \bar{\mathcal{C}} \sqrt{d}/\xi_d  -  l_d r_d  \bigg \} =: \bigg\{ |U_i| > M \bigg\},
\end{align*}
where $M= \delta \sqrt{N}/\xi_d  - \bar{\mathcal{C}}\sqrt{d}/\xi_d   -  l_d r_d \rightarrow \infty$ by assumption of the Theorem.  The numerator of the term $L_1$ is bounded as
\begin{align*}
   \E (\alpha' Q^{-1}  p_i (U_i+r_i))^2  &\leq (\alpha' Q^{-1} Q Q^{-1} \alpha) \sup_{x \in \mathcal{X}} \E [(U_i+r_i)^2| X_i=x]1_{\{ |U_i| > M \}}  \\
   &\leq 2\alpha' Q^{-1} \alpha  \bigg( \sup_{x \in \mathcal{X}}\E [U_i^2|X_i=x] 1_{ \{ |U_i| > M\}}  \bigg)+ l_d^2 r_d^2 \Pr ( |U_i| > M ) \rightarrow 0,
\end{align*}
where  the last inequality holds by the Lindeberg condition assumed in the Theorem.   To  bound the numerator of term $L_2$, observe that
\begin{align*}
     \E (\alpha' Q^{-1}  \tau_1(Z_{i};\mu_0))^2 1_{\{ |U_i| > M  \}} &\lesssim d \bar{\mathcal{C}}^2 \Pr ( |U_i| > M ) \rightarrow 0.
\end{align*}
Finally, the common denominator of $L_1$ and $L_2$, equal to $(\alpha' \Omega^{\dagger} \alpha)^{-1}$, is bounded by  $(C_{\text{min}}^{\dagger})^{-1}< \infty$ by assumption of the Theorem.

\end{proof}

\begin{proof}[Proof of Theorem \ref{thrm:cte:uniform}]
Proof of Theorem \ref{thrm:cte:uniform} (a). The bound on
\begin{align}
    \sup_{x \in \mathcal{X}} | \alpha(x)' \EN \tau_1(Z_i;\mu_0) | &\leq \sup_{x \in \mathcal{X}} \| \alpha(x) \| \| \EN \tau_1(Z_i;\mu_0) \| \nonumber \\
    &\leq \sqrt{d} \| \EN \tau_1(Z_i;\mu_0) \|_{\infty} \lesssim_{P} \sqrt{d \log d/N} \label{eq:bound},
\end{align}
where the last inequality follows from McDiarmid inequality for bounded sample averages of bounded r.v. Therefore, the total bound on
\begin{align*}
    \sup_{x \in \mathcal{X}} | \alpha(x)' (\widehat{\beta}^{\dagger} - \beta_0) | &\leq^{i}  \sup_{x \in \mathcal{X}} | \alpha(x)' Q^{-1}[\EN p_i (U_i + r_i) +  \tau_1(Z_i; \mu_0) ] |+ \bar{R}_{1N} + \bar{R}_{2N} + \bar{R}_{1N}^{\dagger} \\
    &\lesssim_{P}  N^{-1/2} (\sqrt{\log N} + \bar{R}_{1N} + \bar{R}_{2N} +\sqrt{d \log d} +\bar{R}_{1,N}^{\dagger}),
\end{align*}
where i follows from Lemmas \ref{lem:uniform} and \ref{lem:uniform:cte} and ii follows from \eqref{eq:bound}. By assumption of the Theorem, $\sqrt{d \xi_d^4 \log^3 N/N^2} = o(\sqrt{\xi_d^2 \log^2 N/N})=o(\sqrt{\log N})$.  Furthermore, $(d \log N/N)^{1/2} =o(1)$ and $\xi_d \log^2 N/N=o(1)$. Therefore, the only non-negligible term in $\bar{R}_{1,N}^{\dagger}$ is  $\sqrt{d} \textbf{m}_N$. 

Proof of Theorem \ref{thrm:cte:uniform} (b). We shall apply Yurinskii’s coupling (see Theorem 4.4 in \cite{NewOls}). To apply the coupling, we need to verify that $\E \| \tau_1(Z_i;\mu_0)\|^3 \lesssim d^{3/2} $, which holds by assumption of the Theorem. The rest of the Theorem follows from the proof of Theorem 4.4 of \cite{NewOls}. 

\end{proof}

\begin{proof}[Proof of Corollary \ref{thrm:CATE}] Let us show that $Y(\eta)$ given by Equation \eqref{eq:catesig}  satisfies Assumption \ref{ass:smallbias}.
\textbf{ Step 1. Choice of $B_N$.}
\begin{align*}
	&Y(\eta) - Y(\eta_0) =[\mu(1,Z) - \mu_0(1,Z)][ 1-  \dfrac{D} {s_0(Z)} ]   - [\mu(0,Z) - \mu_0(0,Z)][ 1-  \dfrac{1-D} {1-s_0(Z)} ]  \\
	&+ [ s_0(Z) -s (Z)] \big[  [ Y - \mu_0(1,Z)  ] \dfrac{D}{s(Z) s_0(Z)}  - [ Y - \mu_0(0,Z)  ] \dfrac{1-D}{(1-s(Z))(1- s_0(Z))} \big]   \\
	&+ [ s_0(Z) -s (Z)]\big[  [\mu(1,Z)  - \mu_0(1,Z)  ] \dfrac{D}{s(Z) s_0(Z)}  - [\mu(0,Z) - \mu_0(0,Z)  ] \dfrac{1-D}{(1-s(Z))(1- s_0(Z))} \big]  \\
	&= S_1 + S_1' + S_2 + S_2' + S_3 + S_3'.
\end{align*}
We will bound the terms $S_1,S_2, S_3$. The bounds on the terms $S_1',S_2',S_3'$ follow by a similar argument. To choose $B_N$ in Assumption \ref{ass:smallbias}, observe that
\begin{align*}
    \E p(X) S_1 &:= \E  p(X) [\mu(1,Z) - \mu_0(1,Z)] [1- \dfrac{D}{s_0(Z)} ] \\
    &=^{i} \E \bigg[ p (X) [\mu(1,Z) - \mu_0(1,Z)] \E \bigg[1- \dfrac{D}{s_0(Z)} \bigg| Z \bigg] \bigg] =^{ii} 0,
\end{align*}
where $i$ holds because $X$ is a subvector of $Z$ and $ii$ holds by definition of the propensity score. Likewise,
\begin{align*}
    \E p(X) S_2 &:= \E   p(X)  [ s_0(Z) -s (Z)] \big[  [ Y - \mu_0(1,Z)  ]  \\
    &= \E \bigg[p (X)  [ s_0(Z) -s (Z)]  \E [   Y  - \mu_0(1,Z) | Z  ] \bigg]=^{ii} 0,
\end{align*}
where $ii$ holds by the definition of regression function. Observe that $\Pr (D^2 = 1|Z) = \Pr (D = 1|Z)=s_0(Z)$ and 
\begin{align}
    \label{eq:dss}
  \E [ D^2 s^{-2} (Z) s^{-2}_0(Z)  | Z=z] \leq \bar{\mathcal{C}}^3 \text{ for any } z \in \mathcal{Z}, \\
   \Ep \bigg[ \bigg[ 1-  \dfrac{D} {s_0(Z)} \bigg]^2 \bigg| Z=z  \bigg] = 1 - 2 + s_0^{-1}(z) \leq \bar{\mathcal{C}} \quad \text{ for any } z \in \mathcal{Z} \label{eq:dss2} 
\end{align}
Therefore,
\begin{align*}
    \| \E p(X) S_3 \|^2 \leq  \sum_{j=1}^d  (\E p_j(X)  S_3(X))^2
    \lesssim  \bar{\mathcal{C}}^3 d \textbf{m}_N^2 \textbf{s}_N^2
\end{align*}
if the basis is bounded, in which case $B_N:=\sqrt{N} d \textbf{m}_N \textbf{s}_N$ satisfies Assumption \ref{ass:smallbias}. Otherwise,
\begin{align*}
     \| \E p(X) S_3 \|^2 \lesssim d \textbf{m}_{N,2\gamma}^2 \textbf{s}_{N,2\kappa}^2
\end{align*}
for any $\kappa, \gamma$ so that $1/\kappa + 1/\gamma=1$, in which case $B_N:=\sqrt{N} d \textbf{m}_{N,2\gamma} \textbf{s}_{N,2\kappa}$ satisfies Assumption \ref{ass:smallbias}. 
By Assumption \ref{ass:fsrate}, $B_N = o(1)$ for either choice of $B_N$.

\textbf{ Step 2. Choice of $\Lambda_N$.} As stated in Assumption \ref{ass:smallbias}, $\Lambda_N \leq \xi_d \sup_{\eta \in \mathcal{T}_N} (\E (Y(\eta) - Y(\eta_0))^2)^{1/2}$. Therefore, it suffices to bound $\E S_k^2$ for $k=1,2,3$. \eqref{eq:dss2} implies the following inequalities 
$$	\Ep S_1^2 = \Ep [[\mu(1,Z) - \mu_0(1,Z)]^2 \Ep [ [ 1-  D/s_0(Z) ]^2 | Z  ]] \lesssim \textbf{m}^2_N  \bar{\mathcal{C}},
$$
and
$$	\Ep  S_2^2 =^{i} \Ep [ s(Z) -s_0(Z)]^2 \Ep  [ [ Y- \mu_0(1,Z) ]^2 | Z, D=1  \big]  \dfrac{1} {s_0(Z) s^2(Z)}  \lesssim \textbf{s}_N^2 \bar{\mathcal{C}}^3.$$
Since one of the terms $(s_0(Z) -s (Z))^2$ or $(\mu(1,Z)  - \mu_0(1,Z))^2$ is bounded by 
$\bar{\mathcal{C}}^2$ almost surely in $Z$,
$$  \Ep S_3^2 = \E [ s_0(Z) -s (Z)]^2 [\mu(1,Z)  - \mu_0(1,Z)  ]^2 s^{-2}(Z) s_0^{-1}(Z) \leq \min ( \textbf{m}_N^2, \textbf{s}_N^2) \bar{\mathcal{C}}^5.$$ Therefore,
 $\Lambda_N = \xi_d(\textbf{m}_N  \vee \textbf{s}_N)$ satisfies Assumption \ref{ass:smallbias}. By Assumption \ref{ass:fsrate}, $\Lambda_N = o(1)$. 
\end{proof}

\begin{proof}[Proof of Corollary \ref{thrm:MD}]
Let us show that $Y(\eta)$ given by Equation \eqref{eq:mdsig} satisfies Assumption  \ref{ass:smallbias}.
\begin{align*}
	Y(\eta) - Y(\eta_0) &=[ \mu(Z) - \mu_0(Z)][ 1-  \dfrac{D} {s_0(Z)} ] + [ s_0(Z) -s (Z)][ Y - \mu_0(Z)  ] \dfrac{D}{s(Z) s_0(Z)} \\
	&+ [ s_0(Z) -s (Z)][ \mu(Z) - \mu_0(Z) ] \dfrac{D}{s(Z) s_0(Z)}\\
	&= S_1 + S_2 + S_3.
\end{align*}
To choose $B_N$ in Assumption \ref{ass:smallbias}, observe that
\begin{align*}
    \E p(X) S_1 = 0 \quad \text{ and } \quad  \E p(X) S_2 = 0 
\end{align*}
by Step 1 of the proof of Corollary \ref{thrm:CATE}. Invoking \eqref{eq:dss},
\begin{align*}
     \|   \Ep p(X) S_3  \|^2 \leq  \sum_{j=1}^d (\E p_j(X) S_3)^2 \lesssim d \textbf{m}_N^2 \textbf{s}_N^2
\end{align*}
if the basis is bounded, in which case $B_N=\sqrt{d} \textbf{m}_N \textbf{s}_N $ satisfies Assumption \ref{ass:smallbias}. Otherwise,
\begin{align*}
     \|   \Ep p(X) S_3  \|^2 \leq  \sum_{j=1}^d (\E p_j(X) S_3)^2 \lesssim d \textbf{m}_{N,2\gamma}^2 \textbf{s}_{N,2\kappa}^2,
\end{align*}
where $1/\kappa+1/\gamma=1$, in which case $B_N=\sqrt{d} \textbf{m}_{N,2\gamma} \textbf{s}_{N,2\kappa}$ satisfies Assumption \ref{ass:smallbias}.  By Assumption \ref{ass:fsrate}, $B_N = o(1)$.   

As stated in Assumption \ref{ass:smallbias}, $\Lambda_N \leq \xi_d \sup_{\eta \in \mathcal{T}_N} (\E (Y(\eta) - Y(\eta_0))^2)^{1/2}$. Similar to Step 2 in the proof of Corollary \ref{thrm:CATE},
\begin{align*}
	\Ep S_1^2 &= \Ep [ \mu(Z) - \mu_0(Z)]^2 \Ep \big[ [ 1-  \dfrac{D} {s_0(Z)} ]^2 | Z=z  \big] \lesssim \textbf{m}_N^2 \bar{\mathcal{C}} \\
	\Ep S_2^2 &= \Ep [ s(Z) -s_0(Z)]^2 \Ep \big[ [ Y- \mu_0(Z) ]^2 | Z=z,D=1  \big]  \dfrac{1} {s_0(Z) s^2(Z)} \lesssim \textbf{s}_N^2 \bar{\mathcal{C}}^3 \\
	\Ep S_3^2 &=  \Ep  [ s_0(Z) -s (Z)]^2[ \mu(Z) - \mu_0(Z) ]^2 \dfrac{1}{s^2(Z) s_0(Z)} \lesssim \min (\textbf{m}_N^2, \textbf{s}_N^2)\bar{\mathcal{C}}^3.
\end{align*}
 Therefore,
 $\Lambda_N = \xi_d(\textbf{m}_N  \vee \textbf{s}_N)$ satisfies Assumption \ref{ass:smallbias}. By Assumption \ref{ass:fsrate}, $\Lambda_N = o(1)$. 
\end{proof}

\begin{proof}[Proof of Corollary \ref{thrm:CAPD}]
Let us show that $Y(\eta)$ given by Equation \eqref{eq:capdsig} satisfies Assumption \ref{ass:smallbias}.
\begin{align*}
	Y(\eta) - Y(\eta_0) &=-\partial_{d} \log s_0(D|Z) [\mu(D,Z) - \mu_0(D,Z) ]  +  \partial_{d}  [\mu(D,Z) - \mu_0(D,Z) ]    \\
	&+ [\partial_{d} \log s_0(D|Z) - \partial_{d} \log s(D|Z)][ Y - \mu_0(D,Z)]  \\
	&+ [\partial_{d} \log s_0(D|Z) - \partial_{d} \log s(D|Z)] [\mu(D,Z) - \mu_0(D,Z) ]   \\
	&= S_1 + S_2 + S_3.
\end{align*}

\textbf{ Step 1. Choice of $B_N$.} Suppose the conditional density $s_0(d|z)$ of $D,Z$ has bounded support. Integration by parts implies
\begin{align*}
   & \E \partial_{d}  [\mu(D,Z) - \mu_0(D,Z) ] |Z=z] = \int_{t \in  \mathrm{R}} \partial_{t} (\mu(t,z) - \mu_0(t,z)) s(t|z) dt  \\
    &= - \int_{t\in  \mathrm{R}} (\mu(t,z) - \mu_0(t,z)) 
    \partial_{t}  s(t|z) d t \\
    &= \int_{t\in  \mathrm{R}} (\mu(t,z) - \mu_0(t,z))  \partial_{t}  \log s(t|z) s(t|z) dt \\
    &= \E [\partial_{d} \log s_0(D|Z) (\mu (D,Z) - \mu_0(D,Z)) |Z=z].
 \end{align*}
Therefore, $\E[ S_1 | Z= z]=0$ for any $z \in \mathcal{Z}$. Since $X$ is a subvector of $Z$, for any function of $X$,
\begin{align*}
    \E p(X) S_1 = 0. 
\end{align*}
By definition of regression function $\mu_0(D,Z) = \E[ Y | D,Z]$,
\begin{align*}
\E p(X) S_2 = \E p (X) [\partial_{d} \log s_0(D|Z) - \partial_{d} \log s(D|Z)] \E [ Y - \mu_0(D,Z)] =0.
\end{align*}
Finally,
\begin{align*}
\| \E p(X) S_3 \|^2 \leq \sum_{j=1}^d (\E p_j(X)  S_3)^2 \leq d \textbf{s}_N^2 \textbf{m}_N^2
\end{align*}
if the basis is bounded, in which case $B_N=\sqrt{d} \textbf{s}_N \textbf{m}_N$ obeys Assumption \ref{ass:smallbias}. Otherwise, $B_N=\sqrt{d} \textbf{s}_{N,2\kappa} \textbf{m}_{N,2\gamma}$ obeys Assumption \ref{ass:smallbias}.

\begin{align*}
	\Ep S_1^2 &\leq \Ep (-\partial_{d} \log s_0(D|Z) [\mu(D,Z) - \mu_0(D,Z) ]  )^2 \\
	&+ \Ep ( \partial_{d}  [\mu(D,Z) - \mu_0(D,Z) ]  )^2 \lesssim \textbf{m}_N^2   \\
	\Ep S_2^2 &\leq (\Ep \big[ [ Y - \mu_0(D,Z)] ^2 | D,Z \big]) \textbf{s}_N^2 \lesssim  \textbf{s}_N^2\\
	\Ep S_3^2 &\leq (\Ep  [\mu(D,Z) - \mu_0(D,Z) ]^2) \textbf{s}_N^2 \lesssim \min [\textbf{f}^2_N \vee \textbf{m}^2_N] 
\end{align*}
 Therefore,  $\Lambda_N = \xi_d  (\Ep (Y(\eta) - Y(\eta_0) )^2)^{1/2}  \lesssim \xi_d[\textbf{m}_N \vee \textbf{s}_N]$ obeys Assumption \ref{ass:smallbias}. 
\end{proof}





\section*{Appendix B: Verification of High-Level Conditions and Monte-Carlo Evidence}
\renewcommand{\thetable}{B.\arabic{table}}
\renewcommand{\theassumption}{B.\arabic{assumption}}
\renewcommand{\theexample}{B.\arabic{example}}
\renewcommand{\thelemma}{B.\arabic{lemma}}
\renewcommand{\theequation}{B.\arabic{equation}}
\renewcommand{\thesection}{B}
\setcounter{equation}{0}
\setcounter{lemma}{0}
\medskip

\label{sec:suppapp}

Here give an example of a model and and a first-stage estimator that satisfy the small bias condition from the main text.

\begin{example}[Partially Missing Outcome with High-Dimensional Sparse Design]
\label{ex:sparse}
Consider the setup of Example \ref{MD}. Let the observable vector $(D,X,DY^{*})$ consist of the covariate vector of interest $X$ and a partially observed variable $Y^{*}$, whose presence is indicated by $D \in \{1,0\}$. In addition, suppose there exists an observable vector $Z$ such that Missingnesss at Random  is satisfied conditional on $Z$.  Let $p_{\mu}(z),p_{s}(z) $ be  high-dimensional basis functions of the vector $Z$ that approximate the conditional expectation functions $\mu_0(z), s_0(z)$ using the linear and logistic links, respectively:
 \begin{align}
\mu_0 (z)&= p_{\mu}(z)' \theta_0 + r_{\mu} (z)  \label{eq:mu0}\\
s_0 (z)&= \mcL(p_{s}(z) ' \delta_0) + r_s(z) := \dfrac{ \exp (p_{s}(z) ' \delta) }{ \exp (p_{s}(z) ' \delta) + 1} + r_s(z)
\end{align}
 where $\theta, \delta$ are the vectors in $\mathcal{R}^{p_{\theta}}, \mathcal{R}^{p_{\delta}}$ whose dimensions are allowed to be larger than the sample size $N$, and $r_{\mu}(z), r_s(z)$ are the misspecification errors of the respective link functions  that vanish  as described in Assumptions \ref{ass:ident}, \ref{ass:ident2}. For each $\gamma \in \{ \theta, \delta \}$, denote a support set $$T_{\gamma} := \{ j : \gamma_j \neq 0, j \in \{1,2,..,p_{\gamma}\}\} $$ and  its cardinality, which we refer to as sparsity index of $\gamma$,  $$s_{\gamma} := |T| = \| \gamma \|_0 \quad \forall \quad \gamma \in \{ \theta, \delta \}. $$   We allow the cardinality of $s_{\delta}, s_{\theta}$ to grow with $N$. 

Let  $r_N, \Delta_N \rightarrow 0$ be a fixed sequence of constants approaching zero from above at a speed at most polynomial in $N$: for example, $r_N \geq \frac{1}{N^c}$
for some $c>0$, $\ell_N = \log N$. For $\gamma \in \{ \theta, \delta\}$, let $c_{\gamma},C_{\gamma}, \kappa'_{\gamma}, \kappa''_{\gamma}$   and $\nu \in [0,1]$ be positive constants that do not depend on $N$. Finally, let $\| V\|_{P_{N},2} := (N^{-1} \sum_{i=1}^N V_i^2)^{1/2}$.

 \begin{assumption}[Regularity Conditions for Linear Link]
 \label{ass:ident}
We assume that the following standard conditions hold. (a) With probability $1-\Delta_N$, the minimal and maximal empirical  RSE are bounded from below by $\kappa_{\mu}'$ and from above by $\kappa_{\mu}''$: $$\kappa_{\mu}' \leq \inf_{ \| \delta \|_0 \leq s_{\theta} \ell_N, \| \delta \| =1 } \| D p_{\mu}(Z)  \|_{P_N,2}  \leq \sup_{ \| \delta \|_0 \leq s_{\theta} \ell_N, \| \delta \| =1 } \| D p_{\mu}(Z)\|_{P_N,2}  \leq \kappa_{\mu}''.$$ (b) There exists  absolute constants $B_{\theta},C_{\theta},c_{\theta}>0$: regressors $\max_{1 \leq j \leq p_{\theta}} |p_{\mu,j}(Z) | \leq B_{\theta} \quad \text{a.s.}$ and  $c_{\theta} \leq \max_{1 \leq j \leq p_{\theta}}   \Ep p_{\mu,j}(Z)^2 \leq C_{\theta}$ (c) With probability $1-\Delta_N$, $\EN r_{\mu}^2(Z_i) \leq C_{\theta} s_{\theta} \log (p_{\theta} \vee N)/N $. (d) Growth restriction: for some $r_N =o(1)$, $\log (p_{\theta} \wedge N) \lesssim r_N N^{1/3}$. (e) The moments of the model are boundedly heteroscedastic: $c_{\theta} \leq \E [(Y - \mu_0(Z))^2|Z]  \leq C_{\theta}$,  $\max_{1 \leq j \leq p_{\theta}}  \Ep [|p_{\mu,j}(Z) (Y - \mu_0(Z))|^3 + |p_{\mu,j}(Z) Y|^3] \leq C_{\theta}$. (f) With probability $1-\Delta_N$,  $$\max_{1 \leq j \leq p_{\theta}} [ \E_N - \E] [p_{\mu,j}^2(Z) (Y - \mu_0(Z))^2 + p_{\mu,j}^2(Z) Y^2] \leq r_N N^{-1/2}.$$ 
\end{assumption} 

 \begin{assumption}[Regularity Conditions for Logistic Link]
 \label{ass:ident2}
We assume that the following standard conditions hold. With probability $1-\Delta_N$, the minimal and maximal empirical  RSE are bounded from below by $\kappa_s'$ and from above by $\kappa_s''$:  $$\kappa_s' \leq \inf_{ \| \delta \|_0 \leq s_{\delta} \ell_N, \| \delta \| =1 } \| p_{s}(Z)' \delta \|_{P_N,2}  \leq \sup_{ \| \delta \|_0 \leq s_{\delta}  \ell_N, \| \delta \| =1 } \| p_{s}(Z)' \delta \|_{P_N,2}  \leq \kappa_s''.$$ (b) There exist  absolute constants $B_{\delta},C_{\delta},c_{\delta}>0$: regressors $\max_{1 \leq j \leq p_{\delta}} |p_{s,j}(Z) | \leq B_{\delta} \quad \text{a.s.}$ and  $c_{\delta} \leq  \max_{1 \leq j \leq p_{\delta}}  \Ep p_{s,j}(Z)^2 \leq C_{\delta}$ (c) With probability $1-\Delta_N$, $\EN r_{s}^2(Z_i) \leq C_{\delta} s_{\delta} \log (p_{\delta} \vee N)/N $. (d) Growth restriction: for some $r_N =o(1)$, $\log (p_{\delta} \wedge N) \lesssim r_N N^{1/3}$. (e) With probability $1-\Delta_N$,  $$\max_{1 \leq j \leq p_{\delta}} [ \E_N - \E] [p_{\delta,j}^2(Z) (D - s_0(Z))^2] \leq r_N.$$

\end{assumption} 
Assumptions  \ref{ass:ident} and \ref{ass:ident2} are a simplification of the Assumption 6.1-6.2 in \cite{Program}. The following estimators of $\mu_0(Z)$ and $s_0(Z)$ are available. 
 
 \begin{definition}[Lasso Estimator of the Regression Function]
 \label{rem:lin}
Let  $\lambda = 1.1\sqrt{N} \Phi^{-1} ( 1 - 0.05/(N \vee p_{\theta} \log N)) $ and $\widehat{\Psi}_{\theta}=\text{diag}(\widehat{l}_1, \widehat{l}_2, \dots, \widehat{l}_{p_{\theta}})$ be a diagonal matrix of data-dependent penalty loadings chosen as in Algorithm 6.1 in \cite{Program}. Define  $\widehat{\theta}$ as a solution to the following optimization problem:
\begin{align*}
\widehat{\theta} := \arg \min_{\theta \in \mathcal{R}^{p_{\theta}} }  \EN D_i (Y_i - p_{\mu}(Z_i) '\theta)^2 + \lambda \| \widehat{\Psi}_{\theta} \theta \|_1   
\end{align*}
 and a first-stage estimate of $\mu $ as    $$\widehat{\mu}(z) := p_{\mu}(z)'\widehat{\theta}.$$
\end{definition} 
\begin{definition}[Lasso Estimator of the Propensity Score]
Let  $\lambda = 1.1\sqrt{N} \Phi^{-1} ( 1 - 0.05/(N \vee p_{\delta} \log N)) $ and $\widehat{\Psi}_{\delta}=\text{diag}(\widetilde{l}_1, \widetilde{l}_2, \dots, \widetilde{l}_{p_{\delta}})$ be a diagonal matrix of data-dependent penalty loadings chosen as in Algorithm 6.1 in \cite{Program}. Let $\underbar{s}>0$ be a positive constant. Define $\widehat{\delta}$ as a solution to the following optimization problem:
\begin{align*}
\widehat{\delta} := \arg \min_{\delta \in \mathcal{R}^{p_{\delta}}} \EN [\log (1+ \exp(p_{s}(Z_i)'\delta)) - D_i p_{s}(Z_i)'\delta] + \lambda \| \widehat{\Psi}_{\delta} \delta \|_1   
\end{align*} and a first-stage estimate of $s_0$ as    $$\widehat{s}(z) :=  \max( \underbar{s}/2, \mcL(p_{s}(z)' \widehat{\delta})).$$
\end{definition}
\begin{lemma}[Sufficient Conditions for Assumption \ref{ass:fsrate}]
\label{lem:lowlevel}
Suppose Assumptions \ref{ass:ident} and \ref{ass:ident2} hold.  Then, the following statements hold.
(1) There exists $C_{\theta} <\infty$ be such that w.p. $1-o(1)$:
$\| p_{\mu}(Z)'(\widehat{\theta} -\theta_0) \|_{P_{N,2}} \leq  C_{\theta} \sqrt{\dfrac{s_{\theta}  \log p_{\theta} }{N}} $ and $\| \widehat{\theta} -\theta_0 \|_{1} \leq  C_{\theta}  \sqrt{\dfrac{s_{\theta}^2  \log p_{\theta} }{N}} $. There exists $C_{\delta} <\infty$ be such that w.p. $1-o(1)$:
$\| p_{s}(Z)'(\widehat{\delta} -\delta_0) \|_{P_{N,2}} \leq  C_{\delta} \sqrt{\dfrac{s_{\delta}  \log p_{\delta} }{N}} $ and $\| \widehat{\delta} -\delta_0 \|_{1} \leq  C_{\delta}  \sqrt{\dfrac{s_{\delta}^2  \log p_{\delta} }{N}} $. 

(2)  Define the nuisance realization sets $ M_N$ and $ S_N$ as:
\begin{align}
\label{eq:mn}
 M_N:= \bigg\{ \mu(z) = p_{\mu}(z) ' \theta: \theta \in \mathcal{R}^{p_{\theta}}: \quad &\| p_{\mu}(Z)'(\theta -\theta_0) \|_{P_{N,2}} \leq  C_{\theta}  \sqrt{\dfrac{s_{\theta}  \log p_{\theta} }{N}}, \\
 &\| \widehat{\theta} -\theta_0 \|_{1} \leq  C_{\theta} \sqrt{\dfrac{s_{\theta}^2  \log p_{\theta} }{N}}   \bigg\} \nonumber 
  \end{align} 
  \begin{align}
 \label{eq:sn}
 S_N:= \bigg\{ s(z) = \mcL(p_{s}(z)' \delta):  \delta \in \mathcal{R}^{p_{\delta}}: \quad   &\| p_{\delta}(Z)'(\delta -\delta_0) \|_{P_{N,2}} \leq  C_{\delta}  \sqrt{\dfrac{s_{\delta}  \log p_{\delta} }{N}}, \\&\| \widehat{\delta} -\delta_0 \|_{1} \leq  C_{\delta}'  \sqrt{\dfrac{s_{\delta}^2  \log p_{\delta} }{N}}    \bigg\}   \nonumber 
 \end{align} 
Then, w.p. $1-o(1)$, $\widehat{\mu}(\cdot) \in M_N$ and the set $M_N$ shrinks at rate $\textbf{m}_N:=\sqrt{\dfrac{s_{\theta}  \log p_{\theta} }{N}}$.  Then, w.p. $1-o(1)$, $\widehat{s}(\cdot) \in S_N$ and the set $S_N$ shrinks at rate $\textbf{s}_N:=\sqrt{\dfrac{s_{\delta}  \log p_{\delta} }{N}}$. 

(3) Suppose $\xi_d(\textbf{m}_N \vee \textbf{s}_N) = o(1)$  and the product of sparsity indices $s_{\theta} s_{\delta}$ grows sufficiently slow:
$$ \sqrt{N}\sqrt{d} \textbf{m}_N  \textbf{s}_N= \sqrt{d}  \sqrt{  \dfrac{   s_{\theta} s_{\delta} \log p_{\theta} \log p_{\delta} }{N}  } = o(1).$$  Then, Assumption \ref{ass:fsrate} holds.

\end{lemma}

\end{example}
\textbf{Proof of Lemma \ref{lem:lowlevel}}.

\textbf{ Step 1}.Define $\widetilde{Y} = DY$, $\widetilde{p}_{\mu} (Z) = Dp_{\mu}(Z)$, $\widetilde{r}_{\mu}(Z) = Dr_{\mu}(Z)$, and $\widetilde{\epsilon} = D[Y - \mu_0(Z)]$. Here we verify the conditions of Theorem 6.1 from \cite{Program} for the model $ \widetilde{Y}, \widetilde{p}_{\mu} (Z), \widetilde{r}_{\mu}(Z), \widetilde{\epsilon}$.  Let us show that the original coefficient $\theta_0$, defined in Equation \eqref{eq:mu0},
satisfies $$\widetilde{Y} = \widetilde{p}_{\mu} (Z)'\theta_0 + \widetilde{r}_{\mu}(Z)+ \widetilde{\epsilon}, \quad \Ep [\widetilde{\epsilon}| \widetilde{D}] = 0 $$ Indeed, \begin{align*}
\Ep [\widetilde{Y}| D,Z] &= \Ep [DY | D, Z] = D \Ep [Y| Z,D] = D [\mu_0(Z)] = Dp_{\mu}(Z)' \theta_0 + D r(Z) \\
&= \widetilde{p}_{\mu} (Z)' \theta_0 + \widetilde{r}_{\mu}(Z). \\
\Ep [\widetilde{\epsilon} | D ] &= \Ep[ \Ep  \big[ [\widetilde{Y} -\widetilde{p}_{\mu} (Z)' \theta_0 - \widetilde{r}_{\mu}(Z) ] | D,Z  \big] |D ] = 0
\end{align*}
Recognize that Assumption \ref{ass:ident} implies that an analog of Assumption \ref{ass:ident} holds for $\widetilde{p}_{\mu} (Z), \widetilde{Y}$.  Assumption \ref{ass:ident} (a) directly assumes bounded Restricted Sparse Eigenavalues for Observed Regressors $\widetilde{p}_{\mu}(Z) = {p}_{\mu}(Z) D$. Assumption \ref{ass:ident} (b) is satisfied with $\Ep \widetilde{p}_{\mu,j}(Z)^2  = \Ep D  {p}_{\mu,j}^2(Z) \geq \underbar{s} \Ep {p}_{\mu,j}^2(Z) =: c'$ where $c':= c \underbar{s}$ is the new lower bound on the moments of $\Ep \widetilde{p}_{\mu,j}(Z)$ for observed regressors. Assumption \ref{ass:ident} (c) is satified with the $\Ep \widetilde{r}_{\mu}(Z)^2 \leq \Ep r_{\mu}(Z)^2 \leq C s_{\theta} \log (p_{\theta} \vee N)/N $. By Theorem 6.1 of \cite{Program}, w.p. $1-o(1)$, $p_{\mu}(\cdot)'\widehat{\theta} \in M_N$. Under Assumption \ref{ass:ident2}, Theorem 6.2 of \cite{Program} implies that  w.p. $1-o(1)$,  $\mcL(p_{s}(\cdot)'\widehat{\delta}) \in S_N$.

\textbf{ Step 2}.Let $C'_{\theta}= (C_{\theta}^2+\sqrt{2}B_{\theta})^{1/2}$.  For any $\theta \in \mathcal{R}^{p_{\theta}}$ such that $p_{\mu}(\cdot)'\theta \in M_N$ the following inequality holds:
\begin{align*}
&\E_{Z} (p_{\mu}(Z)'(\theta - \theta_0))^2 \leq \| \theta - \theta_0 \|_{P_{N,2}}^2 \\
&+ \big| (\theta - \theta_0)^\top [\E_N p_{\mu}(Z_i) p_{\mu}^\top (Z_i) - \E  p_{\mu}(Z_i) p_{\mu}^\top (Z_i)](\theta - \theta_0)\big| \\
 &\leq \| \theta - \theta_0 \|_{P_{N,2}}^2 +  \| \theta - \theta_0 \|_{1}^2 \max_{1 \leq i, j \leq p_{\theta}} \bigg|\E_N p_{i,\mu}(Z_i) p_{j,\mu}^\top (Z_i) - \E  p_{i,\mu}(Z_i) p_{j,\mu}^\top (Z_i) \bigg| \\
 &\lesssim^{i}_{P}  C_{\theta}^2 \textbf{m}_N^2 + B_{\theta} \sqrt{ \dfrac{s^2_{\theta} \log p_{\theta} }{N}  \dfrac{2 \log p_{\theta} }{N}} \lesssim_{P} (C_{\theta}')^2 \textbf{m}_N^2, 
\end{align*}
where $i$ follows by McDiarmid maximal inequality for bounded random variables. Therefore, the nuisance realization set $M_N$ shrinks at rate $\textbf{m}_N$. 

\textbf{ Step 3}. Observe that $ (\E_{Z} ( \mcL( p_{s}(Z)'\delta) -  \mcL( p_{s}(Z)'\delta_0))^2)^{1/2}  \leq \sup_{t \in \mathcal{R}} \mcL'(t) (\E_{Z} (p_{s}(Z)' (\delta -  \delta_0))^2)^{1/2} \leq \frac{1}{4} (\E_{Z} (p_{s}(Z)' (\delta -  \delta_0))^2)^{1/2}$. By the argument similar to Step 2, the nuisance realization set $S_N$ shrinks at rate $\textbf{s}_N$.

\end{document}